\documentclass[acmsmall,screen,nonacm]{acmart}

\usepackage{mathtools}
\usepackage{cleveref}
\usepackage{xspace}
\usepackage{tcolorbox} 
\usepackage{ebproof} 
\usepackage[inference]{semantic}
\usepackage{makecell}
\usepackage{multirow}
\usepackage{enumitem}
\usepackage[frozencache]{minted}
\usepackage{stmaryrd}
\usepackage{bbm}

\setminted[haskell]{
  fontsize=\small,
escapeinside=!!}
\setminted[coq]{
  fontsize=\small,
escapeinside=!!}

\newcommand {\hask}{\mintinline{haskell}}
\newcommand {\rocq}{\mintinline{coq}}

\tcbset{colback=white,boxrule=0.3mm,arc=3mm}

\tcbuselibrary{minted, skins, breakable}

\tcbset{
  codebox/.style={
    enhanced,
    frame hidden,
    borderline north={0.4pt}{0pt}{black!60},
    borderline south={0.4pt}{0pt}{black!60},
    colback=white,
    sharp corners,
    boxrule=0pt,
    top=2pt,
    bottom=0pt,
    left=0pt,
    right=0pt,
    overlay={
      \node[anchor=north east, inner sep=2.2pt, font=\footnotesize]
        at (frame.north east) {\textsc{#1}};
    },
  }
}

\newtcblisting{haskellcode}{
  listing engine=minted,
  minted language=haskell,
  listing only,
  codebox={Haskell},
}

\newtcblisting{lrcode}{
  listing engine=minted,
  minted language=haskell,
  listing only,
  codebox={$\er$},
}

\newtcblisting{rocqcode}{
  listing engine=minted,
  minted language=coq,
  listing only,
  codebox={Rocq},
}

\newcommand \seq[1] {\overline{#1}} 
\DeclareMathOperator \dom {dom} 
\DeclareMathOperator \codom {codom} 
\newcommand \fnreplace[3] {#1[#2 \leftarrow #3]}

\newcommand \lava {Lava\xspace}
\newcommand \lh {Liquid Haskell\xspace}
\newcommand \stainless {Stainless\xspace}
\newcommand \fstar {F*\xspace}
\newcommand \flux {Flux\xspace}
\newcommand \pl[1] {#1}
\newcommand \coq {\pl{Rocq}\xspace}
\newcommand \haskell {\pl{Haskell}\xspace}

\newcommand \toolname \lava

\newcommand\appendixref[1]{\iflongversion \cref{#1}\else the appendix\fi}
\newcommand\appendixrefp[1]{\iflongversion (\cref{#1})\else (see the appendix)\fi}

\newcommand\projleft{\lfloor \! \! |}
\newcommand\projright{ | \!\! \rfloor }

\newcommand\eg{\textit{e.g.,}\xspace}
\newcommand \ie {\textit{i.e.,}\xspace}

\newcommand \keyw[1] {\textbf{\texttt{#1}}} 
\newcommand \var[1] {#1} 
\newcommand \gramSym[1] {#1} 
\newcommand \gramTitle[1] {\textbf{\textit{#1}}} 
\newcommand \gramopt[1] {(#1)^?} 
\newcommand \gramplus[1] {(#1)^+} 

\newcommand \sloc {LoC\xspace}
\newcommand \hsLines {4}
\newcommand \rocqLines {6.5}
\newcommand \blowup {14.6}

\newcommand\formatBenchmark[1]{\texttt{#1}\xspace}
\newcommand\SafeIndexing{\formatBenchmark{Overview}}
\newcommand\PeanoNats{\formatBenchmark{PeanoNats}}
\newcommand\RBinsToBins{\formatBenchmark{RBinsToBins}}
\newcommand\HigherOrder{\formatBenchmark{HigherOrder}}
\newcommand\Foldr{\formatBenchmark{FoldrUniversal}}
\newcommand\SoftFound{\formatBenchmark{SoftwareFnd}}
\newcommand\ClassLaws{\formatBenchmark{ClassLaws}}
\newcommand\Lists{\formatBenchmark{List}}

\newcommand\verifTimeSafeIndexing {1.17}
\newcommand\verifTimeSoftFound {0.04}

\newcommand\embed[1]{\ensuremath{[\![#1]\!]}}
\newcommand\embedleft{[\![}
\newcommand\embedright{]\!]}
\newcommand\tInt{\texttt{Int}}
\newcommand \Int {\keyw{Int}}

\newcommand \Bool {\keyw{Bool}}

\newcommand \Unit {\keyw{()}}
\newcommand \unit {\keyw{()}}

\newcommand \true {\keyw{True}}
\newcommand \false {\keyw{False}}
\newcommand \refT[3] {\{#1\mathord{:} #2 \mid #3\}} 
\newcommand \refTT[1] {\{#1\}} 
\newcommand \bind[2]{#1\mathord{:}#2}
\newcommand \coqbind[2]{#1{:}#2}
\newcommand \ap {\,} 
\newcommand \ite[3] {\keyw{if}\;#1\;\keyw{then}\;#2\;\keyw{else}\;#3}

\newcommand \lets[4] {\keyw{let}\; (#1 \coloneqq #2)_{#3}\;\keyw{in}\;#4}
\newcommand \lequiv {\Leftrightarrow} 
\newcommand \case[2] {#1 \mapsto #2} 

\newcommand \innersubst[2] {#1 \backslash #2}
\newcommand \subst[2]{[\innersubst{#1}{#2}]} 

\DeclareMathOperator \fv {fv} 

\newcommand \varSet {\mathcal{V}}

\newcommand \stlc {\lambda_\to}

\newcommand \lctyp {\gramSym A} 
\newcommand \lcTC {{\var{TC}}} 
\newcommand \lcdt[2] {\langle #1, #2 \rangle} 
\newcommand \lcexp {\gramSym e} 
\newcommand \lccase[2] {\keyw{case } #1 \mathrel{\keyw{of}} #2} 
\newcommand \lclam[2] {\lambda #1.\,#2} 
\newcommand \lclet[3] {\keyw{let } #1 \coloneqq #2 \keyw{ in } #3}
\newcommand \lcval {\gramSym v} 
\newcommand \lcConstr {\var C} 
\newcommand \lccst {\var c} 
\newcommand \lclit \lhlit 
\newcommand \lcop \lhop 
\newcommand \lcbop \lhbop 
\newcommand \lcplus \lhplus 
\newcommand \lcminus \lhminus 
\newcommand \lctimes \lhtimes 
\newcommand \lcdiv \lhdiv 
\newcommand \lcmodulo \lhmodulo 
\newcommand \lcnegb \lhnegb 
\newcommand \lceqb \lheqb 
\newcommand \lceqz \lheqz 
\newcommand \lcandb \lhandb 
\newcommand \lcorb \lhorb 
\newcommand \lcleq \lhleq 

\newcommand \er {\lambda_r} 

\newcommand \lhdecl {\gramSym D} 
\newcommand \lhdecls {{\seq\lhdecl}}
\newcommand \lhdata[2] {\keyw{data } #1 \coloneqq #2} 
\newcommand \lhTC {{\var{TC}}} 
\newcommand \lhConstr {\var C} 
\newcommand \lhdef[3] {\keyw{def } #1:: #2 \coloneqq #3} 
\newcommand \lhrefl[3] {\keyw{reflect def } #1:: #2 \coloneqq #3} 
\newcommand \lhdefrefl[3] {\gramopt{\keyw{reflect}} \keyw{ def } #1:: #2 \coloneqq #3} 

\newcommand \lhref {\gramSym R} 
\newcommand \lharr {\gramSym A} 
\NewDocumentCommand \arr{O{x} O{\lhref_x} O{\lhref}} {\bind{#1}{#2} \to #3} 
\NewDocumentCommand \arrs{O{x_i} O{\lhref_i} O{i \leq n} O{\lhref}} {(\bind{#1}{#2})_{#3} \to #4} 
\NewDocumentCommand \arrsSeq{O{x} O{\lhref_x} O{\lhref}} {\seq{(\bind{#1}{#2})} \to #3} 
\newcommand \lhbuiltin {\gramSym B} 
\newcommand \lhtyp {\gramSym B} 
\newcommand \lhsmptyp {\gramSym A} 
\newcommand \lhdt[2] {\langle #1, #2 \rangle} 
\newcommand \lhexp {\gramSym e } 
\newcommand \lhcase[2] {\keyw{case}~{#1} \mathrel{\keyw{of}} #2} 

\newcommand \dead {\keyw{unreachable}} 
\newcommand \lhtm {\gramSym r} 
\newcommand \lhval {\gramSym v} 
\newcommand \lhcst {\var c} 
\newcommand \lhlit {\var \ell} 
\newcommand \lhop {\odot} 
\newcommand \lhbop \lhop 
\newcommand \lhplus \cqplus 
\newcommand \lhminus \cqminus 
\newcommand \lhtimes \cqtimes 
\newcommand \lhdiv \cqdiv 
\newcommand \lhmodulo \cqmodulo 
\newcommand \lhnegb \cqnegb 
\newcommand \lheqb {\mathbin{=}} 
\newcommand \lheqz \cqeqz 
\newcommand \lhandb {\land} 
\newcommand \lhorb \cqorb 
\newcommand \lhleq \cqleq 
\newcommand \lhimp \Rightarrow 
\newcommand \lhiff \Leftrightarrow 

\newcommand \lebZRel{\relnmOp\leq}
\newcommand \ltbZRel{\relnmOp<}
\newcommand \addZRel{\relnmOp+}
\newcommand \lhleqZnm {\texttt{leqZ}}

\newcommand \lhqmark[2] {#1\mathbin{?}#2}
\newcommand \ann[1] {#1^e} 
\newcommand \lhinj[2] {\keyw{inj}(#1, #2)} 
\newcommand \lhsub[3] {\keyw{sub}(#1, #2, #3)} 
\newcommand \lhproj[1] {\keyw{proj}(#1)}
\newcommand \varar[2] {#1_{/(#2,\_)}} 
\newcommand \varartop[2] {#1_{/(#2,\topSymbol)}} 
\newcommand \vararloc[2] {#1_{/(#2,\locSymbol)}} 
\newcommand \vararlg[3] {#1_{/(#2,#3)}}
\newcommand \varlg[2] {{#1}_{/{#2}}}
\newcommand \varloc[1] {\varlg{#1}{\locSymbol}}
\newcommand \vartop[1] {\varlg{#1}{\topSymbol}}
\newcommand \varrec[1] {\varlg{#1}{\recSymbol}}
\newcommand \varargrammar {x_{\color{gray}{/(n,s)}}}
\newcommand \annlhqmark[3] {#1\mathbin{?}({#2}\mathbin{\mathtt{proves}}{#3})}
\newcommand \annlhqmarkgrammar {\lhexp \mathbin{?} {\color{gray}{(}} \lhtm\ {\color{gray}{\mathtt{proves}\ \lhtm)}}}

\newcommand \lhlet[4] {\keyw{let}\; {#1 \mathrel{::} #2} \coloneqq #3 \;\keyw{in}\; #4} 

\newcommand \guard {\mathcal{P}}
\NewDocumentCommand \eqguard{O{\sigma_\params} O{\sigma_\brSub} m} {\mathcal{P}(#1,#2,#3)}
\newcommand \eqguardName {\mathcal{P}}

\DeclareMathOperator \mkGraph {paths} 
\newcommand \projType[1] {\operatorname{unref}(#1)}
\DeclareMathOperator \noproj {noProj} 

\newcommand \ar {\operatorname{arity}} 
\newcommand \frApp {\texttt{RV}} 

\DeclareMathOperator \repl {repl}

\newcommand \cqTC {\var {TC}} 
\newcommand \cqConstr {\var C} 

\newcommand \Prop {\texttt{Prop}}

\newcommand\rocqFond[1]{\mathbbm{#1}} 
\newcommand \cqtyp {\gramSym \rocqFond{U}} 
\newcommand \cqref {\rocqFond{S}} 

\newcommand \cqBool {\texttt{bool}}
\newcommand \cqInt {\texttt{Z}}
\newcommand \cqUnit {\texttt{Unit}}

\newcommand \cqexp {\gramSym \rocqFond{e}} 
\newcommand \cqlop {\lhop} 
\newcommand \cqlet[3] {\texttt{let } #1 \coloneqq #2 \texttt{ in } #3}
\newcommand \cqmatch[2] {\texttt{match } #1 \texttt{ with } #2 \texttt{ end}}
\newcommand \cqcase[2] {#1 \Rightarrow #2} 
\newcommand \proj[1] {\projleft #1 \projright}
\newcommand \projTwo[1] {\texttt{proj2\_sig}(#1)}
\newcommand \cqlamuntyped[2] {\lambda #1,\,#2} 
\newcommand \cqrec[1] {\{|#1|\}}

\newcommand \cqrecuClass[2] {\cqrec{\relnmUClass := #1; \functrelUClass := #2}}
\newcommand \cqcst \lhcst
\newcommand \ltacOracle {\tac{ltac:(solver)}}
\newcommand \ltacTac[1] {\tac{ltac:(#1)}}
\newcommand \cqtrue {\keyw{true}}
\newcommand \cqfalse {\keyw{false}}
\newcommand \cqunit {\keyw{tt}}
\newcommand \cqlit \lhlit
\newcommand \cqtm {\gramSym \rocqFond{p}}
\newcommand \cqprop {\gramSym \rocqFond{Q}} 
\newcommand \IsTrue {\keyw{is\_true}} 

\newcommand \True {\keyw{True}} 
\newcommand \False {\keyw{False}} 
\newcommand \tac[1] {\texttt{#1}} 
\newcommand \ex[3] {\mathrel{\texttt{exist}}\ #1\ #2\ #3} 
\newcommand \oracle {\tac{oracle}\xspace} 
\newcommand \solver {\tac{solver}\xspace} 

\newcommand{\easyWit}{\tac{quick\_wff\_wit}\xspace}
\newcommand{\fastDone}{\tac{fast\_done}\xspace}
\newcommand{\solveZero}{\tac{fast\_solve}\xspace}
\newcommand{\idTac}{\tac{idtac}\xspace}

\newcommand{\quickSimpl}{\tac{quick\_simpl}\xspace}
\newcommand{\quickSolve}{\tac{quick\_solve}\xspace}

\newcommand{\cleanup}{\tac{cleanup}\xspace}
\newcommand{\liaPreprocessor}{\tac{lia\_preprocessor}\xspace}
\newcommand{\preprocessor}{\tac{preprocessor}\xspace}

\newcommand{\unpack}{\tac{unpack\_all}\xspace}
\newcommand{\axTerms}{\tac{axiomatize\_terms}\xspace}
\newcommand{\axIhAppls}{\tac{axiomatize\_ih\_specializations}\xspace}
\newcommand{\unifyVars}{\tac{unify\_vars}\xspace}
\newcommand{\specHyps}{\tac{specialize\_hyps}\xspace}
\newcommand{\invAxs}{\tac{invert\_axiomatizations}\xspace}
\newcommand{\instHyps}{\tac{instantiate\_hyps}\xspace}

\newcommand{\instExistGoal}{\tac{instExistGoal}\xspace}
\newcommand{\saturateCtx}{\tac{saturate\_context}\xspace}
\newcommand{\quickSimplCleanup}{\tac{quick\_cleanup}\xspace}
\newcommand{\shapeBased}{\tac{shape\_based}\xspace}
\newcommand{\fEqualInd}{\tac{f\_equal\_ind}\xspace}
\newcommand{\simplLoop}{\tac{solving\_loop}\xspace}
\newcommand{\finish}{\quickSolve}

\newcommand \exfalso {\tac{exfalso}}

\newcommand \despattern {\gramSym{des}} 

\newcommand \subCastSymbol {\hookleftarrow}
\newcommand \subCastExplicit[3] {#1 \subCastSymbol_{#3} #2} 
\newcommand \subPack{\texttt{subPack}}
\newcommand \subRef{\texttt{subRef}}

\newcommand \cqplus + 
\newcommand \cqminus - 
\newcommand \cqtimes \times 
\newcommand \cqdiv / 
\newcommand \cqmodulo \% 
\newcommand \cqnegb \neg 
\newcommand \cqnegp \neg 
\newcommand \cqeqb {\mathbin{==?}} 
\newcommand \cqeqz {\mathbin{==?}} 
\newcommand \cqandb \land 
\newcommand \cqandp \land 
\newcommand \cqorb \lor 
\newcommand \cqleq \leq 
\newcommand \cqeqp {==} 

\newcommand \sub {\preceq} 
\newcommand \trule[1] {\textsc{#1}\xspace} 

\newcommand \typingOutput {\rightsquigarrow} 
\newcommand \hwfj[3] {#1 \vdash #2 \typingOutput #3} 

\newcommand \hdecl[3] {#1 \vdash #2 \typingOutput #3}
\newcommand \hsynj[4] {#1 \vdash #2 \Rightarrow #3 \typingOutput #4}

\newcommand \hchkj[4] {#1 \vdash #2 \Leftarrow #3 \typingOutput #4}
\newcommand \hsubj[3] {#1 \vdash #2 \sub #3} 
\newcommand \csttype[1] {\operatorname{ty}(#1)} 
\newcommand \hsmpj[4] {#1 \vdash_{\text{SIMPLE}} #2 \mathbin{:} #3 \typingOutput #4}

\newcommand \lctsname {\mathcal {LC}} 
\newcommand \lctyj[3] {#1 \vdash_{\lctsname} #2: #3} 

\DeclareMathOperator \trLhCqName {RtoR} 
\NewDocumentCommand \trLhCq{O{} m} {\trLhCqName_{#1}(#2)} 
\DeclareMathOperator \utransName {RtoU}
\newcommand \utrans[1] {\utransName(#1)} 

\DeclareMathOperator \altInGraphRel {pathInd}
\DeclareMathOperator \hypFromAmName {Hyp}
\newcommand \hypFromAm[1] {\hypFromAmName(#1)}
\newcommand \hypQ {\mathcal Q}
\NewDocumentCommand \hypFromAmQ {O{\hypQ} m} {\hypFromAmName_{#1}(#2)}
\NewDocumentCommand \hypTerm {m O{\typEnv}} {\operatorname{Hyp'}(#2;#1)}

\DeclareMathOperator \trLhdeclCqName {DtoR} 
\newcommand \trLhdeclCq[1] {\trLhdeclCqName(#1)} 
\newcommand \params {{xs}}

\newcommand \branch {p}
\newcommand \brSub {\branch}

\newcommand\getRel{\texttt{fun\_to\_rel}}

\DeclareMathOperator \exprToTacName {EtoTac}

\newcommand \trLhexpCq[1] {\exprToTacName(#1)} 
\newcommand \trLhtmCq[1] {\trLhCqName(#1)} 

\DeclareMathOperator \typesToRefName {TtoR}
\NewDocumentCommand \oldtrLhrefCq{m O{\typEnv}} {\typesToRefName(#2; #1)}
\newcommand \trLhrefCq[1] {\typesToRefName(#1)}
\newcommand \trLhrefCqEmp[1] {\typesToRefName(#1;\emptyset)}

\newcommand \trBuiltin[1] {\operatorname{Prim}(#1)} 
\newcommand \trLit[1] {\operatorname{Prim}(#1)} 
\newcommand \utrLit[1] {\trLit{#1}} 
\NewDocumentCommand \oldutrLhtmCq{m O{\typEnv}} {\utrans{#2;#1}} 
\newcommand \utrLhtmCqName \utransName
\newcommand \utrLhtmCq[1] {\utransName(#1)} 

\DeclareMathOperator \typesToUnRefName {TtoU}
\newcommand \utrLhrefCq[1] {\typesToUnRefName(#1)} 
\DeclareMathOperator \refToPropName {RtoP}

\NewDocumentCommand \oldutrLhtmCqprop{m O{\typEnv}} {\refToPropName(#2; #1)} 
\newcommand \utrLhtmCqprop[1] {\refToPropName(#1)}
\NewDocumentCommand \trLhrefCqclass{m O{\typEnv}} {\operatorname{pack}(#2;#1)}

\newcommand \refnm[1] {#1} 
\newcommand \unrefnm[1] {#1_u} 
\newcommand \relnm[1] {#1\texttt{\_rel}} 
\newcommand \relnmOp[1] {#1_{\texttt{rel}}} 

\newcommand \ArgListT {\texttt{ArgListT}\xspace} 
\newcommand \ArgList {\texttt{ArgList}\xspace} 
\newcommand \UArgListT {\texttt{UArgListT}\xspace} 
\newcommand \UArgList {\texttt{UArgList}\xspace} 
\newcommand \projectsArgListT {\texttt{projectsArgListT}\xspace} 
\newcommand \prArgList {\texttt{prArgList}\xspace}

\newcommand \flattenP {\texttt{flattenP}\xspace}

\newcommand \fClass {\texttt{Pack}} 
\newcommand \fuClass {\texttt{uPack}} 
\newcommand \gProj {\texttt{Projection}} 
\newcommand \gSub {\texttt{Subsumption}} 
\newcommand \gSubVal {\texttt{sub\_res}} 
\newcommand \fInst[1] {#1\texttt{\_pack}} 
\newcommand \fuInst[1] {#1\texttt{\_upack}} 
\newcommand \relnmClass {\texttt{rel}} 
\newcommand \relnmUClass {\texttt{rel\_u}} 
\newcommand \relnmInst[1] {#1.(\relnmClass)} 
\newcommand \relnmUInst[1] {#1.(\relnmUClass)} 
\newcommand \getPackRel[1] {\texttt{getPackRel}\ap #1}
\newcommand \getUPackRel[1] {\texttt{getUPackRel}\ap #1}
\newcommand \defnmClass {\texttt{def}} 
\newcommand \defnmInst[1] {#1.(\defnmClass)} 

\newcommand \unrefmk[1] {#1_{rel}\_mk}
\newcommand \wffnm[1] {\operatorname{wf}_{#1}} 
\newcommand \unrefproof[1]{#1\_\relnm{#1}} 
\newcommand \unrefproofClass{\texttt{f\_f\_rel}} 

\newcommand \functrel[1]{#1\texttt{\_rel\_funct}} 
\newcommand \functrelClass{\texttt{funct}} 
\newcommand \functrelUClass{\texttt{funct\_u}} 
\newcommand \functrelInst[1]{\fInst{#1}.(\functrelClass)} 
\newcommand \uClassToClass[1]{\texttt{projPack}(#1)}

\newcommand \wit[1] {\hat{#1}} 
\newcommand \ihasm[1] {\texttt{IH}{#1}} 

\newcommand \trLhLc[1] {\lfloor {#1} \rfloor_{H\lambda}}

\newcommand \trCqLc[1] {\lfloor {#1} \rfloor_{Q\lambda}}

\NewDocumentCommand \lhtpint{m O{\rho} O{\lhdecls}} {\llbracket #1 \rrbracket_{#2}^{#3}}
\NewDocumentCommand \cqtpint{m O{\rho} O{\lhdecls}} {\llparenthesis #1 \rrparenthesis_{#2}^{#3}}

\newcommand \rew[1] {\to_{#1}} 
\newcommand \rewn[1] {\rew{#1}^*} 

\NewDocumentCommand \ssred{O{\seq\lhdecl} m m} {#1 \vdash #2 \rew{\er} #3}
\NewDocumentCommand \ssredn{O{\seq\lhdecl} m m} {#1 \vdash #2 \rewn{\er} #3}

\newcommand \stlcredName {\lambda} 
\NewDocumentCommand \stlcred{O{\lhdecls} m m} {#1 \vdash #2 \rew\stlcredName #3}
\NewDocumentCommand \stlcredn{O{\lhdecls} m m} {#1 \vdash #2 \rewn\stlcredName #3}

\newcommand\typEnv{\Gamma} 
\newcommand\topSymbol{{G}}
\newcommand\locSymbol{{L}}
\newcommand\recSymbol{{Y}}
\DeclareMathOperator \caseify {caseify}
\DeclareMathOperator \caseifyaux {caseify\_aux}

\newcommand\SVar{\trule{S-Var}}
\newcommand\SVarLoc{\trule{S-VarL}}
\newcommand\SData{\trule{S-Data}}
\newcommand\SLit{\trule{S-Lit}}

\newcommand\SNeg{\trule{S-Neg}}
\newcommand\SEq{\trule{S-Eq}}
\newcommand\SBin{\trule{S-Bin}}
\newcommand\SApp{\trule{S-App}}

\newcommand\CLet{\trule{C-Let}}

\newcommand\CSyn{\trule{C-Syn}}

\newcommand\CCase{\trule{C-Case}}
\newcommand\CHint{\trule{C-Hint}}

\newcommand\TSRef{\trule{S-Base}}
\newcommand\TSRefAx{\trule{S-Base-Ax}}
\newcommand\TSArr{\trule{S-Arr}}

\newcommand\WFTFun{\trule{WF-TFun}}
\newcommand\WFTBase{\trule{WF-TRef}}

\newcommand\WFDDef{\trule{WF-DDef}}

\newcommand\WFDTC{\trule{WF-DTC}}

\newcommand \lambdaert {\lambda_{\text{ert}}} 

\usepackage{thmtools,thm-restate}

\crefname{section}{\S}{\S}

\newif\iflongversion
\longversiontrue

\title{Lava: Foundations for Liquid Haskell via Translation to Rocq}

\author{Colin Rothgang}
\affiliation{%
  \institution{IMDEA Software Institute}
  \city{Madrid}
  \country{Spain}
}

\author{Loïc Peyrot}
\orcid{0000-0002-1398-7460}
\affiliation{%
  \institution{IMDEA Software Institute}
  \city{Madrid}
  \country{Spain}
}

\author{Yannick Forster}
\orcid{0000-0002-8676-9819}
\affiliation{%
  \institution{Inria}
  \city{Paris}
  \country{France}
}

\author{Michael Greenberg}
\affiliation{%
  \institution{Stevens Institute of Technology}
  \country{USA}
}

\author{Niki Vazou}
\orcid{0000-0003-0732-5476}
\affiliation{%
  \institution{IMDEA Software Institute}
  \city{Madrid}
  \country{Spain}
}

\begin{document}
\begin{abstract}
  Liquid Haskell is an extension of the Glasgow Haskell compiler (GHC) that adds support
  for refinement types: types restricted by logical predicates.
  Specifications are automatically checked by an SMT solver, thus offering
  lightweight verification directly within the language.
  However, the necessarily large trusted codebase behind Liquid Haskell
  (from GHC to the SMT solver to the refinement typechecker itself)
  raises the question of trust in the verification and the necessity for
  a formalized semantics of Liquid Haskell linked to its implementation.
  Grounding Liquid Haskell in theory is complicated by the interaction between two semantic
  levels in the language: a computational one and a logical one (for refinement predicates).
  While both are well understood separately, their interaction is subtle and
  has only been partially described.

  We give a semantics to Liquid Haskell by translation
  to Rocq. Both the computational and logical parts are translated to
  the same language, with a well-understood semantics; our translation makes the discrepancy between the two semantic levels explicit.
  The translation is implemented in Haskell and comes with a Rocq library
  that includes powerful proof automation standing in for the SMT solver.
  Thus, on the practical level, our tool Lava serves as a certificate checker
  for the supported subset of Liquid Haskell.
\end{abstract}

\maketitle

\section{Introduction}\label{sec:intro}
Refinement types checkers, such as \lh~\cite{vazou14},
\stainless~\cite{hamza19}, \fstar~\cite{fstar},
and \flux~\cite{flux},
rely on SMT solvers to reason about program properties.
They work on top of existing, industrial programming languages,
allowing programmers to work in familiar environments
while enjoying
a high degree of automation in program verification.
However, practical refinement types are disconnected from their theoretical foundations.
Even though program behaviors and SMT logics are very well understood separately,
the connection between them has not been made explicit.

In refinement types,
SMT validity checking is used to decide subtyping.
To do so, a connection between program behaviors and SMT logics
needs to be made.
As an example, consider an expression $e$ with a type that ensures it is an integer greater than $5$,
\ie $e:\refT {x} \tInt {5 < x}$.
Since  $5 < x$ implies $0 < x$, for all integers $x$,
$e$ can also be used at type $\refT {x} \tInt {0 < x}$.
This reasoning is captured by the following two rules:
\[
\inference{
    \Gamma \vdash e : \tau_l \quad
    \Gamma \vdash \tau_l \sub \tau_r
  }{
    \Gamma \vdash e : \tau_r
  }[\trule{T-Sub}]
\qquad\qquad
  \inference{
    \text{SMT-Valid }(\embed{\Gamma} \Rightarrow \embed{\lhtm_l} \Rightarrow \embed{\lhtm_r})
  }{
    \hsubj \typEnv {\refT {x} \lhtyp {\lhtm_l}} {\refT {x} \lhtyp {\lhtm_r}}
  }[\TSRef]
\]

The left rule, \trule{T-Sub}, allows for \emph{implicit subtyping}:
an expression $e$ of type $\tau_l$ can be given the type $\tau_r$, if $\tau_l$ is a subtype of $\tau_r$.
To decide subtyping, refinement type systems reduce it to SMT validity checking using
a set of subtyping rules, that ultimately boil down to
the rule \TSRef on the right.
This rule states that a refinement type $\{x:\lhtyp \mid \lhtm_l\}$ is a subtype of $\{x:\lhtyp \mid \lhtm_r\}$
if the logical implication $\lhtm_l \Rightarrow \lhtm_r$ is valid in the SMT logic
and under the assumptions in the type environment $\Gamma$.

The critical part is: program elements are translated to the SMT logic via the embedding function $\embed{\cdot}$.
In cases where both program elements and logics are the same (\eg integers with arithmetic operations),
the embedding function is trivial.
For example, when relating $\refT {x} \tInt {5 < x}$ and $\refT {x} \tInt {0 < x}$.
However, in general,
defining the embedding function is intricate,
and the details of its definition have not been formally accounted for.

Formalization of refinement type systems
have provided partial descriptions
of only the critical parts  of the embedding function
\cite{vazou14,hamza19,fstar,flux}.
Accepting that it is impractical to formally specify
the actual embeddings of real refinement type checker implementations,
we take a different approach to establish solid foundations for SMT-based verification:
\emph{translation to foundational type systems}.

We present \toolname, a translation from the refinement type system \lh to \coq.
\toolname connects the two worlds of programs and SMT logics by modeling
both in the formal system of \coq and
explicitly embedding program elements into logical ones.
Thus, our translation
offers both $a)$ deeper insights into the foundations of refinement type checking and
$b)$ an automated certificate checker for \lh verified programs.
Our automated checker is able to automatically verify sizable examples,
but crucially allows the user to manually intervene whenever automation fails. 

\subsection{Key Challenge of the Translation: Coexistence of Refined and Unrefined Semantics}\label{sec:refUnrefSemantics}

The main difficulty in our translation
is that \lh has multiple semantic levels that coexist in programs and serve different purposes.
Consider for example the following snippet
in which a binder $x$, which is specified via a refinement type to be greater than 5,
is used to divide $42$:
\begin{haskellcode}
let x :: {x:Int | 5 < x} = 8 in div 42 x -- div :: Int !$\to$! {v:Int | 0 < v} !$\to$! Int
\end{haskellcode}
Refinement types restrict the domain of a type, here \hask{Int}, with a logical predicate.
In \hask{div} and below, we write \hask{Int} for the trivially refined \hask!{v:Int | True}!.

In this snippet, the binder $x$ has three different semantics.
When the snippet is seen as a \haskell program,
the binder $x$ is a \textit{a)} \emph{program integer} that can be evaluated to $8$.
But \lh also attaches to $x$ the information that it is greater than $5$.
Thus, in type theoretical terms,
$x$ is a \textit{b)} \emph{dependent pair}
storing the evidence that the integer attached to $x$ is greater than $5$.
Finally, at the call site of the division,
\lh will ensure that the arguments do not violate the specification of the function.
In that case, since the second argument of division is specified to be greater than $0$,
it will create the verification condition
$\forall x, 5 < x \Rightarrow 0 < x$ and check it using an SMT solver.
When passed to the SMT solver,
the binder $x$ is a \textit{c)} \emph{logical integer} that admits arithmetic operations.

Two parallel worlds coexist:
\begin{itemize}
    \item the \emph{refined world} of the \lh type checker, where $x$ is the \textit{b)} dependent pair and
    \item the \emph{unrefined world} where $x$ is a \textit{a)} program integer or \textit{c)} logical integer.
\end{itemize}
These two semantics are crucial for how refinement types work in \lh.
The \lh type checker uses \emph{refined definitions} in the program to generate the verification conditions (VCs),
while the SMT solver uses \emph{unrefined definitions} in the refinements to solve the VCs.

\emph{
The key challenge of the translation is to account for these two
intertwined semantics, while establishing a correspondence between the refined
and unrefined translations of the same object.
}

\subsection{Refined and Unrefined Semantics in Functions}%
\label{sec:intro-functions}

The coexistence of refined and unrefined semantics becomes even more delicate
once we move from simple values to functions.
Functions in \lh simultaneously inhabit both worlds:
they are defined as Haskell programs with refinement type specifications,
but are also used in the logical layer to specify refinements.

As an example of this dual nature, consider a function that \hask{get}s an element from an integer list.

\begin{haskellcode}
{-@ get :: xs:IList !$\to$! i:{Int | 0 !$\leq$! i !$\land$! i < llen xs} !$\to$! Int @-}
get :: IList !$\to$! Int !$\to$! Int
get (Cons x xs') i = if i == 0 then x else get xs' (i-1)
\end{haskellcode}
\noindent The function \hask{get xs i} returns the \hask{i}-th element of the list \hask{xs}.
The \haskell definition is partial, so a refinement type specification (in \lh's special comment form \hask?{-@ ... @-}?)
ensures that the index argument is always in-bounds.
This is the main use-case for refinement types:
they guarantee the lack of runtime errors without the need to defensively handle
out-of-bounds cases.

\haskell definitions are on the \emph{refined} semantic level:
each call \hask{get xs i} generates a VC ensuring that the index \hask{i} is in-bounds.
Type refinements however are on the \emph{unrefined} level:
in specifications, \hask{get} is treated as an SMT logical function, thus, it cannot be partial.

As a concrete example we prove that getting the \hask{i}-th element of a list \hask{xs} is the same as getting the \hask{(i+1)}-th element of \hask{(Cons x xs)}:
\begin{haskellcode}
{-@ thm_get :: xs:IList !$\to$! x:Int !$\to$! i:{Int | 0 !$\leq$! i && i < llen xs}
	          !$\to$! { get xs i == get (Cons x xs) (i+1) } @-}
thm_get :: IList !$\to$! Int !$\to$! Int !$\to$! ()
thm_get _ _ _ = ()
\end{haskellcode}
\noindent
\lh accepts this definition by
 unfolding the right-hand side \hask{get (Cons x xs) (i+1)}
into \hask{get xs i}
and then generating a VC that equates the two.
Equality is checked by the SMT solver,
thus \hask{get} should be encoded
as a logical function
that \emph{models} the behavior of the program function \hask{get}.
The full definition of this logical function in \coq is given in \cref{sec:unref-graph}.

The dual nature of functions in \lh is challenging to capture in a \coq translation,
that has to define the refined \emph{and} unrefined versions of \hask{get} and then has to relate them properly.

\subsection{The Design of \lava}

We designed \lava, a translation from well-typed \lh programs to \coq that
respects the semantics of \lh by carefully accounting for
both the coexistence of refined and unrefined semantics in \lh programs
as well as implicit subtyping.

\paragraph{Modeling Refined and Unrefined Semantics}
To translate functions, like \texttt{get},
\lava generates two objects.
One object models the \emph{refined semantics} of the function,
as a dependently typed \coq function that includes the necessary proofs
to ensure that the function respects its refinement type specification.
The other object models the \emph{unrefined semantics} of the function,
as a \coq relation that captures the graph of the function.
Finally, \lava also generates and proves lemmas that relate the two semantics,
ensuring that the refined function behaves according to its unrefined graph relation.
To support higher-order functions,
a refined higher-order function needs to take as input not only the refined variant
of its input, but also the unrefined variant, a proof that the unrefined variant is a functional relation,
and a proof connecting the two.
We call this ``function plus graph plus two proofs'' a \emph{Pack}; packing these four components together
simplifies function subtyping.

\paragraph{Modeling Implicit Subtyping with Explicit Coercions}
When translating \lh programs, \lava systematically uses casts
to model implicit subtyping, \ie the typing rule \trule{T-Sub}
which lies at the core of many refinement type systems.
\lava translates program definitions
based on the type checking derivation tree and
translates implicit subtyping to explicit type coercions.
Those type coercions generate proof obligations.
Along with the translation itself, we also provide a custom automation tactic
that we call \solver designed to discharge these obligations.
The \solver tactic is able to automatically discharge the proof obligations
generated by our case studies (\cref{sec:eval}),
closely mirroring the level of automation provided by \lh's SMT backend.

\paragraph{Implementation}
We implemented \lava, a translation from programs verified in \lh to \coq.
The implementation consists of three main components.
First, the \haskell code implementing the \emph{“skeleton”} translation of terms to \coq.
This translation is the main subject of the paper and it is presented over
a core calculus for refinement types $\er$ that is a subset of \lh.
In reality \lava translates \lh terms to their refined and unrefined \coq counterparts
and translates implicit subtyping to explicit coercions.
Second, a \coq library with the definitions necessary to support our translation.
Third, a collection of \coq (Ltac) automation tactics that
prove obligations due to refinements (in particular the \solver tactic)
or prove structural lemmas, notably those relating unrefined and refined terms.

\paragraph{Contributions}
We contribute the design, metatheory, and implementation of \lava:
\begin{enumerate}[leftmargin=*]
  \item \textbf{A foundational semantics for monomorphic Liquid Haskell} (\cref{sec:translation}):
    The translation of verified \lh programs to \coq formalizes a semantics of \lh
    and crystallizes the need for the two intertwined semantics.
    Our design supports monomorphic programs, with
    user defined data types and higher-order functions.
    The intricacy of the semantics lies in making all fine-grained details
    work together as a whole and required significant fine-tuning.

  \item \textbf{A correctness theorem for the translation} (\cref{sec:properties}):
    We prove that the translation preserves the semantics of declared functions and
    types (\cref{l:operationalmain,l:same-model}) and that inconsistencies are
    preserved (\cref{lem:consistency}), so \lava is a certificate checker for \lh.
    The meta-theory sheds light on how to design such a translation for similar sources and targets.

    \item \textbf{An implementation and evaluation of \toolname} (\cref{sec:impl,sec:eval}):
    Our tool consists of around \hsLines K \sloc of \haskell code implementing
    the skeleton translation, and around \rocqLines K \sloc for the automation tactics and \coq library.
    The implementation of the tool and automation tactics is a significant engineering effort.
    We evaluate the translation on eight representative \lh programs.
\end{enumerate}
While our work focuses on \lh---and, more marginally, on \coq---we believe it
reveals fundamental aspects of the semantics of refinement type systems.
It offers a basis for the translation of other refinement type frameworks to proof assistants.

\section{Overview of The Translation}\label{sec:overview}
We begin with an informal overview of \lava, presenting the key features of the translation.
\lava works in three steps.
First, parsing the \lh program into our core calculus $\er$, a step not described in this paper.
Second, typechecking the $\er$ program to add explicit casts annotations.
Third, translating the elaborated terms into \coq.
The main focus of this overview is the third step,
as we go we detail the three different sorts of casts added in the second step.

Our explanations are centered around the “safe” indexing function \hask{get} from the introduction,
defined on a type of integer lists arbitrarily restricted to elements greater than 5.
We start by translating the element type and the list datatype,
then translate \hask{get} with a refined specification and to an unrefined graph
relation, followed by the theorem \hask{thm_get}.
We finish with an example of a higher-order function (\hask{apply}) to show how higher-order values are handled.

\subsection{Translation of Refinement Types}
\label{subsec:refinement}

\lh is a refinement type system for subset types built atop \haskell and it adopts \haskell's underlying norms.
Refinement types (\textit{aka} subset types) $\refT x \lhtyp {\lhtm}$
denote the set of elements $e$ of type $\lhtyp$ such that $\lhtm \subst x e$ holds.
\coq also supports subset types, using the same notation to denote
a dependent pair of an element of $e \mathord{:} \lhtyp$ \emph{and a proof} of $\lhtm \subst x e$.

Our list indexing example depends on integers.
\lh comes with a builtin integer type, \hask{Int},
with over- and underflow left as undefined behavior.
We translate these underspecified machine integers to
\coq's unbounded integer type \rocq{Z}.
We translate the \lh type of integers greater than five \hask{{x:Int | 5 < x}} to \coq (\cref{def:trtypes}) as
\rocq?{x:Z | !$\ltbZRel$! 5 x true}?.

Refinement types translate to \coq's subset types,
so refinements translate to \coq's propositions, in $\Prop$.
Here, \rocq{5 < x} turns into a proposition about $\ltbZRel$, the three-way relation modeling \rocq{<} on integers: \rocq{(5,x,true)} $\in$ $\ltbZRel$ if and only if \rocq{5 < x}.
In the translation of refinements, we represent function calls (which includes
integer operations) with inductive relations, as detailed in \cref{sec:RToProp}.

In \lh, a value of type \hask?{x:Int | 5 < x}? is instantiated by a \haskell integer expression,
as verification of the refinement is external.
In \coq, given \rocq{x:Z} and \rocq{p_x: !$\ltbZRel$! 5 x true},
we build a value of the type \rocq?{x:Z | !$\ltbZRel$! 5 x true}? as
\rocq{exist (!$\forall$!(x:Z), !$\ltbZRel$! 5 x) true x p_x},
where \rocq{exist} is the constructor for dependent pairs,
requiring both an integer expression~\rocq{x} and a proof~\rocq{p_x}.

On dependent pairs, the conversion between the unrefined and refined semantics of a value is
straightforward by construction or destruction of dependent pairs (see \cref{subsec:reffunc}).
Instead, for each inductive datatype and function,
the translation must generate an unrefined and a refined version.

\subsection{Translation of Refined Inductive Datatypes}
\label{subsec:intlists}

Let us move on to our inductive datatype example: lists of integers greater than five.
\lh programs refine \haskell's inductive datatypes
with predicates on both the argument and return types of the constructors.
\lava supports first-order monomorphic inductive datatypes.
\begin{haskellcode}
{-@ data IList where
  Nil :: IList
  Cons_gt5 :: x:{Int | 5 < x} !$\to$! IList !$\to$! IList @-}
data IList where Nil :: IList ; Cons_gt5 :: Int !$\to$! IList !$\to$! IList
\end{haskellcode}
\noindent
Mapping unrefined definitions to \coq amounts to a change in syntax:
\begin{rocqcode}
Inductive IList_u := Nil_u : IList_u | Cons_gt5_u : Z !$\to$! IList_u !$\to$! IList_u.
\end{rocqcode}
\noindent The unrefined type \rocq{IList_u} mimics the underlying \haskell type.
It is both the type assigned to lists inside type refinements
(in the logical, unrefined world) and the basis of the refined type \rocq{IList}.

\begin{rocqcode}
Fixpoint IList_wf (xs : IList_u) : Prop := ...
Notation IList := {xs : IList_u | IList_wf xs !$\land$! True}.
Definition Nil : IList := ...
Definition Cons_gt5 (x:Z | !$\ltbZRel$! 5 x) (xs: IList) : IList := ...
\end{rocqcode}
\noindent
Indeed, \rocq{IList} is a restriction of \rocq{IList_u} to \emph{well-formed} values.
Such values are characterized by a \emph{recursive refinement predicate}~\cite{atkey11,sekiyama15}
\rocq{IList_wf: IList_u !$\to$! Prop} encoding the original refinements,
defined by recursion on \rocq{IList_u} and asserting that all elements are greater than $5$.
The type \rocq{IList} also has a trivial refinement \rocq{True},
as refined datatypes can be further refined in \lh,
so \coq's \rocq{IList} notation corresponds to \lh's \hask?{v:IList | True}?.

The \emph{refined constructors} \rocq{Nil} and \rocq{Cons} are defined separately
with the refined signature of the respective constructors.
\coq's notation \rocq{(x:Z | !$\ltbZRel$! 5 x)} is syntactic sugar for \rocq{(x:{x:Z | !$\ltbZRel$! 5 x})},
in the same way that \hask?x:{Int | 5 < x}? stands for \hask?x:{x:Int | 5 < x}? in \lh.

Our translation does not target standalone refined datatypes in \coq directly---such translations are only sometimes possible.
In \lh, constructor refinements may refer to constructors:
\begin{haskellcode}
{-@ ... ConsNonEmpty :: {x:Int | 5 < x} !$\to$! {xs:IList | xs !$\neq$! Nil} !$\to$! IList @-}
\end{haskellcode}
\noindent\coq's type theory forbids such a refinement:
we cannot use a constructor \emph{while defining the type}.
Instead we define the unrefined datatype with its constructors first,
and since refinements use the unrefined constructors,
we can use \rocq{ConsNonEmpty_u} while defining the well-formedness predicate.

\subsection{Translation of Functions with Dependent Pattern Matching}
\label{subsec:reffunc}

We assume an \hask{IList}-length function \hask{llen} in \lh, correctly translated to a refined definition \rocq{llen}
and an unrefined graph relation \rocq{llen_rel}.
Our example is the safe-indexing function \hask{get}, which retrieves the element at a given index from an \hask{IList}.
\lh functions typically specify refinement types atop their \haskell definition:
\begin{haskellcode}
{-@ get :: xs:IList !$\to$! i:{Int | 0 !$\leq$! i !$\land$! i < llen xs} !$\to$! {v:Int | 5 < v} @-}
get :: IList !$\to$! Int !$\to$! Int
get (Cons x xs') i = if i == 0 then x else get xs' (i-1)
\end{haskellcode}

We translate \hask{get} to a \coq function typed with a refined signature below.
The \rocq{IList} argument is treated as trivially refined with \hask{True}:
types of top-level functions in \lh are always refined. 
We translate it to \coq as the subset type \hask{IList}.
The translation of the refinement of the index \hask{i} uses the relations for
$<$, $\leq$, and \rocq{llen}.
For the last, it is necessary to introduce a fresh variable \rocq{llen_xs}
capturing the result of the function call (see \cref{sec:RToProp}).
The return type is translated as before.

The definition is given as a \emph{proof script} terminated by \rocq{Defined}.
\begin{rocqcode}
Definition get (xs:IList)
  (i:Z | !$\lebZRel$! 0 i true !$\land$! !$\exists$! llen_xs, llen_rel !$\projleft$!xs!$\projright$! llen_xs !$\land$! !$\ltbZRel$! i llen_xs true)
  : {v:Z | !$\ltbZRel$! 5 v true}.
  destruct xs as [xs xs_p]. destruct i as [i i_p].
  generalize dependent i; induction xs as [ | x xs' IH_xs' ]; intros.
  - (* xs = Nil_u *) exfalso; solver.
  - (* xs = Cons_gt5_u x xs' *) destruct (i ==? 0) eqn:?.
    + (* i = 0 *) exact (exist (!$\forall$!(v:Z), !$\ltbZRel$! 5 v) x ltac:(solver)).
    + (* i !$\neq$! 0 *) exact (IH_xs' ltac:(solver) (i-1) ltac:(solver)).
Defined.
\end{rocqcode}
\noindent
The structure of the definition follows the structure of the original term,
with an induction on the list mirroring the definition's recursion.
Using \rocq{induction} in a “definition by tactics” gives us dependent pattern matching,
with propagation of the case analysis information, as well as an inductive hypothesis (\rocq{IH_xs'}).
Tactics neatly handle the required bookkeeping for us\,---\,the main reason we use definitions by tactics.
The parameter \rocq{xs} is destructed first into a value of type \rocq{IList_u}
and the proof of well-formedness \rocq{xs_p}, which allows matching on the
unrefined \rocq{xs}.

If \rocq{xs = Nil_u}, we have a contradiction:
\hask{0 !$\leq$! i !$\land$! i < llen Nil}, but
\hask{llen Nil = 0}.
We indicate unreachability of the branch by setting our goal to $\False$ with the \rocq{exfalso} tactic (for a contradiction).
Our automation tactic \rocq{solver} identifies the contradictory assumptions and applies
the principle of explosion appropriately, just as SMT would for \lh.

On the other hand, if \rocq{xs = Cons_gt5_u x xs'},
we use \rocq{destruct} to go by cases on whether \rocq{i} is 0,
translating \hask{(==)} in \lh to a generic decision procedure \rocq{(==?)} in \coq.
Both branches are translated directly as (Gallina) $\lambda$-terms that are
injected into the proof script by the tactic \rocq{exact}.
In our translation, structural expressions (match, lambdas and let) are translated with tactics,
while simpler terms (applications, literals...) are translated directly as $\lambda$-terms.

We could have used frameworks like Program~\cite{sozeau07} or Equations~\cite{sozeau19} to define \rocq{get}.
Instead, we choose to rely only on basic \coq tactics, which are flexible and
reliable---a trade-off that works for our implementation (\cref{sec:related}).

\paragraph{Injection and projection casts}

During elaboration of our core calculus $\er$, we insert explicit
\emph{injections} and \emph{projections} between refined and unrefined types.
This is required by the intertwined semantics of \lh,
where terms are used as refined or unrefined in different contexts.

\emph{Injections} annotate unrefined values.
In $\er$, every term is considered refined, except for literals
(Integers, Booleans and Unit) and first-order local variables.
We use unrefined first-order variables, so they can be matched on and
destructed in the \lh code, like \hask{xs}.
This is also a pragmatic choice to simplify automated verification in \coq:
it is easier to deal with a deconstructed dependent pair than with the pair itself.
We always destruct function parameters, like \rocq{xs} and {i} above.

In the elaboration of \hask{get}, the occurrence of \hask{x} from the branch \hask{i = 0}
is annotated with an injection to a singleton type: \hask?inj (x, {v:Int | 5 < v})?.
The injection is translated as the dependent pair constructor \rocq{exist}.
The refinement witness (the second element of the pair)
is generated by our automation tactic \rocq{solver}.

\emph{Projections} occur in specifications,
when refined terms appear in that logical, unrefined world.
In the type of \hask{get}, the first parameter \hask{xs} is of type \hask{IList},
a \emph{refined} type, but then appears in the refinement of \hask{i},
as an argument to the logical function \hask{llen}.
The name \hask{xs} in the refinement of \hask{i} refers to the first projection
of the refined \hask{xs}.
The elaboration phase adds the relevant annotation:
\begin{lrcode}
get :: xs:IList !$\to$! i:{Int | 0 !$\leq$! i !$\land$! i < llen (proj xs)} !$\to$! {v:Int | 5 < v}
\end{lrcode}
\noindent A \hask{proj} annotation is translated to a custom operation $\proj\cdot$,
generic on dependent pairs and packs (see \cref{subsec:overview-ho}).
On dependent pairs, $\proj\cdot$ is defined as \coq's builtin first projection \rocq{proj1_sig}.

In \lh and in \coq, term application triggers a substitution in the type.
Consider for instance the partial application \hask{get Nil} (ill-typed in \lh, which is irrelevant here).
To type this term, we substitute the variable \hask{xs} by the argument \hask{Nil} in the refinement of \hask{i}:
\begin{lrcode}
get Nil :: i:{Int | 0 !$\leq$! i !$\land$! i < llen (proj Nil)} !$\to$! {v:Int | 5 < v}
\end{lrcode}
When \hask{get Nil} appears in a refined context, like a function definition, \hask{Nil} is refined.
However, \hask{llen} in the refinement of \hask{i} requires an unrefined argument.
The projection is implicit in \lh, but must be done explicitly in \coq.

\subsection{Translation of Functions as Unrefined Graph Relations}
\label{subsec:unreffunc}

A \lh user can \emph{reflect} functions into the SMT logic~\cite{vazou17}.
While functions are encoded in SMT as uninterpreted symbols,
refinement reflection essentially permits function unfolding,
thus enables reasoning about the function's definition.

The code that is reflected into the logic is the \haskell definition---\emph{unrefined}.
To be reflected, a function must terminate and be defined by structural induction.
In \coq, we must give a translation of the reflected \haskell function, to be
used in the translation of refinements.
This function can be partial in \haskell (but must be total on the refined domain, like \hask{get}), so
we choose to encode it as an inductive predicate representing the graph of the function.
We translate an $n$-ary function into an $n{+}1$-ary relation in $\Prop$ (\cref{sec:transfunc}).

Assuming \hask{get} is reflected using the directive
\hask?{-@ reflect get @-}?, we generate an inductive predicate with two
constructors, one for each branch of the conditional in \hask{get}:
\begin{rocqcode}
Inductive get_rel : IList_u !$\to$! Z !$\to$! Z !$\to$! Prop :=
| get_Cons_true x xs' i : i == 0 !$\to$! get_rel (Cons_gt5_u x xs') i x
| get_Cons_false x xs' i : i /= 0 !$\to$! !$\forall$!(im1:Z), subZ_rel i 1 im1
   !$\to$! !$\forall$!(v:Z), get_rel xs' im1 v !$\to$! get_rel (Cons_gt5_u x xs') i v.
\end{rocqcode}
\noindent
We replace the function calls by relation calls (introducing fresh variables),
as we are in an unrefined context.
This includes operator calls for homogeneity and because some operators can be partial.

Notice that \rocq{get_rel} has no constructor for \rocq{Nil_u},
because this case is absent from the \haskell original code.
In general, the inductive predicate can hold for elements outside of the refined domain.
Assume that \hask{get xs i} has an additional equation \hask{get Nil i = 0} in \lh.
In that case, \rocq{get_rel} has a third constructor \rocq{get_rel_Nil i : get_rel Nil_u i 0}.
The relation is consistent with the evaluation of \hask{get Nil i} to \hask{0} in \haskell.

Crucially, we automatically prove (using custom tactics) that the graph denotes
a function (\emph{functionhood}) and that functions and their graphs coincide on the refined domain (\emph{coherence}):
\begin{rocqcode}
Lemma get_funct xs i v v' : get_rel xs i v !$\to$! get_rel xs i v' !$\to$! v = v'.
Lemma get__get_rel xs i v : !$\projleft$!get xs i!$\projright$! = v !$\leftrightarrow$! get_rel !$\projleft$!xs!$\projright$! !$\projleft$!i!$\projright$! v.
\end{rocqcode}
After proving coherence, we mark \rocq{get} as opaque, in other words,
its body is invisible in later contexts.
This way, only the graph relation is available to automated reasoning.
This mimics the reflection of the unrefined function into the logical context for verification in \lh.

We also define a relation for unreflected functions, but in terms of the refined
definition and with a single constructor.
Supposing that \rocq{get} were \emph{not} reflected, we infer the relation:
\begin{rocqcode}
Inductive get_rel2 : IList_u !$\to$! Z !$\to$! Z !$\to$! Prop :=
| get_xs_i (xs:IList) (i:Z | ...) (v:Z) : !$\projleft$!get xs i!$\projright$! = v !$\to$! get_rel2 !$\projleft$!xs!$\projright$! !$\projleft$!i!$\projright$! v.
\end{rocqcode}
\noindent
The functionhood and coherence lemmas trivially hold in this case.
Since we made the definition of the refined \rocq{get} opaque,
reasoning using the definition of \hask{get} is not possible.
This mimics the encoding of unreflected functions as uninterpreted symbols in the logic.

\subsection{Theorems and Subtyping}
\label{subsec:overview-thm}

The graph relation \rocq{get_rel} is used in the translation of \hask{thm_get} from the introduction. 
\begin{rocqcode}
Definition thm_get (xs:IList) (x:Z | !$\ltbZRel$! 5 x true)
  (i:Z | !$\lebZRel$! 0 i true !$\land$! !$\exists$! llen_xs, llen_rel !$\projleft$!xs!$\projright$! llen_xs !$\land$! !$\ltbZRel$! i llen_xs true)
  : {{ !$\exists$! get_xs, get_rel !$\projleft$!xs!$\projright$! !$\projleft$!i!$\projright$! get_xs !$\land$! !$\exists$! ip1, !$\addZRel$! !$\projleft$!i!$\projright$! 1 ip1
       !$\land$! !$\exists$! get_C, get_rel (Cons_u !$\projleft$!x!$\projright$! !$\projleft$!xs!$\projright$!) ip1 get_C !$\land$! get_xs == get_C }}.
  destruct xs as [xs xs_p]. destruct i as [i i_p].
  exact (subCast (!$\forall$!(u:unit), u == tt) (!$\forall$!(u:unit), !$\exists$! get_xs ... get_xs == get_C)
                 (exist (!$\forall$!(u:unit), u == tt) tt ltac:(solver))).
Defined.
\end{rocqcode}
\noindent
A theorem in \lh is a function returning a refinement of the unit
(where \hask?{ e }? stands for \hask?{u:() | e}?).
In particular, theorems may be defined by pattern matching on the arguments.
Thus, we translate theorems in the same way as any other function (but do not
generate graph relations or packs).
In \coq, we write \rocq?{{ e }}? for \rocq?{u:unit | e }?.
\rocq{tt} is the unit term.

In \lh, verification of \hask{get_rel} is automatic:
the “proof” is just the unit \hask{()}.
A single verification condition is generated by the subtyping constraint
from the singleton type \hask?{u:() | u == ()}? of \hask{()}
to \rocq{thm_get}'s return type.
Using the notations from the introduction, \lh verifies:
\begin{center}
  SMT-Valid $(\embedleft$\hask{5 < x}$\embedright \land \embedleft$
  \hask{0 !$\leq$! i !$\land$! i < llen xs}$\embedright$
  $\Rightarrow \embedleft$\hask{u == ()}$\embedright$
  \\$\Rightarrow \embedleft$\hask{get xs i == } \hask{get (Cons x xs) (i+1)}$\embedright)$
\end{center}

In our intermediate calculus $\er$, we do not check such implications:
subtyping is an axiom, effectively transforming \lh's subtyping rule \TSRef to \TSRefAx below:
\[
  \inference{
    \text{SMT-Valid }(\embed{\Gamma} \Rightarrow \embed{\lhtm_l} \Rightarrow \embed{\lhtm_r})
  }{
    \hsubj \typEnv {\refT {x} \lhtyp {\lhtm_l}} {\refT {x} \lhtyp {\lhtm_r}}
  }[\TSRef]
  \qquad
  \inference{
  }{
    \hsubj \typEnv {\refT {x} \lhtyp {\lhtm_l}} {\refT {x} \lhtyp {\lhtm_r}}
  }[\TSRefAx]
\]
Verification is the responsibility first of \lh, then of \coq.
However, during elaboration, we make all (non-trivial) subtyping occurrences explicit,
by adding a subsumption cast annotation.
The body of \hask{thm_get} is elaborated as:
\begin{lrcode}
  sub (inj ((), {u:() | u == ()}),
       {u:() | u == ()}, {u:() | get xs i == get (Cons x xs) (i+1)})
\end{lrcode}
The subsumption annotation takes the source and target type as second and third parameters.
As any literal, \hask{()} is injected into its singleton type.

The body of \coq's \rocq{thm_get} is a direct translation of this elaborated term.
It uses a custom cast operation \rocq{subCast} (described in \cref{sec:casts}).
Verification is handled by our automation tactic, as always.

Verification of \hask{thm_get} in \lh involves a step of evalution of the right side of the equality.
This is permitted by reflection of \hask{get}.
In \coq, the proof is also only possible if \hask{get} is reflected,
by inversion on \rocq{get_Cons_false}.

\subsection{Function Packs}
\label{subsec:overview-ho}

With first-order objects (datatypes or literals), it is easy to cast between refined and unrefined
with injections and projections.
But functions are first-class objects of \haskell, and are also used
implicitly as refined or unrefined terms in \lh depending on the context.
Therefore, inside the translation of a \lh function, we must be able to use higher-order
values both as refined functions and as inductive graph relations.

To this aim, we translate higher-order values to records storing the refined
definition and the unrefined graph relation, as well as the coherence lemma and
functionhood property.
This record 
is called a \emph{pack}, its signature is parameterized by an arrow type.\footnote{
  In practice, to support a common representation over functions of different
  arities, inside packs, we encode the function parameters as a dependent list
of types and refinements.}
We always create a pack when translating a top-level function.
For \rocq{get}, we have:
\begin{rocqcode}
Instance get_pack : !$\fClass$! (!$\forall$!(xs:IList) (i:Z | ...), {v:Z | !$\ltbZRel$! 5 v true}) :=
  {| !$\defnmClass$! := get; !$\relnmClass$! := get_rel; !$\unrefproofClass$! := get__get_rel; !$\functrelClass$! := get_funct |}.
\end{rocqcode}

\noindent
For an example use of packs, consider a refined \rocq{apply} function:
\begin{haskellcode}
{-@ reflect apply @-}
{-@ apply :: f:(xs:IList !$\to$! i:{Int | ...} !$\to$! {v:Int | 5 < v})
             !$\to$! xs:IList !$\to$! {v:Int | v = f xs} @-}
apply :: (IList !$\to$! Int !$\to$! Int) !$\to$! IList !$\to$! Int
apply f xs i = f xs i
\end{haskellcode}
\noindent
The functional argument \hask{f} appears unrefined in the output refinement,
but refined in the body of \hask{apply}.
Its type is translated to a \emph{pack signature}, making both versions of \hask{f} available:
\begin{rocqcode}
Definition apply (f: !$\fClass$! (!$\forall$!(xs:IList) (i:Z | ...), {v:Z | !$\ltbZRel$! 5 v true}))
  (xs:IList) (i:Z | ...) : {v:Z | !$\relnmInst{f}$! xs v}.
  destruct xs as [xs_u xs_p], i as [i_u i_p].
  exact (subCast (!$\forall$!v, !$\lebZRel$! 0 v true) (!$\forall$!v, !$\relnmInst{f}$! xs_u i_u v)
                 (!$\defnmInst{f}$! (exist _ xs ltac:(solver)) (exist _ i ltac:(solver)))).
Defined.
\end{rocqcode}
In the output refinement, we use the graph relation $\relnmInst f$ from the pack,
while in the body we use the refined function $\defnmInst f$.\footnote{
  In practice, we use wrapper functions, not record fields, since packs are defined over dependent lists of parameters.}
The function \hask{apply} can be applied to the pack \hask{get_pack}. 
Then, the projection $\relnmInst f$ is instantiated to $\relnm{\mathtt{get}}$,
while $\defnmInst f$ is instantiated to \hask{get}.

In the unrefined world, when building the relation \rocq{apply},
we also use a pack for the higher-order arguments, this time an \emph{unrefined pack}.
Unrefined packs only contain the two unrefined elements of the packs:
the relation (used below with $f.(\relnmUClass)$) and the functionhood property.
\begin{rocqcode}
Inductive apply_rel : !$\fuClass$! (IList_u !$\to$! Z !$\to$! Z) !$\to$! IList_u !$\to$! Z !$\to$! Z !$\to$! Prop :=
| apply_f_xs f xs i : !$\forall$! v, !$\relnmUInst{f}$! xs i v !$\to$! apply_rel f xs i v.
\end{rocqcode}
To instantiate \rocq{apply_rel} with \hask{get}, we define a projection from refined to unrefined packs
(here from \rocq?!$\fClass$! (IList !$\to$! {v:Z | !$\lebZRel$! 0 v true})? to \rocq?!$\fuClass$! (IList_u !$\to$! Z)?).
This projection simply forgets two of the fields to keep only the graph relation and functionhood property.

\section{The Translation}\label{sec:translation}
\toolname translates \lh to \coq in three steps.
The first step elaborates
\haskell's internal representation Core \cite{GHC-CORE} and \lh specifications into our
calculus $\er$ (presented below).

The second step elaborates the terms of $\er$ to prepare the translation.
\lh uses implicit subtyping \cite{713327}
and binders can appear as refined or unrefined according to the context.
Therefore, elaboration converts subtyping constraints
to type casts in $\er$-terms (\cref{sec:typing}).

The third step is the translation of annotated terms into \coq, detailed in
\cref{sec:transDecl,sec:refined,sec:unrefined}.

\subsection{Elements of the Translation and Internal Representation}\label{sec:syntax}

\subsubsection{Syntax of \texorpdfstring{$\er$}{the internal calculus}}

\Cref{fig:er} presents the syntax of $\er$,
a core $\lambda$-calculus with refinement types,
user-defined first-order inductive data types and functions.
The core grammar of $\er$ is in black, while the elaborated grammar includes the
terms in gray.
This work does not cover the following features of \lh:
polymorphism, partial application and lambdas in refinements or reflected functions, data constructors with higher-order
parameters, and non-structurally recursive functions.

\begin{figure}
  \[ \begin{array}{rrclll}
    \gramTitle{Base Types}
        & \lhtyp & \Coloneqq &
        \Unit \mid \Bool \mid \Int \mid \lhTC
        \\ \gramTitle{Refinement Types}
        & \lhref & \Coloneqq & \refT x \lhtyp \lhtm
        \mid \arr[x][\lhref][\lhref]
        \\[5pt] \gramTitle{Declarations}
        & \lhdecl & \Coloneqq &
        \lhdata \lhTC {\gramplus{\lhConstr ::\lhref}}
        \mid \lhdefrefl f \lhref \lhexp
        \\ \gramTitle{Expressions}
        & \lhexp & \Coloneqq &
        \lhtm \mid \lhlet x \lhref \lhexp \lhexp
        \mid \annlhqmarkgrammar \mid\ite \lhtm \lhexp \lhexp
     \\ &&\mid & \lhcase \lhtm {\gramplus{\case{\lhConstr\;\seq x}{\lhexp}}}
     \mid \dead
        \\[5pt] \gramTitle{Refinements}
        & \lhtm & \Coloneqq &
        \varargrammar \mid \lhlit
        \mid \lhConstr
        \mid \lhtm \ap \lhtm
        \mid \lhnegb \lhtm
        \mid \lhtm_1 \lheqb \lhtm_2
        \mid \lhtm_1 \lhop \lhtm_2
        \quad \color{gray}{s \in \{\topSymbol,\locSymbol,\recSymbol\}}
     \\ &&\mid& \color{gray}{\lhinj \lhtm \lhref}
     \mid {\color{gray}{\lhsub \lhtm \lhref \lhref}}
     \mid \color{gray}{\lhproj\lhtm}
        \\ \gramTitle{Literals}
        & \lhlit & \Coloneqq &
        \unit \mid \true \mid \false \mid 0, $-1$, 1 \dots
        \\ \gramTitle{Binary Operators}
        & \lhop & \Coloneqq &
        \lhplus \mid \lhminus \mid \lhtimes \mid \lhdiv\; 
        \mid \mathbin{<} \mid \mathbin\leq \mid \mathbin{>} \mid \mathbin\geq
        \mid \mathbin\lhimp \mid \lhandb \mid \lhorb \mid \lhiff
  \end{array} \]
  \Description{Syntax of the calculus lambda r}
  \caption{Syntax of $\er$; its monomorphic base types $B$ are built into more general refined types $R$.
    We write $\seq{(\cdot)}$ for a sequence of zero or more elements,
    $\gramplus\cdot$ for a non-empty sequence and
    $\gramopt\cdot$ for an optional annotation.
    In gray, the extended grammar after typing.}
  \label{fig:er}
\end{figure}

The calculus~$\er$ uses the metavariables
$\lhTC$ and $\lhConstr$ for type and data constructors,
while the term variables are drawn from the set of all variables $\varSet \supseteq \left\{ x, y, z, f, g, v \right\}$.
The name $v$ is reserved for binders in refinement types.
In the extended syntax, variables are annotated with their arity ($n$) and their scope:
whether they are top-level ($\topSymbol$), local ($\locSymbol$), or
the name of the function being defined ($\recSymbol$).

We define $\er$'s \emph{types} in two layers: base and refinement types.
A \emph{base type} ($\lhtyp$) is either builtin (Unit, Boolean or Integer)
or a user-defined type constructor.
A \emph{refinement type} ($\lhref$) is either a base type refined by a proposition
or a function between two refinement types.
In function types, the co\-domain's refinement can refer to the domain.
We abbreviate the trivial refinement $\refT x \lhtyp \true$ to $\refTT{\lhtyp}$;
we  abbreviate $\bind{x_1}{\lhref_1} \to \dots \to \bind{x_n}{\lhref_n} \to \lhref$ to
$(\bind{x_i}{\lhref_i})_{i \leq n} \to \lhref$, where $n$ is maximal.

A $\er$ program is a sequence of \emph{declarations} ($\lhdecl$),
where each declaration defines either
a  type constructor or a function.
Just like in \haskell, every function can be recursive;
just like in \lh, every function can be
\emph{reflected} into the logic.
Top-level function definitions contain implicit $\lambda$-binders and are $\eta$-long:
$(\keyw{reflect})\ \lhdef f {
	\arr[x_1][\lhref_1][\dots \to \arr[x_n][\lhref_n][\refT x \lhtyp \lhtm]]
} \lhexp$
binds the formal parameters $x_1$ through $x_n$ of $f$.
Type constructor definitions must be first-order: every constructor's arguments
must be refinements of base types, not arrow types.

The body of a definition is a $\er$ \emph{expression} ($\lhexp$).
We use a form of generalized ANF~\cite{sabry92}, where expressions are broken
into two levels: refinements $\lhtm$ and structural expressions $\lhexp$.
The upper level contains let-bindings, conditionals, pattern matching,
as well as \emph{proof hints} ($\lhqmark{\lhexp}{\lhtm_h}$) that provide hints
(to the SMT solver) to prove theorems.
Let-bindings are annotated with a refinement type: like in top-level declarations,
if the type is an arrow, it binds the formal parameters of $\lhexp$.
In other words, all $\lambda$-abstractions are implicit and let-bound.
Pattern matching can be partial, but all missing branches are signaled as $\dead$.

\emph{Refinements} $\lhtm$ are both expressions and type refinements.
We distinguish Boolean equality from binary operators.
\emph{Binary operators} ($\lhop$) are primitives for
arithmetic, equality, and Boolean logic.
\emph{Literals} ($\lhlit$) are Unit, Boolean, or Integer values.
The extended grammar has three type casts: \emph{injections} from base
to refined types, \emph{subsumptions}
over refinement types, and
\emph{projections} from refined to base types.

\subsubsection{Structure of the translation from $\er$ to $\coq$.}
\label{sec:transStructure}

\begin{table}
  \label{fig:summarytrans}
    \begin{tabular}{|c|c|c|c|}
        \hline
        \textbf{Target World}
        & \textbf{Element of $\er$} & \textbf{Function} & \textbf{Definition}\\
        \hline
        \textbf{Refined and Unrefined}
        & Declaration
        & $\trLhdeclCqName$
        & \Cref{fig:transdata,fig:transfunc} \\
        \hline
        \multirow{3}{*}{\textbf{Unrefined}}
        & Type
        & $\typesToUnRefName$
        & \Cref{def:utransreftype} \\
        & Refinement
        & $\utransName$
        & \Cref{def:utranstm} \\
        & Refinement (as $\Prop$) & $\refToPropName$ & \Cref{def:utrLhtmCqprop} \\
        \hline
        \multirow{3}{*}{\textbf{Refined}}
        & Type & $\typesToRefName$ & \Cref{def:trtypes} \\
        & Refinement & $\trLhCqName$ & {\iflongversion \Cref{fig:transtm} \else Appendix \fi} \\
        & Expression (as tactics) & $\exprToTacName$ & {\iflongversion \Cref{fig:transexp} \else Appendix \fi} \\
        \hline
  \end{tabular}
  \Description{Table summarizing the translation functions}
  \caption{The core functions of the \toolname translation from $\er$ to \coq.}
\end{table}
The translation of $\er$ is performed by the seven functions summarized in~\cref{fig:summarytrans}.
The translation follows the structure of a \lh file and translates each
datatype or function
declaration independently with $\trLhdeclCqName$ (\cref{sec:transDecl})
outputting both a refined and an unrefined version of the declaration.
Types, refinements, and expressions have two sets of translations:
to the unrefined world (\cref{sec:unrefined}) and to the refined world (\cref{sec:refined}).
In the unrefined world, type refinements are translated to propositions
($\refToPropName$) and refinement expressions to arbitrary terms
($\utransName$).
In the refined world, expressions are translated as tactics ($\exprToTacName$).

\subsection{Elaboration of Explicit Type Casts}\label{sec:typing}

Before translating $\er$ to \coq, \toolname performs an elaboration phase
turning implicit subtyping into explicit type casts,
via well-formedness and typing derivation rules mirroring \lh rules.
We expand the explanations of the overview (\cref{sec:overview}) with formal
definitions of the significant rules.
The remaining rules of the type system pose no surprise and are given in \appendixref{app:typing}.

\paragraph{Elaboration of refinement types}

We define a well-formedness judgment on types that annotates type refinements
with projections and pairs arity/scope on the variables.
For instance, the full elaboration of the type of \hask{get} (\cref{subsec:reffunc}) is:
\begin{lrcode}
  xs:IList !$\to$! i:{Int | 0 !$\leq$! i!$\vararloc{}{0}$! !$\land$! i!$\vararloc{}{0}$! < llen!$\varartop{}{1}$! (proj xs!$\vararloc{}{0}$!)} !$\to$! {v:Int | 5 < v!$\vararloc{}{0}$!}
\end{lrcode}
Well-formedness is checked by two rules:
\[ \begin{array}{c@{\enspace}c}
  \inference{
    \text{if } \lhtyp = \lhTC\text{, then } \lhtyp \in \dom(\typEnv) \\
    \hsmpj {\typEnv, \bind x \refTT{\lhtyp}} \lhtm \Bool {\ann\lhtm}
  }{
	  \hwfj \typEnv {\refT x \lhtyp \lhtm} {\refT x \lhtyp {\ann\lhtm}}
  }[\!\!\WFTBase]
  &
  \inference{
	  \hwfj \typEnv {\lhref_x} {\ann{\lhref_x}}
    &\hwfj {\typEnv, \bind{x}{\ann{\lhref_x}}} \lhref {\ann\lhref}
  }{
    \hwfj \typEnv {\arr} {\arr[x][\ann{\lhref_x}][\ann{\lhref}\subst x {\lhproj x}]}
  }[\!\!\WFTFun]
\end{array} \]

Rule~\WFTBase checks that type refinements $\lhtm$ are Boolean in a
simple type system (defined in \appendixref{app:typing}).
Additionally, it annotates the variables inside of $\lhtm$ with their arities and scopes.
The scope is either $\locSymbol$, for variables local to the function,
or $\topSymbol$, for variables defined at top-level.
We do not typecheck refinements with a refinement type system as
we only want to generate type annotations and add projections,
not check that refinements are satisfied.
In this and in all other rules, $\typEnv$ is a typing environment populated
with elaborated types.

Rule~\WFTFun elaborates the arrow type $\arr$ by projecting the occurrences
of $x$ in refinements of $\lhref$.
The definition of the substitution is standard, it is capture-avoiding and
assumes all bound variables are distinct.

\paragraph{Subsumption}

The elaboration of expressions introduces subsumptions and injections,
as well as the arity/scope annotations on variables.
Our type system for expressions is \emph{bidirectional}:
the types of refinements are \emph{synthesized} ($\Rightarrow$), \ie are an
output of the judgment,
while structural expressions (lets, conditionals, case and hints) are \emph{checked}
($\Leftarrow$) against a type.
Both kinds of judgments output an annotated term.

As usual in bidirectional systems, subtyping is where synthesis meets checking:
the synthesized type $\lhref_s$ must be a subtype of the required type $\lhref_t$
(both are already elaborated by construction).
\[
  \inference{
    \hsynj \typEnv \lhtm {\lhref_s} {\ann{\lhtm}}
    \quad \hsubj \typEnv {\lhref_s} {\lhref_t}
  }{
    \hchkj \typEnv \lhtm {\lhref_t} {\lhsub {\ann\lhtm} {\lhref_s} {\lhref_t}}
  }[\!\!\CSyn]
\]
Since types of structural expressions are only ever checked, subsumption only
applies to refinements, so the rule is defined on a refinement $\lhtm$
rather than an expression $\lhexp$.
However, subtyping applies on any kind of refinement type $\lhref$: a subset or an arrow.
The subtyping judgment is contravariant on arrows, and valid for any two subset
of the same base type (with the rule given in \cref{subsec:reffunc}).

\paragraph{Injection}

There are two unrefined atoms of $\er$: literals and local first-order variables.
They are elaborated by an injection into a subset type (our judgments only derive refinement types):
\[
  \inference{}{\hsynj \typEnv \lhlit {\csttype\lhlit} {\lhinj{\lhlit}{\csttype\lhlit}}}[\!\!\SLit]
  \quad\inference{(\bind{\varloc x}{\lhref}) \in \typEnv & \ar(\lhref) = 0}
  {\hsynj \typEnv x \lhref {\lhinj{\vararloc{x}{0}}{\lhref}}}[\!\!\SVarLoc]
\]
Rule~\SLit casts a literal into a singleton type (or \textit{selfified} type~\cite{ou04})
$\csttype\lhlit = \refT v \lhtyp {\vararloc{v}{0} == \lhlit}$,
with $\lhtyp$ the expected base type of $\lhlit$.
Rule~\SVarLoc is the rule for a local variable.
The scope is present in the environment and inserted into the elaborated term.

\paragraph{Application}

The application rule is important, even though it does not add annotations by itself.
\[
  \inference{
    \hsynj \typEnv {\lhtm_1} {\arr} {\ann{\lhtm_1}}
    & \hchkj \typEnv {\lhtm_2} {\lhref_x} {\ann{\lhtm_2}}
  }{
    \hsynj \typEnv {\lhtm_1 \ap \lhtm_2} {\lhref \subst x {\ann{\lhtm_2}}} {\ann{\lhtm_1} \ap \ann{\lhtm_2}}
  }[\!\!\SApp]
\]
The type of the argument is \emph{checked},
meaning that a subsumption annotation will be inserted on the argument,
as the argument is a refinement whose type will be synthesized.

More importantly, the (elaborated) argument is substituted into the refinements of $\lhref$.
This motivates the syntactic restriction allowing only refinements as arguments.
It also justifies surrounding $x$ by projections in the type refinements
of $\lhref$, as $\ann{\lhtm_2}$ has a refinement type.

\paragraph{Hints}

Our formalization explicitly considers ``proof hints'' $\lhqmark{\lhexp}{\lhtm_h}$,
a common technique in Liquid Haskell for adding a fact ($\lhtm_h$) for SMT to use while considering $\lhexp$.
The term $\lhqmark{\lhexp}{\lhtm_h}$ elaborates to $\annlhqmark{\lhexp}{\lhtm_h}{\lhtm_p}$,
where $\lhtm_h$ is the proof term given as hint and $\lhtm_p$ is the
proposition (the refinement of the type of $\lhtm_h$) that it proves.
\[
  \inference{
    \hsynj \typEnv {\lhtm_h} {\refT x \Unit {\lhtm_p}} {\ann{\lhtm_h}}
    \quad \hchkj {\typEnv, \bind{\varloc x}{\refT x \Unit {\lhtm_p}}} {\lhexp} \lhref {\ann\lhexp}
  }{
    \hchkj \typEnv {\lhqmark{\lhexp}{\lhtm_h}} \lhref {\annlhqmark{\ann{\lhexp}}{\ann{\lhtm_h}}{\lhtm_p}}
  }[\!\!\CHint]
\]

\subsection{Top-level Translation of Data Type and Function Declarations}\label{sec:transDecl}

The translation of programs follows the syntactic structure, translating datatypes (\cref{sec:transdata}) 
and functions (\cref{sec:transfunc}),
using the function $\trLhdeclCqName$ (\cref{fig:summarytrans}).
The translation also generates additional lemmas and proof automation hints; 
omitted here but detailed in \appendixref{app:lemmas}.

\subsubsection{Translation of Inductive Datatype Declarations}%
\label{sec:transdata}

Recall the refined datatype \hask{IList} from \cref{subsec:intlists},
and its translation given in \cref{fig:trans-ilist}
(with complete definitions of \rocq{IList_wf}, \rocq{Nil} and \rocq{Cons}):
\begin{lrcode}
data IList where Nil :: IList | Cons_gt5 :: x:{Int | 5 < x} !$\to$! IList !$\to$! IList
\end{lrcode}

\begin{listing}
\begin{rocqcode}
Inductive IList_u := Nil_u : IList_u  | Cons_gt5_u : Z !$\to$! IList_u !$\to$! IList_u.
Fixpoint IList_wf (xs:IList_u) : Prop :=
  match xs with
  | Nil_u => True                      (* refinement for the conclusion (trivial) *)
  | Cons_gt5_u x xs' => !$\ltbZRel$! 5 x true  (* refinement on x (5 < x)                 *)
                  !$\land$! True              (* refinement for xs' (trivial)            *)
                  !$\land$! IList_wf xs'      (* recursion on xs'                        *)
                  !$\land$! True              (* refinement for the conclusion (trivial) *)
  end.
Notation IList := {xs:IList_u | IList_wf xs !$\land$! True}.
Definition Nil : IList := exist _ Nil_u ltac:(solver).
Definition Cons_gt5 (x:Z | !$\ltbZRel$! 5 x) (xs:IList) : IList :=
  exist _ (Cons_gt5_u !$\projleft$!x!$\projright$! !$\projleft$!xs!$\projright$!) ltac:(solver).
\end{rocqcode}
\Description{Translation of IList}
\caption{Translation of \hask{IList}}
\label{fig:trans-ilist}
\end{listing}

The translation of a refined datatype $\lhTC$ (\cref{fig:transdata}) generates the same kinds of \coq declarations:
the unrefined datatype $\unrefnm\lhTC$, the well-formedness predicate $\wffnm\lhTC$,
the subset type $\lhTC$ and the refined constructors $\lhConstr_1, \dots, \lhConstr_n$.
In the well-formedness predicate, we use a substitution $\sigma$ to replace the
projections of the constructor parameters by the parameters themselves,
as they are unrefined in that context.
The translation also instantiates two equalities (not shown here):
$==$ (returning $\Prop$) and $\cqeqb$ (returning $\cqBool$).

\begin{figure}
\begin{tcolorbox}
  $\trLhdeclCq{\lhdata \lhTC {(\lhConstr_i :: (\bind{x_i^j}{\lhref_i^j})_{j \leq m_i} \to \refT x \lhTC {\lhtm_i})_{i \leq n}}} \doteq$
  \begin{rocqcode}
Inductive !$\unrefnm\cqTC$!: Set := !$(\mid \unrefnm{\lhConstr^i}: \utrLhrefCq{\lhref_i^1} \to \dots \to \utrLhrefCq{\lhref_i^{m_i}} \to {\unrefnm\lhTC})_{i \leq n}$!.
Fixpoint !$\wffnm{\lhTC} \ap (\coqbind{x}{\unrefnm\lhTC})$!: Prop :=
  match !$x$! with !$\left(\mid \case{\unrefnm{\lhConstr^i} \ap x_i^1 \dots x_i^{m_i}}{\utrLhtmCqprop{\lhtm_i\sigma\subst{x}{\lhConstr\ap x_i^1 \dots \ap x_i^{m_i}}} \wedge \bigwedge_{j \leq m_i} {\cqtm_i^j}} \right)_{i \leq n}$! end.
Global Notation !$\lhTC$! := {!$\coqbind{x}{\unrefnm\lhTC}$! | !$\wffnm\lhTC \ap x$! !$\land$! True}.
!$\Big($!Definition !$\refnm{\lhConstr_i} \ap (\coqbind{x_i^j}{\trLhrefCqEmp{\lhref_i^j}})_{j \leq m_i}: \trLhrefCqEmp\lhref$! :=
  exist (!$\cqlamuntyped{x}{\wffnm\lhTC \ap x \land \True}$!) !$(\unrefnm{\lhConstr^i} \ap \proj{x_i^1} \dots \proj{x_i^{m_i}})$! ltac:(solver). !$\Big)_{i \leq n}$!
  \end{rocqcode}
  {\small
  \textit{where} $\sigma = \subst{\lhproj{x_i^k}}{x_i^k}_{k \leq m_i}$
  and $\refT{\_}{\_}{\cqtm_i^j} \doteq \trLhrefCq{\lhref_i^j\sigma}$.}
\end{tcolorbox}
\Description{Translation of inductive datatypes}
\caption{Translation of an inductive datatype. See \cref{fig:summarytrans} for the definitions of the auxiliary declarations.}
\label{fig:transdata}
\end{figure}

\subsubsection{Translation of Function Declarations}%
\label{sec:transfunc}

We recall the translation of \rocq{get} in \cref{fig:trans-get}.
Formally, the translation of a function $f$ (\cref{fig:transfunc})
creates the same elements:
the refined function $f$, the unrefined graph relation $\relnm f$,
the coherence lemma $\unrefproof f$, the functionhood lemma $\functrel f$,
the refined pack $\fInst f$ and the unrefined pack $\fuInst f$
\footnote{For theorems (functions returning a refinement of the unit), we do not create
  relations nor packs as we assume that they are never used in an unrefined context.
  Similarly, in functions, higher-order parameters returning the unit are
translated as dependent arrows rather than packs.}.

\begin{listing}
\begin{rocqcode}
Definition get (xs:IList)
  (i:Z | !$\lebZRel$! 0 i true !$\land$! !$\exists$! llen_xs, llen_rel !$\projleft$!xs!$\projright$! llen_xs !$\land$! !$\ltbZRel$! i llen_xs true)
  : {v:Z | !$\ltbZRel$! 5 v true}.
  destruct xs as [xs xs_p]. destruct i as [i i_p]. ... Defined.
Inductive get_rel : IList_u !$\to$! Z !$\to$! Z !$\to$! Prop :=
| get_Cons_true x xs' i: ... | get_Cons_false x xs' i: ...
Lemma get_funct xs i v v' : get_rel xs i v !$\to$! get_rel xs i v' !$\to$! v = v'.
Lemma get__get_rel xs i v : !$\projleft$!get xs i!$\projright$! = v !$\leftrightarrow$! get_rel !$\projleft$!xs!$\projright$! !$\projleft$!i!$\projright$! v.
Instance get_pack : !$\fClass$! (!$\forall$!(xs:IList) (i:Z | ...), {v:Z | !$\ltbZRel$! 5 v true}) :=
  { !$\defnmClass$! := get; !$\relnmClass$! := get_rel; !$\unrefproofClass$! := get__get_rel; !$\functrelClass$! := get_funct }.
Instance get_upack: !$\fuClass$! (IList_u !$\to$! Z !$\to$! Z) :=
  { !$\relnmUClass$! := get_rel; !$\functrelUClass$! := get_funct }.
\end{rocqcode}
\Description{Declarations in the translation of get}
\caption{Declarations in the translation of \hask{get}}
\label{fig:trans-get}
\end{listing}

\begin{figure}
  \begin{tcolorbox}
    $\trLhdeclCq{\lhdefrefl f {\lhref_f} \lhexp} \doteq$
    \begin{rocqcode}
Definition !$f$! !$(x_i: \trLhrefCqEmp{\lhref_i})_{i \leq n}$!: !$\trLhrefCqEmp\lhref$!.
  destruct (!$x_i$! as [!$x_i$! !$\wit{x_i}$!])!$_{i \in I}$!. !$\trLhexpCq{\lhexp}$! Defined.
Inductive !$\relnm{f}$!: !$\utrLhrefCq{\lhref_1} \to \dots \to \utrLhrefCq{\lhref_n} \to \utrLhrefCq\lhref \to \Prop$! :=
  !$\textit{if}$! !$\keyw{reflect}$! !$\textit{then}$! !$\left(\mid f^\branch: \altInGraphRel(f, \branch)\right)_{\branch \in \mkGraph(\lhexp; x_1, \dots, x_n)}$! !$\textit{else}$! !$\ltacTac{\getRel\ap f}$!.
Lemma !$\functrel{f}$! !$(x_i: \utrLhrefCq{\lhref_i})_{i \leq n} \ap (v\,v':\utrLhrefCq{\lhref})$!: 
  !$\relnm{f} \ap x_1 \dots x_n \ap v \to \relnm{f} \ap x_1 \dots x_n \ap v' \to v = v'$!.
Lemma !$\unrefproof{f}$! !$(x_i:\trLhrefCqEmp{\lhref_i})_{i \leq n}\ap (v:\utrLhrefCq\lhref)$!: !$\proj{f\ap(x_i)_{i \leq n}} = v \lequiv \relnm{f}\ap (\proj{{x}_i})_{i \leq n}\ap v$!.
Instance !$\fInst{f}$! : !$\trLhrefCqEmp{\lhref_f}$! :=
  { !$\defnmClass$! := !$f$!; !$\relnmClass$! := !$\relnm{f}$!; !$\unrefproofClass$! := !$\unrefproof{f}$!; !$\functrelClass$! := !$\functrel{f}$! }.
Instance !$\fuInst{f}$! : !$\utrLhrefCq{\lhref_f}$! :=
  { !$\relnmUClass$! := !$\relnm{f}$!; !$\functrelUClass$! := !$\functrel{f}$! }.
  \end{rocqcode}
  {\small
    \textit{where} $\lhref_f = \arrs$,
    $f^{\branch}$ is a unique identifier, 
    and $I = \{i \in [1..n] \mid \lhref_i \text{ is not an arrow}\}$.
  }
  \end{tcolorbox}
  \Description{Translation of function declarations}
  \caption{Translation of a function declaration}
  \label{fig:transfunc}
\end{figure}

The definition of the unrefined graph relation depends on whether $f$ is reflected.
If $f$ is not, we disallow computation on the relation: the translation builds a
generic, opaque graph relation (\cref{subsec:unreffunc}) with the tactic~\getRel.
If $f$ is reflected, its definition is encoded in the inductive $\relnm f$.

To build the inductive graph, the function $\mkGraph(\lhexp; x_1, \dots, x_n)$
collects all the \emph{function paths} of $f$, excluding the $\dead$ ones.
For each path $\branch$, we map the unique constructor $f^\branch$
to the path relation, given by $\altInGraphRel(f, \branch)$,
defined in~\ref{def:altInGraphRel}.

The original \lh program can also have top-level declarations of constants ($0$-ary functions).
Their refined definition is translated in the same way as functions, with tactics.
But we do not generate a graph relation or a pack, and simply use the dependent pair projection if
a constant is used in an unrefined context.

\subsection{Unrefined Translations of Types and Refinements}\label{sec:unrefined}

We present the translation of types ($\typesToUnRefName$) and of refinements to the unrefined world.
Refinements are either type refinements, in which case they are translated as
$\Prop$ terms (\cref{def:utrLhtmCqprop}),
or refinement expressions, and translated within the unrefined graph relation of
the function (\cref{def:altInGraphRel}).

\begin{definition}[Unrefined translation of types]
  \label{def:utransreftype}%
  A refinement type is translated as unrefined by scrapping refinements,
  and using the unrefined inductive datatypes.
  An arrow type becomes an unrefined pack signature.
  We translate the fixed length $\Int$ type to $\cqInt$.
  Overflow and underflow on \haskell's $\Int$ operations are undefined behavior, so
  arbitrary precision arithmetic on $\cqInt$ respects these semantics.
  \[ \begin{array}{c}
    \utrLhrefCq{\refT x \lhbuiltin \lhtm} \doteq \trBuiltin\lhbuiltin
    \quad\utrLhrefCq{\refT x \lhTC \lhtm} \doteq \unrefnm\lhTC
    \\\utrLhrefCq{\arrs[x_i][\lhref_i][1 \leq i \leq n][\lhref]} \doteq
    \fuClass \ap (\utrLhrefCq{\lhref_1} \to \dots \to \utrLhrefCq{\lhref_n} \to \utrLhrefCq\lhref)
    \\[0.2em]
    \textit{where} \quad
    \trBuiltin{\Bool} \doteq \cqBool \quad
    \trBuiltin{\Int} \doteq \cqInt \quad
    \trBuiltin{\Unit} \doteq \cqUnit
  \end{array} \]
\end{definition}

\subsubsection{Translation of refinements}
\label{sec:RToProp}

Our example for the translation of type refinements is the refinement of the
parameter \hask{i} of the function \hask{get},
of elaborated type (omitting annotations on variables):
\hask?{i:Int | 0 !$\leq$! i !$\land$! i < llen (proj xs)}?
and its translation:
\begin{rocqcode}
{i:Z | !$\lebZRel$! 0 i true !$\land$! !$\exists$! llen_xs, llen_rel !$\projleft$!xs!$\projright$! llen_xs !$\land$! !$\ltbZRel$! i llen_xs true}
\end{rocqcode}
The translation of refinements is complicated by our encoding of functions as relations.
Function calls must generally be translated to a call to the corresponding relation,
substituted by a free variable standing for the result (\rocq{llen_xs})
and extracted as an hypothesis
(\rocq{!$\exists$! llen_xs, llen_rel !$\projleft$!xs!$\projright$! llen_xs}).

\begin{example}
  The translation of the refinement
  $0 \leq i \land i < \texttt{llen} (\texttt{proj}\ap\texttt{xs})$ to $\Prop$ is as follows.
  First, external logical connectives are translated to corresponding ones in $\Prop$:
  $\utrLhtmCqprop{0 \leq i \land i < \texttt{llen}\ap(\texttt{proj}\ap\texttt{xs})}
  = \utrLhtmCqprop{0 \leq i} \land \utrLhtmCqprop{i < \texttt{llen}\ap(\texttt{proj}\ap\texttt{xs})}$.
  We translate builtin operators as relations:
  $\utrLhtmCqprop{0 \leq i} =\ \lebZRel\ap\utrLhtmCq{0}\ap\utrLhtmCq{i}\ap \cqtrue =\ \lebZRel\ap 0 \ap i\ap \cqtrue$,
  where the function $\utrLhtmCq\lhtm$ translates a refinement $\lhtm$ to a \coq
  term refining an unrefined base type.
  There is no need to introduce a fresh variable here.

  A function or operator call that is not directly under logical connectives
  must be extracted as a hypothesis (as the relation is $\Prop$-kinded) and substituted by a fresh variable name.
  We create a hypothesis $h = \exists\, \texttt{llen\_xs}, \texttt{llen}\ap\proj{\texttt{xs}}\ap\texttt{llen\_xs}$.
  Then:
  \begin{align*}
    \utrLhtmCqprop{i < \texttt{llen}\ap(\texttt{proj}\ap\texttt{xs})}
    &= h \to\ \ltbZRel\ap \utrLhtmCq{i}\ap \utrLhtmCq{\texttt{llen}\ap\proj{\texttt{xs}}} \ap \cqtrue
    \\&= \exists\, \texttt{llen\_xs}, \texttt{llen}\ap\proj{\texttt{xs}}\ap\texttt{llen\_xs}
    \to\ \ltbZRel\ap i\ap \texttt{llen\_xs} \ap \cqtrue
  \end{align*}
  To handle the names, we create a map $\frApp$ (for \emph{result variables})
  associating each function call (outside of a projection) to a fresh name.
  Here, $\frApp \supseteq \{\texttt{llen}\ap(\texttt{proj}\ap\texttt{xs}) \mapsto \texttt{llen\_xs}\}$.
  The function $\utrLhtmCqName$ translates a function call to its name in $\frApp$.
\end{example}

\begin{definition}
  \label{def:utrLhtmCqprop}%
  The function $\utrLhtmCqprop\lhtm$ inputs a refinement and assumes a map
  $\frApp$ from fully applied functions and operators in $\er$ to fresh variable names.

  We consider a function $\hypFromAmQ\lhtm$ that given a quantifier
  $\hypQ \in \{\forall, \exists\}$ returns the sequence hypotheses needed for $\lhtm$, using $\frApp$
  (given in \appendixref{def:hypFromAm}).
  For instance, $\hypFromAmQ[\exists] {i < \texttt{llen}\ap(\texttt{proj}\ap\texttt{xs})}
  = \exists \texttt{llen\_res}, \texttt{llen}\ap\proj{\texttt{xs}}\ap\texttt{llen\_res}$.
  We assume the hypotheses to be returned in the correct scope order.
  Given refinements $\lhtm_1, \dots, \lhtm_n$,
  we define $\hypFromAmQ[\exists] {\lhtm_1, \dots, \lhtm_n}
  \doteq \bigwedge_{i \leq n} (\bigwedge_{h \in \hypFromAmQ[\exists] {\lhtm_i}} h)$.
  Then,
  \begin{align*}
    \utrLhtmCqprop{\lhnegb\lhtm'}
    &\doteq \cqnegp\utrLhtmCqprop{\lhtm'}
    \quad\qquad\utrLhtmCqprop{\lhtm_1 \lhbop \lhtm_2}
    \doteq \utrLhtmCqprop{\lhtm_1} \lhbop \utrLhtmCqprop{\lhtm_2}
    &\lhop \in \{\lhandb, \lhorb, \lhimp, \lhiff\}
    \\ \utrLhtmCqprop{f \ap \lhtm_1 \dots \ap \lhtm_n}
    &\doteq \hypFromAmQ[\exists] {\lhtm_1, \dots, \lhtm_n} \cqandp
    \relnm f \ap \utrLhtmCq{\lhtm_1} \dots \utrLhtmCq{\lhtm_n} \ap \cqtrue
    \\ \utrLhtmCqprop\lhtm
    &\doteq \hypFromAmQ[\exists] \lhtm \land \IsTrue(\utrLhtmCq\lhtm)
    &\text{otherwise}
  \end{align*}
  We use \coq's $\IsTrue$ to cast a Boolean to $\Prop$.
\end{definition}

While $\refToPropName$ translates refinements to $\Prop$ at the external level,
inner terms are translated as $\lambda$-terms of unrefined base and inductive types
with function~$\utrLhtmCqName$.
Most notably, this function substitutes function and operator calls with their
associated result variable in $\frApp$.
\begin{definition}[Translation of refinements as terms]
  \label{def:utranstm}%
  The function $\utrLhtmCq\lhtm$ inputs a refinement and assumes a map $\frApp$.
  \[ \begin{array}{r@{\;}c@{\;}l@{\qquad}r@{\;}c@{\;}l}
    \multicolumn{6}{c}{
    \utrLhtmCq{\vararloc x n} \doteq x
    \qquad\utrLhtmCq{\varartop x {n>0}} \doteq \fuInst x
    \qquad\utrLhtmCq{\varartop x 0} = \proj{\varartop x 0}}
    \\[3pt]
    \utrLhtmCq{\lhlit} &\doteq& \utrLit{\lhlit}
    &\utrLhtmCq{\lhtm_1 \lheqb \lhtm_2}
    &\doteq& \utrLhtmCq{\lhtm_1} \cqeqb \utrLhtmCq{\lhtm_2}
    \\[3pt]
    \utrLhtmCq{\lhConstr \ap \seq\lhtm }
    &\doteq& \unrefnm\lhConstr \ap \seq{\utrLhtmCq\lhtm}
    &\utrLhtmCq{\lhproj\lhtm} &\doteq& \proj{\trLhtmCq\lhtm}
    \\[3pt]
    \utrLhtmCq{\lhtm}
    &\doteq& \multicolumn{4}{l}{\utrLhtmCq{\lhtm'},
      \quad \textit{if } \lhtm \in
      \{ \lhinj {\lhtm'} \lhref, \lhsub {\lhtm'} {\lhref_s} {\lhref_t} \}
    }
    \\[3pt]
    \utrLhtmCq{\lhtm}
    &\doteq& \multicolumn{4}{l}{\frApp(\lhtm),
      \quad \textit{if } \lhtm \in \{
        \lhtm_1 \lhbop \lhtm_2,
        \varar{f}{n>0} \ap\lhtm_1\dots\lhtm_n,
        \lhproj{\varar{f}{n>0}} \ap\lhtm_1\dots\lhtm_n 
      \}
    }
  \end{array}\]
  \textit{where}
  $\utrLit\true \doteq \cqtrue$,
  $\utrLit\false \doteq \cqfalse$,
  $\utrLit\unit \doteq \cqunit$, and 
  $\utrLit\lhlit \doteq \lhlit$ for integers.
\end{definition}
\noindent
Note the difference of treatment between different kinds of variables:
local variables are kept as is, while top-level functions, used as values,
are mapped to their $\fuClass$ instance, and top-level constants are projected.
We use our generic Boolean equality $\cqeqb$ directly, while other operators
are extracted, as any function call.
Within a projection $\lhproj\lhtm$ (that only occurs in type refinements),
the term $\lhtm$ is translated as refined, then projected.
Injections and subsumptions are ignored.

\subsubsection{Construction of the graph relation}
\label{sec:unref-graph}

We build the inductive graph relation $\relnm f$ for a function $f$ from the
sequence of \emph{function paths} obtained from its definition.
For instance, let $\lhexp_\texttt{get}$ be the body of \hask{get}.
The graph of \hask{get}, computed by $\mkGraph(\lhexp_\texttt{get}; \texttt{xs}, \texttt{i})$, has two paths:
\begin{center}
  ([\hask{xs !$\mapsto$! Cons_gt5 x xs'}, \hask{i !$\mapsto$! i}], [\hask{i == 0 !$\mapsto$! True}], \hask{x}) and
  \\([\hask{xs !$\mapsto$! Cons_gt5 x xs'}, \hask{i !$\mapsto$! i}], [\hask{i == 0 !$\mapsto$! False}], \hask{get xs' (i-1)})
\end{center}

Formally, a function path is a triple $(\sigma_\params, \sigma, \lhtm)$
of a map from parameters to their pattern, additional matches and a term.
By construction, the last term is always a refinement.
The function $\mkGraph$ is defined in \appendixref{def:functionpaths}.
It essentially reorganizes and collapses all pattern matching, substitutes away
let-bindings and drops unreachable branches and proof hints.

Recall the definition of \rocq{get_rel}:
\begin{rocqcode}
Inductive get_rel : IList_u !$\to$! Z !$\to$! Z !$\to$! Prop :=
| get_Cons_true x xs' i : i == 0 !$\to$! get_rel (Cons_gt5_u x xs') i x
| get_Cons_false x xs' i : i /= 0 !$\to$! !$\forall$! (im1:Z), subZ_rel i 1 im1
   !$\to$! !$\forall$! (v:Z), get_rel xs' im1 n !$\to$! get_rel (Cons_gt5_u x xs') i v.
\end{rocqcode}
\noindent The translation uses the functions $\utrLhtmCqName$ and $\hypFromAmName_\forall$,
this time using a universal quantifier and implications to introduce the
hypotheses as premises.
For each path, we instantiate \rocq{get_rel} with the parameters patterns
(here \rocq{Cons_gt5_u x xs'} and \rocq{i} in both branches).
We also introduce the free variables of these terms (\rocq{x}, \rocq{xs'} and \rocq{i}).
Additional pattern matchings are encoded as hypotheses (\rocq{i == 0} and \rocq{i /= 0}).
Other hypotheses arise from the extraction of function calls, in the second
branch for subtraction and the recursive call to \hask{get}.

\begin{definition}
  \label{def:altInGraphRel}%
  We translate a path of a function $f$ to the type of a constructor of $\relnm f$
  (as in \cref{fig:transfunc}) with function $\altInGraphRel$ below, where
  $(h_1, \dots, h_n) \doteq \hypFromAmQ[\forall]\lhtm$, resp. $\hypFromAmQ[\forall]{\lhtm_m}$.
  \[\begin{array}{l@{\;}c@{\;}l}
    \multicolumn{3}{l}{
      \altInGraphRel(f, (\sigma_\params, \sigma, \lhtm))
      \doteq \forall (y)_{y \in \fv(\codom(\sigma_\params))},
      \texttt{go}(\sigma)
    }\\
  \textsf{where }
  \texttt{go}(\emptyset) &\doteq& h_1 \to \dots \to h_n
    \to \relnm{f} \ap (\utrLhtmCq{\lhtm_i})_{\lhtm_i \in \codom(\sigma_\params)} \ap \utrLhtmCq\lhtm
    \\\phantom{\textsf{where }}\texttt{go}((\lhtm_m, \lhtm_p)\colon\!\sigma)
      &\doteq& h_1 \to \dots \to h_n \to
      \forall (y)_{y \in \fv(\lhtm_p)},
      \utrLhtmCq{\lhtm_m} == \utrLhtmCq{\lhtm_p}
      \to \texttt{go}(\sigma)
\end{array}\]
\end{definition}
The function is defined inductively on the elements of $\sigma$.
In the inductive case, we introduce an equality between the matched term
and its pattern.
We use $\hypFromAmQ[\forall]{\lhtm_m}$ but not $\hypFromAmQ[\forall]{\lhtm_p}$
as the pattern $\lhtm_p$ does not contain any function calls to extract.
We shorten \rocq{(i ==? 0) == true} to \rocq{i == 0}
and \rocq{(i ==? 0) == false} to \rocq{i /= 0}.

\subsection{Refined Translations of Types, Refinements, and Expressions}\label{sec:refined}

We now introduce the translation of types ($\typesToRefName$), refinements
($\trLhCqName$) and expressions ($\exprToTacName$) into the refined world.

\begin{definition}[Refined translation of types]
  \label{def:trtypes}%
  Type refinements are translated to $\Prop$ with $\refToPropName$.
  Arrows become refined pack signatures.
  We reuse the function $\trBuiltin\lhbuiltin$ from \cref{def:utransreftype}.
  \[ \begin{array}{r@{\;}c@{\;}l}
    \trLhrefCq{\refT x \lhbuiltin \lhtm}
      &\doteq& \refT x {\trBuiltin\lhbuiltin} {\utrLhtmCqprop\lhtm}
      \\ \trLhrefCq{\refT x \lhTC \lhtm}
      &\doteq& \refT x {\unrefnm\lhTC} {\wffnm\lhTC x \land \utrLhtmCqprop\lhtm}
      \\ \trLhrefCq{\arrs}
      &\doteq&\ \fClass\ (\forall (\bind{x_i}{\trLhrefCq{\lhref_i}})_{i \leq n}, \trLhrefCq\lhref)
  \end{array} \]
\end{definition}

\subsubsection{Translation of expressions and refinements}
\label{sec:transexpr}%

The translation of \lh terms to tactics is cumbersome;
instead of giving the full translation (available in \appendixref{app:refined}),
we illustrate it on a variation of \hask{get}.
The variation introduces a let-bound function \hask{double} that is applied to the return value.
We also show how \lh uses theorems as a hint to the SMT solver.
The theorem \hask{doubleThm} is here for the sake of illustration, in practice it is too trivial to be necessary.

The elaborated $\er$ definition is in \cref{fig:getDouble}.
We write $\uparrow r$ for $\lhinj{\lhtm}{\refT v \lhtyp \True}$ where
$\lhtyp$ is the base type of $\lhtm$.
We use the same notation for a trivial injection translation in \coq.
\begin{listing}
  \begin{lrcode}
def doubleThm :: x:{Int | 5 < x} !$\to$! {u:() | 10 < x + x} =
  inj ((), {u:() | 10 < x + x})

def getDouble :: xs:IList !$\to$! i:{Int | 0 !$\leq$! i !$\land$! i < llen (proj xs)} !$\to$! {v:Int | 10 < v} =
  let double :: (y:Int) !$\to$! {v:Int | proj y + proj y == v} = !$\uparrow$!y + !$\uparrow$!y
   in case xs of
   | Cons_gt5 x xs' !$\to$!
       if !$\uparrow$!i == !$\uparrow$!0 then sub (double !$\uparrow$!x, {v:Int | x + x == v}, {v:Int | 10 < x})
                        ? (doubleThm (inj(x,{x:Int | 5 < x)) proves 10 < x + x)
                   else getDouble !$\uparrow$!xs' (!$\uparrow$!i - !$\uparrow$!1)
   | Nil !$\to$! unreachable
  \end{lrcode}
  \caption{Elaborated functions \hask{doubleThm} and \hask{getDouble}}
  \label{fig:getDouble}
\end{listing}
The type of \hask{double} is taken from \lh's type inference (during parsing to $\er)$.
\lh's addition and other operators are refined:
\begin{lrcode}
+ :: x:{Int | True} !$\to$! y:{Int | True} !$\to$! {v:Int | proj x + proj y = v}.
\end{lrcode}
\noindent Thus, \hask{y} must be injected into \hask{double}, but we can derive the final
refinement of \hask{getDouble}, using \hask{doubleThm}.
In \coq, we use a custom refined addition \rocq{+Z}.

The translation is given in \cref{fig:trans-getDouble}.
For \hask{doubleThm}, it is similar to \hask{thm_get} (\cref{subsec:overview-thm}),
but with the injection on the unit and the subsumption merged into a single injection.
\begin{listing}
  \begin{rocqcode}
Definition doubleThm (x:Z | !$\ltbZRel$! 5 x true)
  : {u:unit | !$\exists$! p, !$\addZRel$! x x p !$\land$! !$\ltbZRel$! 10 p true}.
  destruct x as [x x_p]. exact (exist (!$\forall$!(u:unit), !$\exists$! p, ...) tt ltac:(solver)). Qed.

Definition getDouble (xs:IList)
  (i:Z | !$\lebZRel$! 0 i true !$\land$! !$\exists$! llen_xs, llen_rel !$\projleft$!xs!$\projright$! llen_xs !$\land$! !$\ltbZRel$! i llen_xs true)
  : {v:Z | !$\ltbZRel$! 10 v true}.
  destruct xs as [xs xs_p]. destruct i as [i i_p].
  (* Translation of let double y = y + y *)
  assert (double_def : !$\forall$!(y:Z | True), {v:Z | !$\addZRel$! !$\projleft$!y!$\projright$! !$\projleft$!y!$\projright$! v}).
  { intros [y y_p]. exact (!$\uparrow$!y +Z !$\uparrow$!y). } (* translation of !$\lambda$!y. !$\uparrow$!y + !$\uparrow$!y *)
  refine (let double: Pack (!$\forall$!(y:Z | True), {v:Z | !$\addZRel$! !$\projleft$!y!$\projright$! !$\projleft$!y!$\projright$! v})
           := ltac:(fun_to_pack double_def) in _).
  (* Translation of the case expression *)
  generalize dependent i; induction xs as [ | x xs' IH_xs' ]; intros.
  - exfalso; solver.
  - destruct (i ==? 0) eqn:?.
    + assert (!$\ltbZRel$! 10 x true) by exact (proj2_sig (getDouble (exist _ x ltac:(solver)))).
      exact (subCast _ _ (double.(def) !$\uparrow$!x) ltac:(solver)).
    + exact (IH_xs' ltac:(solver) (i-1) ltac:(solver)).
Defined.
  \end{rocqcode}
  \Description{Translation of getDouble}
  \caption{Translation of \hask{getDouble}}
  \label{fig:trans-getDouble}
\end{listing}

\paragraph{Translation of expressions}

The preliminary \rocq{destruct}s decompose \rocq{xs} and \rocq{i} into unrefined
values and refinement proofs.
This decomposition is always done on dependent pairs, and appears in \cref{fig:transfunc}.

Then, we introduce \rocq{double_def}, the refined version of \hask{double}.
We use a tactic \rocq{assert} to add \rocq{double_def} in the proof context,
its definition is created with the tactics between brackets.
From \rocq{double_def}, we create the whole pack by instantiating a relation from
the refined definition of \rocq{double}.
The tactic \rocq{refine} is like \rocq{exact} but accepts holes to be filled later.
Later uses of \rocq{double} use the pack, like \rocq{double.(def)} for the \haskell call \hask{double x}.
In general, a local function might be passed around as a first-class value,
for instance given as argument to a higher-order function, which makes the creation of a pack necessary.

If the value bound by a let has a subset type, the translation is different.
We still use \rocq{assert}, but instead of creating a pack, we decompose the
dependent pair once created.

The translation of the pattern matching, \hask{Nil} branch and conditional are as described in \cref{subsec:reffunc}.
The first branch shows how a proof hint gets translated: we assert the translated proposition.
The evidence is extracted directly as the second element of the dependent pair
returned by the application of \rocq{getDouble}.
Formally, \[
  \trLhexpCq{\annlhqmark\lhtm{\lhtm_h}{\lhtm_p}} \doteq
  \tac{assert}\ (\utrLhtmCqprop{\lhtm_p})\ \tac{by}\ \tac{exact}\ (\projTwo{\trLhtmCq{\lhtm_h}}).\ \trLhexpCq{\lhtm}
\]

\paragraph{Translation of refinements}

The tactic \rocq{exact} is the interface between the translation of structural
expressions as tactics and as refinements: the $\lambda$-terms
produced by the translation $\utransName$ are inserted in the proof script with \rocq{exact}.
Formally, for any $\lhtm$, $\trLhexpCq\lhtm \doteq \tac{exact} (\trLhtmCq\lhtm)$.

The translation of refinements is mostly straightforward
and follows the structure of the terms.
The main difficulty is the translation of recursive calls.

The recursion \hask{get xs (i-1)}
is translated as a call to the induction hypothesis \rocq{IH_xs'} introduced by the
\rocq{induction} tactic.
This hypothesis is already instantiated for \rocq{xs'}, so we apply it to \rocq{i-1}
(obtained by simplification of \rocq{!$\projleft$! !$\uparrow$!i -Z !$\uparrow$!1 !$\projright$!}).
It is necessary to generalize \rocq{i} before calling \rocq{induction};
otherwise, \rocq{IH_xs'} gets instantiated to \rocq{i} as well.

The induction hypothesis is created in a context in which the dependent pairs of
the arguments are split, so \rocq{IH_xs'} is applied to refinement proofs along with the unrefined values:
the first proof obligation is for \rocq{xs'} and the second for \rocq{i-1}; both are discharged automatically.
What makes the translation of recursive calls complex in general is the necessity to
\textit{(1)} link a recursive call with the correct induction hypothesis and
\textit{(2)} find the correct subset of arguments that still need to be instantiated in that induction hypothesis.

Non-recursive calls $f\ap\lhtm_1\dots\lhtm_n$ are translated recursively
on $\lhtm_1$ to $\lhtm_n$.
We use the name $\defnmInst f$ if $f$ is a pack (that is, $f$ is local, like
\rocq{double}) and $f$ for a top-level function (like \rocq{doubleThm}).

\paragraph{Subsumption casts}\label{sec:casts}

Subtyping annotations are translated to a custom \rocq{subCast} operation.
Its arguments are the original refinement, the desired refinement, the term to cast
and the proof of the implication, given by \rocq{solver}.
For subset types, we define it as follows:
\begin{minted}{coq}
  subCast (!$\forall$!(!$x_s$!:!$\lhtyp$!), !$p_s$!) (!$\forall$!(!$x_t$!:!$\lhtyp$!), !$p_t$!) !$e$!
  !$\doteq$! exist (!$\forall$!(!$x_t$!:!$\lhtyp$!), !$p_t$!) (proj1_sig !$e$!) ltac:(pose proof (proj2_sig !$e$!); solver)
\end{minted}
The cast projects $e$, then injects it back to the expected refined type.
The solver proves the verification condition: assuming the refinement of $e$ (\rocq{proj2_sig !$e$!}),
it creates a proof of $p_t \subst{x_t}{\texttt{proj1\_sig}\ap e}$.

Subsumption on packs separately casts the four elements of the pack:
two unrefined (the graph relation and its functionhood lemma)
and two refined (the definition and the lemma relating the definition and the graph relation).
The first two are unaffected by the cast, being unrefined.
The lemma relating the definition and the graph relation is proven again
using the original lemma.
The refined definition is morally translated by casting the arguments
contravariantly and the function body covariantly (\textit{per}~\cite{sekiyama17}).
The actual implementation, described in~\appendixref{app:pack}, is more complex because of the representation
of function arguments as dependent lists inside packs.

\section{Properties of the Translation}\label{sec:properties}
Our translation is sound,
in the sense that it maps a proof of false in \lh to a proof of false in \coq.
That is, \lava lets \lh use \coq as a proof checker, via translation.
This check increases our confidence in \lh, but our
tool is not complete in the sense
that a well typed \lh program necessarily yields a well-typed \coq program
with provable proof obligations that \solver can resolve.
Such a proof would have to relate an SMT solver to our translation's encoding;
a very ambitious goal! Instead, we evaluate empirically (\cref{sec:eval}).

\subsubsection*{Preservation of evaluation semantics of expressions}

To show that proofs of false are preserved, we first show that our translation preserves the input/output behavior of declared functions\,---\,both
in the graph relations of the unrefined world and in the tactic-generated executable terms of the refined world.
We introduce a conventional call-by-value reduction relation
$\rew\er$ \appendixrefp{def:ssred} for expressions as a small-step operational
semantics, parametrized by a sequence of declarations~$\lhdecls$.

\begin{restatable}{theorem}{operationalmain}
  \label{l:operationalmain}
  Let $\lhdecls$ be a well-formed and elaborated sequence of declarations,
  with all functions reflected.
  Let $\lhrefl f \arrs {\lhexp_f} \in \lhdecls$.
  Let $\typEnv$ be the environment resulting from the elaboration into $\lhdecls$.
  Let $\lhval_1, \dots, \lhval_n, \lhval$ be values (variables, literals and constructors)
  such that $\hchkj\typEnv{\lhval'}\lhref\lhval$ and
  for all $i \leq n$, $\hchkj\typEnv{\lhval_i'}{\lhref_i}{\lhval_i}$
  for some $\lhval'$ and $\lhval_i'$. Then,
  $\ssredn{f \ap \lhval_1 \dots \lhval_n}{\lhval}$ if and only if
  in \coq and under $\trLhdeclCq\lhdecls$:
  $\relnm f \ap \utrLhtmCq{\lhval_1} \dots \utrLhtmCq{\lhval_n} \ap \utrLhtmCq\lhval$ holds.

  Moreover, if each $\trLhtmCq{\lhval_i}$ is of type $\trLhrefCq{\lhref_i}$ and
  $\trLhtmCq\lhval$ is of type $\trLhrefCq\lhref$ in \coq,
  the two properties above are equivalent to
  $\proj{\trLhtmCq{f \ap \lhval_1 \dots \lhval_n}} = \utrLhtmCq\lhval$.
\end{restatable}
\begin{proof}[Proof sketch]
  By induction.
  For the additional property, we use the coherence lemma $\unrefproof f$ (\cref{sec:transfunc})
  from the translation output.
  Since we assume that the output typechecks, this lemma is verified.
  We must assume that the values are of the right translated refined type in
  \coq because our type system does not verify refinements.
  The full proof is in \appendixref{sec:app-operational}.
\end{proof}

\noindent
We restrict this theorem to the case where functions is reflected
(any function can be reflected so long as it does not let-bind higher-order variables).
For unreflected functions, the graph relation is defined in terms of
the refined function, itself defined with tactics.
Any meta-theoretical proof depending on the output of these tactics would be challenging.

\subsubsection*{Preservation of the models of types}

We show that our translation
preserves the semantics of types.
To do so, we interpret both $\er$ and \coq types in the same model and show that
the translation preserves the type's model.
For \coq types, we only need to define models for types in the image of the translation, that is,
instances of the typeclass $\fClass$,
as well as
refinement (or subset) types of the form $\forall (x : \cqref), \cqref$ and $\refT x \cqtyp \cqprop$, where
$\cqtyp$ is a base type and $\cqprop$ a predicate coming from the
translation of a refinement as in \cref{def:utrLhtmCqprop}; we have
$\cqprop\Rightarrow\wffnm\lhTC \ap x$ if $\cqtyp = \unrefnm\lhTC$.

We give the standard interpretation for refinement types \cite[\S~5]{knowles10}.
We interpret them as the set of terms $t$ such that $t$ is well-typed in the
underlying simply-typed system, and in the case of a refined base type with
refinement $\lhtm$ or $\cqprop$, such that $\lhtm$ or $\cqprop$ instantiated by $t$ is true.

We write $\lhtpint\cdot$ for the interpretation of $\er$ types
and $\cqtpint\cdot$ for the interpretation of \coq types.
The $\er$ declaration environment $\lhdecls$ gives datatype and function definition
and is used for the reduction relation, as in \cref{l:operationalmain} and for typechecking.
The assignment $\rho$ maps variables to $\lambda$-calculus values.
The complete definitions and proofs are in \appendixref{app:models}.

\begin{restatable}{theorem}{samemodel}
  \label{l:same-model}%
  Let $\lhref$ be a refinement type in $\er$, $\lhdecls$ be a $\er$ declaration
  environment such that $\lhref$ is well-formed under $\lhdecls$.
  For any assignment $\rho$, $\lhtpint\lhref = \cqtpint{\trLhrefCq\lhref}$.
\end{restatable}

\Cref{l:operationalmain,l:same-model} state that the translation
preserves the input-output relation
at the level of values and at the level of types (including refinements).
In particular, \lh proofs are functions of unit type whose output is trivial,
but the statements proved will be equivalent in \coq.

\subsubsection*{Unprovability of contradiction}

Finally, we confirm that our translation can act as a certificate
generator for \lh.
If we manage to prove false in \lh, we want to make sure that typechecking will
not go through in the \coq output (assuming \coq's consistency).

\begin{theorem}[Consistency]
  \label{lem:consistency}%
  $\nexists \lhexp.\ \emptyset \models_\text{Rocq} \trLhdeclCq{\lhdef f {\refT x \Unit \false} \lhexp}$.
\end{theorem}
\begin{proof}
  Suppose there is a well-formedness derivation for $f$ (in the empty context).
  Then, $\lhexp$ can be translated into a definition of type $\refT x \Unit \false$.
  This would give rise to a proof of $\false$\,---\,a contradiction, because Rocq's type theory is consistent.
\end{proof}

\section{Implementation \& the Proof Automation Tactics}\label{sec:impl}\label{subsubsec:oracle}
The implementation of the translation (\cref{sec:translation}) is done in $\sim$\hsLines{}K \sloc of \haskell.
\lava takes as input a program typechecked by \lh and outputs a \coq file that depends on
a) a \coq library ($\sim$350 \sloc) containing, for example, definitions for the pack typeclasses and
b) proof automation tactics ($\sim$6K \sloc).
In the rest of the section, we describe the development of the tactics.

The \solver tactic discharges the proof obligations introduced by type casts.
\Cref{fig:solver} describes the implementation of \solver in four levels;
each level composes a \emph{preprocessor tactic} and a \emph{solver definition} into a \emph{solving tactic}.
Lower level correspond to lower complexity, with tactics of higher level
subsuming and making use of lower level ones.

\begin{table}
  \begin{tabular}{|c|c|c|c|}
    \hline
        & Solving Tactic & Preprocessing Tactic & Solver Definition\\
        \hline
    Level 1 & $\solveZero$ & $\idTac$ & $\fastDone$  \\
    \hline
    Level 2 & $\quickSolve$ & $\quickSimpl$ & solve repeat $\simplLoop$  \\
    \hline
    Level 3 & $\oracle$ & $\cleanup$ &   $\finish$  \\
    \hline
    Level 4 & $\solver$ & $\preprocessor$ & first [$\quickSolve$ | shelve]  \\
    \hline
  \end{tabular}
  \caption{Solving and preprocessing tactics. Higher level tactics use lower level ones.}
  \label{fig:solver}
\end{table}

The \emph{preprocessor tactics} try to simplify proof obligations,
only succeeding if provability of the goals is preserved by the transformations.
The simpler kind unfolds, applies, and rewrites by lemmas and definitions from the
library or the generated files.
The more complex kind creates and discharges subgoals, relying on solving tactics
from lower levels.
The preprocessor tactics listed in \cref{fig:solver} combine many simpler preprocessor tactics, each
corresponding to a single reasoning step; \cleanup and \preprocessor can invoke
dozens of smaller steps in each successful call.

The \emph{solving tactics} fail if they are unable to discharge the goal.
They culminate in \solver, which uses most of the other definitions and tactics in the library.
Unlike the other solving tactics, \solver shelves subgoals it cannot prove,
which is useful to inspect incomplete proof scripts manually.
The class of solving tactics also contains numerous tactics for proving the lemmas generated by the translation;
they reuse many of the definitions also used by \solver.

Let us give more details on each individual level.

On \emph{level 1}, \tac{idtac} is a noop; \fastDone uses \tac{lia} and
\tac{auto} (with a hint library).

On \emph{level 2}, \quickSimpl uses simple rewriting lemmas to clean up the goal.
The \simplLoop tactic mostly relies on rewriting, unfolding and solving with hints, simplifying with
basic \coq tactics (\eg \tac{simpl}, \tac{split}).
It runs rule-based simplification steps (1-5 \sloc each) to simplify specific syntactic forms like quantifiers, logical connectives, or equalities.

On \emph{level 3}, \cleanup uses \quickSimpl and the lemmas relating refined functions and their graph relations.
It works extensively with graph relations, by mapping refined terms to their
corresponding graph terms, reasoning about result variables, and saturating the context by reconstructing
the refinement witnesses of the final outputs of graph relations.

On \emph{level 4}, \preprocessor extends \cleanup with extra simplification
steps which can require proving difficult subgoals with the \oracle tactic.

Besides \solver, the most powerful, but time-consuming, tactics are the ones relating a reflected definition $f$
to its graph relation $\relnm f$, used in particular in the coherence theorems $\unrefproof f$ (\cref{sec:transfunc}).
These tactics employ many preprocessing steps similar to those of \cleanup.
However, the proofs require additional and more aggressive preprocessing steps. 
These steps apply a forward reasoning that might in general get us stuck in the
resolution, but are safe in this context.

Could we replace our automation tactics with existing SMT or other \coq automation?
It would not be so easy:
(1) a lot of preprocessing is required to bring the proof state into a shape well supported by \coq and the tools in its ecosystem;
(2) we need to interleave forward and backwards reasoning steps matching on and modifying both goals and hypothesis; and
(3) the solving process can be very slow (see the benchmarks in~\cref{fig:eval}), so customizing automation enables us to optimize
performance for our goals.
We discuss the use of Sniper and CoqHammer specifically in \cref{sec:related}.

\section{Evaluation}\label{sec:eval}
We evaluate \lava on eight benchmarks, summarized in \cref{fig:eval}.
The first five test specific features of the translation.
The other three come from external sources: a port of \coq's Software
Foundations~\cite{SF} in \lh and the benchmarks of~\cite{vazou17}.

\paragraph{Description of the examples}
The \SafeIndexing example (\cref{sec:overview})
tests most features of \lava.
It contains a non-trivially refined inductive data type and tests the translation of subtyping
to casts in \coq, the translation of \haskell's
integer types to \coq's \texttt{Z} and of theorems with proof hints.
It also contains functions with preconditions, postconditions, reflection annotations, and higher-order parameters.

\PeanoNats evaluates the translation of
theorems to \coq.
It defines the inductive data type of Peano naturals and reflected functions for
addition, multiplication, subtraction, inequality and equality.
It also contains theorems stating
that zero is the left and right neutral element of addition,
associativity of addition,
distributivity of right-multiplication with addition,
and that $\forall m, n: (m + n) - n = m$.
This example pushes the limits of our first-order proof automation in \coq.

\RBinsToBins tests the translation of a refined datatype.
It defines the refined data type of binary numbers without leading zeros
and a function translating from it to an unrefined version of the data type.
Crucially, the refinement of one of this datatype's constructors depends on another constructor. Supporting this example motivates our translation of refined datatypes (from \cref{subsec:intlists}). 

\HigherOrder tests the translation of $\lambda$-abstractions and properties of function composition.

\Lists tests the list operations append, map, concatMap, reverse, zip, flip, zipWith, unzip and take, and various theorems relating them.

\SoftFound is a \lh port of the first section of Software Foundations~\cite{SF} and the first five definitions/theorems of the second and third sections each. 
The original \coq file is $356$~\sloc long and was ported to $437$~\sloc of \lh.

\ClassLaws tests the monad and monoid laws on List, Maybe and Identity datatypes, monomorphized for integers.
\Foldr encodes the universality of the foldr function, as described in~\cite{foldr}.
These last two examples are taken from the benchmarks of~\cite{vazou17}.
Their complete benchmarks suite also contain class laws for the functor and applicative classes
that use features, such as partial applications and higher-order datatypes, that are
unsupported by the translation.

\begin{table}
\begin{tabular}{|l|r|r|r|r|r|r|}
	\hline
	\multirow{2}{*}{Example}
	& \multicolumn{1}{c|}{\texttt{LH}}
        & \multicolumn{1}{c|}{\coq}
        & \multicolumn{1}{c|}{Blowup}
	& \multicolumn{1}{c|}{Translation }
	& \multicolumn{1}{c|}{Verification}
	& \multicolumn{1}{c|}{Verification} \\
	&
	\multicolumn{1}{c|}{\sloc}
	& \multicolumn{1}{c|}{\sloc}
        & \multicolumn{1}{c|}{~}
	& \multicolumn{1}{c|}{Time}
	& \multicolumn{1}{c|}{Time}
	& \multicolumn{1}{c|}{Time/\sloc}
	\\ \hline
	\SafeIndexing &  67 &  1\,474 & 22$\times$ & 0.04 sec & 28 min 38 sec & \verifTimeSafeIndexing \\
	\PeanoNats    &  88 &  1301 & 14.8$\times$ & 0.18 sec & 1 min 0 sec & 0.05 \\
	\RBinsToBins  &  23 &  165 & 7.2$\times$  & 0.04 sec &         7 sec & 0.04 \\
	\HigherOrder  &  22 &  350 & 15.9$\times$  & 0.12 sec &         24 sec & 0.07 \\
	\Lists & 178 & 3\,100 & 17.4$\times$ & 0.37 sec & 7 min 10 sec & 0.14 \\
	\SoftFound    & 437 & 5\,962 & 13.6$\times$  & 0.04 sec & 4 min 4 sec & \verifTimeSoftFound \\
	\ClassLaws & 151 & 1\,722 & 11.4$\times$ & 0.38 sec & 1 min 56 sec & 0.07 \\
	\Foldr  &  37 &  603 & 16.3$\times$  & 0.08 sec &   2 min      16 sec & 0.23 \\
	\hline
	Total & 1003 & 14\,677 & \blowup$\times$  & 1.25 sec & 45 min 35 sec & 0.19\\
	\hline
\end{tabular}
\caption{
  Quantitative Evaluation of \lava.
	\sloc is semantic lines of code.
        Blowup is the ratio between the LH and the generated \coq.
	Verification time is the time to verify the \coq code.
	Verification time/\sloc is the verification time (in seconds) divided by the number of \coq \sloc.
  }
\label{fig:eval}
\end{table}

\paragraph{Summary}

The translated file is \blowup~times longer than the original.
Indeed, for each original \lh declaration, we generate all the elements described in \cref{sec:transDecl}
as well as other lemmas designed to help automation and specialized for the declarations.
Moreover, type translations can be large when using packs or when numerous function calls appear in refinements.

Translation time (from $\er$ to Rocq) is fast (less than 2 seconds in total), LH verification and parsing into $\er$ is similarly fast. 
But verification of the output lasts up to 29 minutes on the \SafeIndexing example.
Indeed, verification times \emph{per line} of \coq vary widely, ranging from
\verifTimeSafeIndexing\ for \SafeIndexing to \verifTimeSoftFound\ for \SoftFound.
We find that \emph{most} of the verification is immediate, but a few hard proof
obligations take up to a few minutes.

The fact that our automation frequently solves hard goals indicate its power,
but for it to take several minutes shows there is room for optimization.
Much of the verification time in hard proofs comes from backtracking.
In those cases, manually proving the crucial steps (after automatic
preprocessing) can substantially reduce verification time.

\section{Related Work}\label{sec:related}
\paragraph{Translations from \haskell to \coq}

\texttt{hs-to-coq}~\cite{spectorzabusky18,breitner18} translates \haskell programs to \coq
and is used to verify a substantial part of \haskell's \textit{containers} library.
This tool differs from ours in two main aspects.
First, source programs lack refinement types---and so have neither dependent pattern matching nor proofs terms, simplifying the translation.
The desired specifications are either inserted manually or derived from QuickCheck annotations and proved by hand.
Second, \texttt{hs-to-coq} supports syntactic partiality of programs by inserting an opaque default value.
While we allow \haskell programs to have missing cases, we require that \lh programs are in fact total.

Our choice to use inductive predicates has its downsides: it makes support for
partial applications difficult and forces the extraction of applications to
hypotheses in the unrefined translation.
Another way to deal with partiality in the input program is to use a monadic translation, like
\citet{dylus18} and \citet{christiansen19} do for translating \haskell (inspired from \cite{abel05}),
or \citet[][Chapter~3]{claret18} and \citet{garrigue22} do for \pl{OCaml}.
Monads could be used to deal with partial \lh programs, that are partial also in the refined world.
This is not a priority because the default semantics of \lh rejects non-provably total programs.

\paragraph{Translation of programs with refinements in interactive theorem provers}

\emph{Stainless}~\cite{hamza19} is a refinement type checker for Scala.
Stainless has a prototype translation to \coq~\cite{lara24} that is not (yet)
covered in a scientific paper.
But while Stainless internally uses refinement types,
this translation does not support them nor dependent arrows.
In Stainless, refinements are written as pre- and post-conditions and assertions on functions,
and this is the specification that the translation supports.
The translated programs are effectively unrefined programs with assertions.
The primary focus of our work instead is to describe the subtle interaction between
unrefined and refined elements, in particular in the context of the embedding of
programs into the logic.
Thus, their embedding is simpler as everything is unrefined.
The translation of Stainless' programs is intended as a backend for verification,
while \lava is too slow for this goal.

Characteristic formulas are another approach of translating
from general-purpose proof assistants to logic:
\citet{chargueraud10} translates \pl{OCaml} programs to logical formulas in \coq,
instead of \coq programs.
It would be interesting to study similar approaches in future work,
in particular also whether one can translate from GOSPEL-annotated \pl{OCaml} programs~\cite{DBLP:conf/fm/ChargueraudFLP19}
using our translation.

\paragraph{Translations of refined inductive datatypes}
\citet{atkey12} define a systematic translation between refined and
inductive datatypes via categorical logic. This differs from our approach of using unrefined inductive datatypes and inductively-defined predicates instead.
\citet{sekiyama15} compare the representations of refined inductive datatypes with unrefined constructors and unrefined inductive
types with refined constructors in \coq. We use the same correspondence to translate inductive datatypes with refined constructors from \lh to inductive datatypes refined by a predicate in \coq.

\paragraph{Mechanization of \lh}

We aim to improve \lh's foundations by producing certificates for programs, using \coq for a semantics-in-translation. 
With a similar goal, \citet{borkowski24} define and formalize a type system that
models \lh and prove it sound in \coq.
However, logical implication of refinements in subtyping is modeled by a
reduction semantics, so the problem of the embedding of terms into the logic and
the relation with the SMT solver is avoided, which is the focus of our paper.
Still, \citeauthor{borkowski24} prove soundness of the type system
(built on this semantic subtyping) and support polymorphism.
Ideally their approach and ours would converge: the type system at the root of
\lava could be mechanized and proved sound in \coq,
while still elaborating terms for a translation in \coq.

\paragraph{Explicit refinement types.}
\citet{ghalayini23} introduce a calculus ($\lambdaert$) that
makes refinement reasoning explicit by internalizing refinements
directly in the core calculus: refinement proofs are first-class proof terms,
and typing does not rely on subtyping.
\toolname shares the goal of making refinement reasoning explicit,
but differs fundamentally in its approach.
Rather than designing a refinement calculus inside type theory,
\toolname starts from \lh’s SMT-based system and translates it to \coq,
reifying implicit subtyping and proof obligations as explicit casts
and proof obligations in the target.
Moreover, $\lambdaert$ operates in a single semantic layer,
whereas \toolname must model the separation between
refined and unrefined worlds that arises in \lh.
Capturing this intertwined semantics is 
central to our contribution and does not arise in $\lambdaert$.

\paragraph{Programming with dependent types in \coq}
\emph{Equations}, by \citet{sozeau19}, is a \coq extension that simplifies writing
functions defined by dependent pattern matching.
Equations accepts incomplete definitions with holes, and presents them to the user as explicit proof obligations.
We could have used Equations to define the refined functions.
We decided to use tactics in the first place as we expected them to be more predictable.
We plan on re-investigating the adoption of Equations for a future extension towards non-structural recursion.
But this is merely a technical choice, as the way refined functions are defined
does not influence the analysis of the semantics of \lh, and only amounts to a local change.
This made switching to Equations a non-priority, further held up by the
necessity to adapt automation accordingly.
In particular, the coherence lemma between a refined definition and its unrefined graph relation
is a complex part of automation that has to be rewritten as part of that change.

Another possibility for writing refined functions in \coq is the \emph{Program}
extension by \citet{sozeau07}.
Program offers implicit predicate subtyping on subset types.
This generally makes generated code look a lot like \lh code, with no explicit
casts, but with less automation and proofs to be resolved by the user.
We do not target Program, as one goal is to make the relation between the
unrefined and refined semantics explicit, and type casts are part of this.

\paragraph{Integrating SMT solvers and hammers into \coq}

The \emph{SMTCoq} project \citep{armand11} provides an SMT backend to \coq
(currently CVC4~\cite{barrett11} and veriT~\cite{bouton09}),
by checking and reconstructing the solver's proof witnesses back into \coq.
Using SMTCoq directly to solve our proof obligations is hard, as a lot of
preprocessing must be done on our side.
Fortunately, \emph{Sniper} \citep{blot21,blot23} extends SMTCoq with many useful transformations of the \coq proof state.
Even though those transformations are designed to be compositional, safe and, certified,
we experienced that the tool is not
mature enough to be used as a replacement for our automation tactic without
significant effort.

Overall, it is hard to know when \texttt{snipe} will succeed or not without deep knowledge of the internals of Sniper.
Subset types and projections often make the tactic fail or loop on simple obligations.
On success, \texttt{snipe} is always slower than our automation tactic on our examples.
This might be improved by selecting and arranging the transformations that are most relevant to us.
A lot of transformations are also peculiar to our translation,
for instance saturating the context with the hypotheses obtained from the graph relations.
Adopting a structured approach to our automation in the style or combined with
Sniper is an ambitious but interesting project.

Using an SMT solver for the resolution of verification constraints means more hope
that a verified property of \lh is automatically provable in \coq.
But this also means generating more definitions and proofs,
for instance to verify properties of the encoding of datatypes in the solver logic,
or translating every inductive proposition to a recursive function in \texttt{bool}.
Sniper gives a way to generate the latter Boolean fixpoint automatically,
but we often failed to make generation work, or sometimes found the generated fixpoint to be wrong.

\citet{czajka18} define a \emph{hammer}~\cite{blanchette16} for \coq.
The hammer uses a heuristic translation into a fragment of first-order logic,
applies HOL ATPs (like E, Vampire or Z3), then tries to reconstruct the proof
into \coq proof terms.
Their translation and proof reconstruction is neither sound nor complete, but the approach works well enough in practice.
We want to avoid making the user have to complete proofs by itself.
Thus, we preferred focusing on predictable proof automation and avoided using
CoqHammer or machine learning-based tools like \emph{Tactician}~\cite{Blaauwbroek2020}.

\section{Conclusion}\label{sec:conclusion}
\toolname translates from \lh to \coq, formally relating SMT-based
refinement typing and foundational type theory.
\toolname turns implicit subtyping explicit through casts;
it models both semantics\,---\,refined and unrefined\,---\,
and generates (so far automatically solvable) proof obligations in \coq.
Our translation provides a new foundation for \lh that correctly accounts
for both refined and unrefined semantics in its SMT encoding.
Our implementation shows that this translation is practically possible.
Polymorphism and other more advanced features merely demand refactoring and engineering;
better automation (and automatic solving use SMT/hammers) is a question of connecting the dots.
With \toolname, we aim to put \lh and refinement typing on a more solid footing on its own and in relation to type theoretic approaches\,---\,while also helping the community understand the subtleties of refinement types.

\bibliographystyle{ACM-Reference-Format}
\bibliography{bib}

\iflongversion
  \clearpage
  \appendix
  \section{Full Type System for \texorpdfstring{$\er$}{lambda-r}}\label{app:typing}
  In this section, we give all well-formedness and typing rules.

We have six kinds of judgement: well-formedness of types and of declarations,
subtyping, type synthesis for refinements, type checking of expressions
(including refinements), and a simple type system for type refinements.
All except subtyping have an elaboration.

The grammar of type environements is given below.
Notice that we have scope annotations on the variables.
We write $\typEnv(x)$ to look up the variable $x$ in $\typEnv$ and
$\dom(\typEnv)$ to get the domain of $\typEnv$.
\[ \begin{array}{rlclll}
  \gramTitle{Typing environment}
    & \typEnv & \Coloneqq & \emptyset
    \mid \typEnv, \bind{\varlg{x}{c}}{\lhref}
    \mid \typEnv, \lhdt{\lhTC}{\seq{(\bind\lhConstr\lhref)}}
    \quad \color{gray}{s \in \{\topSymbol,\locSymbol,\recSymbol\}}
\end{array} \]

To avoid redundancy, in the whole appendix, we do not single out \keyw{if}
conditionals and treat them as any other pattern matching.

\subsection{Well-formedness and subtyping}

\paragraph{Simple Type System for Type Refinements}

We write $\lhsmptyp$ for a simple type, that is either a base type $\lhtyp$
or an arrow $\lhsmptyp \to \lhsmptyp$.
We use the function $\projType\lhref$ to project a refined type into a type $\lhsmptyp$.
It is defined as $\projType{\refT x \lhtyp \lhtm} \doteq \lhtyp$ and
$\projType\arr \doteq \projType{\lhref_x} \to \projType\lhref$.

\begin{figure}
  \[\begin{array}{c}
    \inference{(\bind{\varlg{x}{s}}{\lhref}) \in \typEnv
    & s \in \{\locSymbol,\topSymbol\}}
    {\hsmpj \typEnv x {\projType\lhref} {\vararlg{x}{\ar(\lhref)}{s}}}
    \quad \inference{
      \hsmpj \typEnv {\lhtm_1} {\lhsmptyp' \to \lhsmptyp} {\ann{\lhtm_1}}
        & \hsmpj \typEnv {\lhtm_2} {\lhsmptyp'} {\ann{\lhtm_2}}
      }{
        \hsmpj \typEnv {\lhtm_1 \ap \lhtm_2} \lhsmptyp {\ann{\lhtm_1} \ap \ann{\lhtm_2}}
      }
      \\[15pt]
      \inference{}{\hsmpj \typEnv \lhlit {\projType{\csttype\lhlit}} {\lhlit}}
      \quad\inference{\exists \lhTC \in \dom(\typEnv).\,(\bind\lhConstr\lhref) \in \typEnv(\lhTC)}
      {\hsmpj \typEnv \lhConstr {\projType\lhref} \lhConstr}
      \\[15pt]
      \inference{\hsmpj \typEnv \lhtm \Bool {\ann\lhtm}}
      {\hsmpj \typEnv {\lhnegb \lhtm} \Bool {\lhnegb{\ann\lhtm}}}
      \quad \inference{
        \hsmpj \typEnv {\lhtm_1} {\lhsmptyp_1} {\ann{\lhtm_1}}
      &\hsmpj \typEnv {\lhtm_2} {\lhsmptyp_2} {\ann{\lhtm_2}}
      &\lhsmptyp_1 = \lhsmptyp_2
    }{\hsmpj \typEnv {\lhtm_1 \lheqb \lhtm_2} \Bool {\ann{\lhtm_1} \lheqb \ann{\lhtm_2}}}
    \\[15pt] \inference{
      \projType{\csttype\lhbop} = \lhsmptyp_1 \to \lhsmptyp_2 \to \lhsmptyp
      & \hsmpj \typEnv {\lhtm_1} {\lhsmptyp_1} {\ann{\lhtm_1}}
      & \hsmpj \typEnv {\lhtm_2} {\lhsmptyp_2} {\ann{\lhtm_2}}
    }{
      \hsmpj \typEnv {\lhtm_1 \lhbop \lhtm_2} \lhsmptyp {\ann{\lhtm_1} \lhbop \ann{\lhtm_2}}
    }
  \end{array}\]
  \caption{Simple type system for refinements}
  \Description{Simple type system for refinements}
  \label{fig:simpletypes}
\end{figure}
The rules are in \cref{fig:simpletypes}.
They are standard and similar to the synthesis rules.
The only rule with a non-trivial elaboration is the one for variables.

\paragraph{Well-Formedness of Types}
\label{sec:wf-types}

The substitution of a term in a type is defined as
$\refT x \lhtyp {\lhtm_x} \subst x \lhtm \doteq \refT x \lhtyp {\lhtm_x}$,
$\refT y \lhtyp {\lhtm_y} \subst x \lhtm \doteq \refT y \lhtyp {\lhtm_y \subst x \lhtm}$,
$(\arr[x][\lhref_x][\lhref]) \subst x \lhtm \doteq \arr[x][\lhref_x][\lhref]$ and
$(\arr[y][\lhref_y][\lhref]) \subst x \lhtm \doteq \arr[y][\lhref_y \subst x \lhtm][\lhref \subst x \lhtm]$.

\[ \begin{array}{cc}
  \inference{
    \text{if } \lhtyp = \lhTC\text{, then } \lhtyp \in \dom(\typEnv) \\
    \hsmpj {\typEnv, \bind {\varloc x} \refTT{\lhtyp}} \lhtm \Bool {\ann\lhtm}
  }{
	  \hwfj \typEnv {\refT x \lhtyp \lhtm} {\refT x \lhtyp {\ann\lhtm}}
  }[\!\!\WFTBase]
  &
  \inference{
	  \hwfj \typEnv {\lhref_x} {\ann{\lhref_x}}
    &\hwfj {\typEnv, \bind{\varloc x}{\ann{\lhref_x}}} \lhref {\ann\lhref}
  }{
    \hwfj \typEnv {\arr} {\arr[x][\ann{\lhref_x}][\ann{\lhref}\subst x {\lhproj x}]}
  }[\!\!\WFTFun]
\end{array} \]

\paragraph{Well-Formedness of Declarations}

Well-formedness of declarations takes as input a typing environment with no local variables
and a sequence of declarations.
\[
  \inference{
    \seq{\hwfj {\typEnv, \lhdt{\lhTC}{(\bind{\lhConstr_k}{\ann{\lhref_k}})_{k < i}}} {\lhref_i} {\ann{\lhref_i}}}
    &\seq{\ann{\lhref_i} = \arrs[x_j][\refT{\_}{\_}{\_}][j \leq m_i][\lhref_i']}
    \\ \hwfj {\typEnv, \lhdt\lhTC{(\bind{\lhConstr_i}{\ann{\lhref_i}})_{i \leq n}}} {\seq\lhdecl} {\seq{\ann\lhdecl}}
    & TC, \seq{C_i} \text{ fresh for } \typEnv
  }{
    \hdecl {\typEnv} {\lhdata \lhTC {(\lhConstr_i :: \lhref_i)_{i \leq n}}; \seq\lhdecl}
    {\lhdata \lhTC {(\lhConstr_i :: \ann{\lhref_i})_{i \leq n}}; \seq{\ann\lhdecl}}
  }[\WFDTC]
\]
Rule \WFDTC requires a fresh type constructor $\lhTC$ and fresh data
constructors $\lhConstr_i$ with well-formed, first-order types returning a
refinement of $\lhTC$.
First-orderness is ensured because every argument to a constructor refines a base type.
The rule recursively checks the well-formedness of the subsequent declarations.

\[
  \inference{
    \hwfj \typEnv {\lhref_f} {\ann{\lhref_f}}
    & \ann{\lhref_f} = \arrs
    \\ \hchkj {\typEnv, \bind{\varrec f}{\ann{\lhref_f}}, (\bind{\varloc{x_i}}{\noproj(\lhref_i)})_{i \leq n}} {\lhexp} {\noproj(\lhref)} {\ann\lhexp}
    &\hwfj {\typEnv, \bind{\vartop f}{\ann{\lhref_f}}} {\seq\lhdecl} {\seq{\ann\lhdecl}}
  }{
    \hdecl \typEnv {\lhdefrefl f {\lhref_f} \lhexp; \seq\lhdecl}
    {\lhdefrefl f {\ann{\lhref_f}} {\ann\lhexp}; \seq{\ann\lhdecl}}
  }[\WFDDef]
\]
Rule \WFDDef checks the body of $f$ using its ``maximal'' function type\,---\,the
return type $\lhref$ must not be a function type.
That is, we forbid point-free definitions and require complete $\eta$-expansion.

We use a function $\noproj(\lhref)$ that removes projections in the refinements of
$\lhtm$, if the projections are not around variables bound in $\lhref$.
For instance, $\noproj(\arr[x][\lhref_x][\refT v \Int {\proj x == \proj y}]
= \arr[x][\lhref_x][\refT v \Int {\proj x == y}]$.
This is used to remove the projections around the parameters $x_i$
that are local variables so considered unrefined in the context and in the return
type $\lhref$.

\paragraph{Subtyping}

Subtyping has two rules: one for base types and one for function types.
The left rule only checks that the base types are equal.
\[ \begin{array}{cc}
  \inference{
  }{
    \hsubj \typEnv {\refT {x_s} \lhtyp {\lhtm_s}} {\refT {x_t} \lhtyp {\lhtm_t}}
  }[\!\!\TSRefAx]
&
\inference{
  \hsubj \typEnv {\lhref_{x_t}} {\lhref_{x_s}}
  & \hsubj {\typEnv, \bind{x_t}{\lhref_{x_t}}} {\lhref_s \subst {x_s} {x_t}} {\lhref_t}
}{
  \hsubj \typEnv {\arr[x_s][\lhref_{x_s}][\lhref_s]} {\arr[x_t][\lhref_{x_t}][\lhref_t]}
}[\!\!\TSArr]
\end{array}\]

\subsection{Typing Expressions}
\label{sec:typingrules}

We have a bidirectional type system where synthesis is done exclusively on refinements
and checking is done on structural expressions,
and always ends with a subtyping rule when done on refinements.
A synthesis rule $\hsynj \typEnv \lhtm \lhref {\ann\lhtm}$ inputs $\typEnv$ and $\lhtm$,
synthesizes the type $\lhref$ and outputs $\ann\lhtm$.
A checking rule $\hchkj\typEnv\lhexp\lhref{\ann\lhexp}$ inputs $\typEnv$, $\lhexp$ and
$\lhref$ and checks that $\lhexp$ has the refinement type $\lhref$, and outputs $\ann\lhexp$.

\begin{figure}
  \[\begin{array}{c}
    \inference{
      (\bind{\varloc x}{\lhref}) \in \typEnv
      & \ar(\lhref) = 0
    }{
    \hsynj \typEnv x \lhref {\lhinj{\vararloc{x}{0}}{\lhref}}}[\!\!\SVarLoc]
    \\[15pt]
    \inference{
      (\bind{\varlg{x}{s}}{\lhref}) \in \typEnv
      & s \in {\locSymbol,\topSymbol,\recSymbol}
      & \text{if } s = \locSymbol \text{ then} \ar(\lhref) > 0
    }{
    \hsynj \typEnv x \lhref {\vararlg{x}{\ar(\lhref)}{s}}}[\!\!\SVar]
    \\[15pt]
    \inference{}{\hsynj \typEnv \lhlit {\csttype\lhlit} {\lhinj{\lhlit}{\csttype\lhlit}}}[\!\!\SLit]
    \qquad\inference{\exists \lhTC \in \dom(\typEnv).\,(\bind\lhConstr\lhref) \in \typEnv(\lhTC)}
    {\hsynj \typEnv \lhConstr \lhref \lhConstr}[\!\!\SData]
    \\[15pt]
    \inference{
      \csttype\lhnegb = \arr
      \\ \hchkj \typEnv \lhtm {\lhref_x} {\ann\lhtm}
    }{
      \hsynj \typEnv {\lhnegb\lhtm}{\lhref\subst{x}{\ann\lhtm}}{\lhnegb\ann\lhtm}
    }[\SNeg]
    \quad \inference{
      \csttype\lhbop = \arr[x_1][\lhref_1][\arr[x_2][\lhref_2][\lhref]]
      \\
      \hchkj \typEnv {\lhtm_1} {\lhref_1} {\ann{\lhtm_1}}
      & \hchkj \typEnv {\lhtm_2} {\lhref_2} {\ann{\lhtm_2}}
    }{
      \hsynj \typEnv {\lhtm_1 \lhbop \lhtm_2}
      {\lhref\subst{x_1}{\ann{\lhtm_1}}\subst{x_2}{\ann{\lhtm_2}}}
      {\ann{\lhtm_1} \lhbop \ann{\lhtm_2}}
    }[\SBin]
    \\[15pt] \inference{
      \hsynj \typEnv {\lhtm_1} {\refT{x_1}{\lhtyp_1}{\lhtm_1'}} {\ann{\lhtm_1}}
      &\hsynj \typEnv {\lhtm_2} {\refT{x_2}{\lhtyp_2}{\lhtm_2'}} {\ann{\lhtm_2}}
      &\lhtyp_1 = \lhtyp_2
    }{
      \hsynj \typEnv {\lhtm_1 \lheqb \lhtm_2}
      {\refT x \Bool {x = (\lhproj{\ann{\lhtm_1}} \lheqb \lhproj{\ann{\lhtm_2}})}}
      {\ann{\lhtm_1} \lheqb \ann{\lhtm_2}}
    }[\SEq]
    \\[15pt]
    \inference{
      \hsynj \typEnv {\lhtm_1} {\arr} {\ann{\lhtm_1}}
        & \hchkj \typEnv {\lhtm_2} {\lhref_x} {\ann{\lhtm_2}}
      }{
        \hsynj \typEnv {\lhtm_1 \ap \lhtm_2} {\lhref \subst x {\ann{\lhtm_2}}} {\ann{\lhtm_1} \ap \ann{\lhtm_2}}
      }[\SApp]
  \end{array}\]
  \caption{Type synthesis rules for refinements}
  \Description{Type synthesis rules for refinements}
  \label{fig:synrules}
\end{figure}

The type synthesis rules for refinements are given in \cref{fig:synrules}.
The types of literals and operators are given by a function $\csttype\cdot$ that
generates a singleton type (or \textit{selfified} type~\cite{ou04}).
Typing equality is similar, but works for any base type.

\begin{figure}
  \[\begin{array}{c}
    \inference{
      \hwfj\typEnv {\lhref_x} {\ann{\lhref_x}}
      &\ann{\lhref_x} = \arrs[x_i][\lhref_i][i \leq n][\lhref']
      \\\hchkj {\typEnv, (\bind{\varloc{x_i}}{\noproj(\lhref_i)})_{i \leq n}} {\lhexp_1} {\noproj(\lhref')} {\ann{\lhexp_1}}
      & \hchkj {\typEnv, \bind{\varloc x}{\ann{\lhref_x}}} {\lhexp_2} \lhref {\ann{\lhexp_2}}
      &\hwfj\typEnv \lhref {\ann\lhref}
    }{
      \hchkj \typEnv {\lhlet x {\lhref_x} {\lhexp_1}{\lhexp_2}} \lhref {\lhlet x {\ann{\lhref_x}} {\ann{\lhexp_1}} {\ann{\lhexp_2}}}
    }[\CLet]
    \\[15pt]
    \inference{
      \hsynj \typEnv \lhtm {\refT v \lhTC {\lhtm_v}} {\ann\lhtm}
      & \lhdt{\lhTC}{(\bind{\lhConstr}{\arrsSeq[y_i][\lhref_y][\refT{v}{\lhTC}{\lhtm_i}]})_{i \leq n}} \in \typEnv
      \\\forall i \leq n \textit{ s.t. } \lhexp_i \neq \dead.
      (\hchkj {\typEnv_i} {\lhexp_i'} {\lhref'} {\ann{\lhexp_i}})
    }{
      \hchkj \typEnv {
        \lhcase \lhtm {(\case{\lhConstr_i \ap \seq{y_i}}{\lhexp_i})_{i \leq n}
      }} \lhref
      {\lhcase {\ann\lhtm} {(\case{\lhConstr_i \ap \seq{y_i}}{\ann{\lhexp_i}})_{i \leq n}}}
    }[\CCase]
    \\[15pt]
    \inference{
      \hsynj \typEnv {\lhtm_h} {\refT x \Unit {\lhtm_p}} {\ann{\lhtm_h}}
      & \hsynj {\typEnv, \bind{\varloc x}{\refT x \Unit {\lhtm_p}}} {\lhtm} \lhref {\ann{\lhtm}}
    }{
      \hsynj \typEnv {\lhqmark{\lhtm}{\lhtm_h}} \lhref {\annlhqmark{\ann{\lhtm}}{\ann{\lhtm_h}}{\lhtm_p}}
    }[\CHint]
    \\[15pt]
    \inference{
      \hsynj \typEnv \lhtm {\lhref_s} {\ann{\lhtm}}
      & \hsubj \typEnv {\lhref_s} {\lhref_t}
    }{
      \hchkj \typEnv \lhtm {\lhref_t} {\lhsub {\ann\lhtm} {\lhref_s} {\lhref_t}}
    }[\CSyn]
  \end{array}\]
  Where in \CCase, if $\lhtm = x$ then
  $\lhexp_i' \doteq \lhexp_i\subst x {\lhConstr\ap\seq{y_i}}$,
  $\lhref' \doteq \lhref\subst x {\lhConstr\ap\seq{y_i}}$
  and $\typEnv_i \doteq \typEnv', \seq{(\varloc{y_i} : \noproj(\lhref_y))}, \typEnv''\subst{x}{\lhConstr\ap\seq{y_i}}$
  assuming $\typEnv = \typEnv', x:\refT v \lhTC {\lhtm_v}, \typEnv''$.
  \\Otherwise, $\lhexp_i' = \lhexp_i$, $\lhref' = \lhref$ and
  $\typEnv_i \doteq \typEnv, \seq{(\varloc{y_i} : \noproj(\lhref_y))}$.
  \caption{Type checking rules for expressions}
  \Description{Type checking rules for expressions}
  \label{fig:checkrules}
\end{figure}

The type checking rules are given in \cref{fig:checkrules}.
Both rules~\CLet and~\CCase use function $\noproj$ to remove projections around
newly-introduced variables in type refinements.

Rule~\CLet has two subtleties.
The first is that, if $\lhref_x$ is an arrow type, we check $\lhexp_1$ in the
context extended by its parameters $x_1$ through $x_n$.
The second is that we check well-formedness of the type $\lhref$: this is to ensure that this type does not depend on the defined variable $x$.

In rule~\CCase, if we match on a variable, we substitute the variable for its
pattern in the typing environment, the branch and the type.
We do not bother adding information on the value of the match in other cases (as
in occurence typing~\cite{tobinhochstadt10}):
such information would only be useful to resolve VCs, which we do not do here.
All branches must be present, but some can be $\dead$.

  \section{Appendix for the Translations}\label{app:translation}
  \subsection{Typeclass for projection}

We use a notation $\proj\cdot$ in \coq that is generic over subset types and packs.
This is not just a notational trick: in our \coq library, we define a generic projection, as part of a typeclass.
The operator $\proj\cdot$ is an alias for the function \rocq!proj! of the typeclass.
This function must respect a “proof irrelevance” property \texttt{po} ensuring equality modulo potential refinements.
\begin{rocqcode}
Class !$\gProj$! {R U : Type} :=
  { proj: R !$\to$! U; po: !$\forall$! (x y : R), proj x = proj y !$\to$! x = y }.
Global Notation !$\projleft$!x!$\projright$! := (proj x).
\end{rocqcode}

For a subset type, \rocq!proj! is defined simply as \coq's \rocq!proj1_sig!.
For a \fClass{}, we define the projection from a $\fClass$ to an $\fuClass$ simply as
$\uClassToClass x \doteq \cqrecuClass {\relnmInst x} {\functrelInst x}$.
The parameters \texttt{R} and \texttt{U} are implicit, as indicated by the curly brackets.
\coq can always infer the correct instance automatically, since we define the
projection operator only on subset types and \fClass{}s.

\subsection{Unrefined Translation}

\subsubsection{Call hypotheses}

We give a thorough definition of the function $\hypFromAmName$ used in
\cref{sec:unrefined}.

\begin{definition}[Call hypotheses]
  \label{def:hypFromAm}%
  Given a refinement $\lhtm$, a quantifier $\hypQ \in \{\forall, \exists\}$
  and assuming the map $\frApp$, $\hypFromAmQ\lhtm$ is defined as:
  \[\begin{array}{r@{\,}c@{\,}l@{\,}l}
    \hypFromAmQ\lhtm
    &\doteq & \{
      \hypQ z, \relnm f
      \ \utrLhtmCq{\lhtm_1} \dots \utrLhtmCq{\lhtm_n} \ z
    & \textcolor{gray}{(*\  f \ap \lhtm_1 \dots \lhtm_n \mapsto z \in \frApp\ *)}
  \\&& \quad \mid
  \varartop f n \ap \lhtm_1 \dots \lhtm_n \text{ subterm of } \lhtm
  \land n \geq 0
\}
  \\& \cup& \{ \hypQ z, (\getUPackRel f)
    \ \utrLhtmCq{\lhtm_1} \dots \utrLhtmCq{\lhtm_n} \ z
    & \textcolor{gray}{(*\  f \ap \lhtm_1 \dots \lhtm_n \mapsto z \in \frApp\ *)}
  \\&& \quad \mid 
  \vararloc f n \ap \lhtm_1 \dots \lhtm_n \text{ subterm of } \lhtm
  \land n > 0
\}
  \\& \cup& \{ \hypQ z, (\getPackRel f)
    \ \utrLhtmCq{\lhtm_1} \dots \utrLhtmCq{\lhtm_n} \ z
    & \textcolor{gray}{(*\ \lhproj f \ap \lhtm_1 \dots \lhtm_n \mapsto z \in \frApp\ *)}
  \\&& \quad \mid 
  \lhproj{\vararloc f n} \ap \lhtm_1 \dots \lhtm_n \text{ subterm of } \lhtm
  \land n > 0
\}
  \\& \cup& \{\hypQ z, \relnmOp\lhbop \ap \utrLhtmCq{\lhtm_1} \ap\utrLhtmCq{\lhtm_2} \ap z
  \mid \lhtm_1 \lhbop \lhtm_2 \text{ subterm of } \lhtm\}
    & \textcolor{gray}{(*\ \lhtm_1 \lhbop \lhtm_2 \mapsto z\in \frApp\ *)}
 \\ & \cup& \{\hypQ z, \relnmOp\lhnegb \ap \utrLhtmCq{\lhtm_0} \ap z
 \mid \lhnegb\lhtm_0 \text{ subterm of } \lhtm\}
    & \textcolor{gray}{(*\ \lhnegb \lhtm_0 \mapsto z \in \frApp\ *)}
  \end{array}\]
  We use $\getUPackRel f$ and $\getPackRel f$ instead
  of $\relnmUInst f$ and $\relnmInst f$ to obtain
  the relation from a pack, going through the conversion
  from the representation inside the pack to the one we expect
  (see \cref{app:pack}).
\end{definition}

\subsubsection{Building function paths}

We now describe how we compute the paths of a reflected function,
necessary to build its unrefined graph relation.

\begin{definition}
  \label{def:functionpaths}%
  Consider a reflected function $f$, not containing
  $\lambda$-abstractions or partial function applications, defined as $\lhdefrefl f \arrs \lhexp$.
  The paths of $f$ are given by an auxiliary function $\eqguardName$.
  \[
    \mkGraph(\lhexp; \seq x)
    \doteq \{(\codom(\sigma_\params),\sigma_\brSub,\lhexp_\brSub)
    \mid (\sigma_\params, \sigma_\brSub, \lhexp_\brSub) \in \eqguard[\seq{(x,x)}][\emptyset]{\lhexp}\}
  \]
\end{definition}

The function $\eqguardName$ takes as input a triple corresponding to the current branch
and returns a set of branches where matches on the arguments are put first,
and nested pattern matchings are merged into one.
The arguments are thus:
(1) a map from variables to simple terms, the patterns for the parameters,
(2) a map from additional matched refinements to their pattern, and
(3) the expression $\lhexp$ to turn into branches.

\begin{definition}
  \label{def:eqguard}%
  Given a map from variables to simple terms $\sigma_\params$,
  a map from simple terms to simple terms $\sigma_\brSub$
  and an expression $\lhexp$,
  $\eqguard{\lhexp}$ is defined as:
  \[ \begin{array}{l}
    \eqguard\lhtm \doteq \{(\sigma_\params, \sigma_\brSub, \lhtm)\}
    \\\eqguard\dead \doteq \emptyset
    \\\eqguard{\lhlet x \lhref {\lhexp_1}{\lhexp_2}}
    \doteq \bigcup_{(\sigma_\params',\sigma_\brSub',\lhtm) \in \eqguard{\lhexp_1}}
    \eqguard[\sigma_\params'][\sigma_\brSub']{\lhexp_2 \subst x \lhtm}
    \\\eqguard{\lhcase \lhtm {(\case{\lhConstr_i\ap\seq{y_i}}{\lhexp_i})_{i \leq n}}}
    \\ \doteq \begin{cases}
      \eqguard{\lhexp_j \subst {y_i}{\lhtm_i}_{i \leq m}}
      \quad\hfill\text{if } \exists j \leq n.\ \lhtm = \lhConstr_j \ap \lhtm_1 \dots \lhtm_m
    &\text{or } \exists (\lhtm, \lhConstr_j \ap \lhtm_1 \dots \lhtm_m) \in \sigma_\params \cup \sigma_\brSub
    \\ \bigcup_{i \leq n} \guard(
    \sigma_\params\subst x {\lhConstr_i\ap\seq{y_i}},
    \sigma_\brSub \subst x {\lhConstr_i\ap\seq{y_i}},
    \lhexp_i \subst x {\lhConstr_i \ap \seq{y_i}})
    &\text{if $\lhtm = x$} \\
    \bigcup_{i \leq n} \guard(\sigma_\params, (\sigma_\brSub, (\lhtm, {\lhConstr_i\ap\seq{y_i}})), \lhexp_i)
    &\text{otherwise}
    \end{cases}
  \end{array} \]
  Where we write $\sigma_\params \subst x {\lhConstr_i\ap\seq{y_i}}$
  to substitute in the codomain of $\sigma_\params$
  and $\sigma_\branch\subst x {\lhConstr_i\ap\seq{y_i}}$ to substitute
  either in the domain or codomain of $\sigma_\branch$,
  and we assume that all $\seq{y_i}$ are different from $x$.
\end{definition}

There are two base cases:
either we find a branch with a refinement at the end,
or we find an $\dead$ branch, that we delete.
This way, we ensure that no $\dead$ appears in the function paths.
Let-bindings are substituted away: this is possible since
all branches lead to a refinement, that can be substituted anywhere,
unlike expressions.
This would not be possible with $\lambda$-abstractions or partial applications.

There are three equations for pattern matching.
A match on an application (last case) creates a new pattern in the map $\sigma_\brSub$.
A matched variable can be either a parameter of the function
or a variable introduced by pattern matching.
We substitute this variable by its pattern in the image of $\sigma_\params$
and everywhere in the second case.
The first case accounts for a match on an already matched variable or application,
in which case there is only one valid pattern.

\subsection{Refined translation}
\label{app:refined}

We give the full definitions of the functions
$\exprToTacName$ and $\trLhCqName$ from \cref{sec:transexpr}.

\begin{figure}
  \[\begin{array}{l}
    \trLhexpCq\lhtm \doteq \tac{exact} (\trLhtmCq\lhtm)
    \qquad \trLhexpCq\dead \doteq \exfalso;\ \ltacOracle.\\
    \trLhexpCq{
    	\lhcase \lhtm {(\cqcase{\lhConstr_i \ap \seq{y_i}}{\lhexp_i})_{i \leq n}}}
    \doteq
    \\\qquad\begin{cases}
    	\left(\tac{generalize dependent}\ z; \right)_{z \in Z}
    	\tac{induction}\ x\ \tac{as}
    	\ [\despattern_1\ |\ \dots\ |\ \despattern_n]; \tac{intros}.
    	\\\quad\{\, \trLhexpCq{\lhexp_i}\, \}_{i \leq n}
    	\\  \tac{destruct}
    	\ \proj{\trLhtmCq\lhtm}\ \tac{as}\ [\seq{y_1} | \dots | \seq{y_n}]\ \tac{eqn:}?.
    	\{\, \trLhexpCq{\lhexp_i}\, \}_{i \leq n}
    \end{cases}\\
  \trLhexpCq{\ite\lhtm{\lhexp_1}{\lhexp_2}}
\doteq \tac{destruct}\ \proj{\trLhtmCq\lhtm}\ \tac{eqn:}?.
\ \{ \trLhexpCq{\lhexp_i} \}_{i = 1,2}\\
  \trLhexpCq{\lhlet x \lhref {\lhexp_1} {\lhexp_2}} \doteq
  \hspace{12.5em}(\ar(\lhref) = 0)
 \\ \qquad\tac{assert}\ (\bind{x'}{\trLhrefCq\lhref}).
  \ \{\ \trLhexpCq{\lhexp_1}\ \}
  \ \tac{destruct}\ x'\ \tac{as}\ [x \ap \wit x].
  \ \trLhexpCq{\lhexp_2}
  \\ \trLhexpCq{\lhlet {x} {\lhref_x \doteq \arrs} {\lhexp_1} {\lhexp_2}} \doteq
  \hspace{4.75em}(\ar(\lhref_x) > 0)
 \\ \qquad \tac{assert}\ (x'\colon \forall (\bind{x_i}{\trLhrefCq{\lhref_i}})_{i \leq n}, \trLhrefCq\lhref).
 \ \{\ \tac{intros}\ (\texttt{des}_i)_{i \leq n}.\ \trLhexpCq{\lhexp_1}\ \}
  \\ \qquad\tac{unshelve refine}\ (\cqlet {\bind{x}{\trLhrefCq{\lhref_x}}} {
    \ltacTac{fun\_to\_pack $x'$}} {\_}).
  \\\qquad\trLhexpCq{\lhexp_2}
  \hspace{4.5em}\textit{where } \texttt{des}_i = x_i \textit{ if $\lhref_i$ is an arrow, otherwise } \texttt{des}_i = [x_i \ap \wit{x_i}]
  \\
	\trLhexpCq{\annlhqmark\lhtm{\lhtm_h}{\lhtm_p}}
	\doteq 
	\tac{assert}\ (\utrLhtmCqprop{\lhtm_p})\ \tac{by}\ (\tac{exact}\ \projTwo{\trLhtmCq{\lhtm_h}}).\ \trLhexpCq{\lhtm}
  \end{array}\]
  \caption{Refined translation of expressions $\exprToTacName$}
  \Description{Refined translation of expressions to tactics}
  \label{fig:transexp}
\end{figure}

The definition of $\exprToTacName$ is in \cref{fig:transexp}.
The formal translations missing from the presentation in \cref{sec:transexpr}
are those for \hask{case}, \hask{if} and \hask{let}.
The translation of let-bindings corresponds to what is explained in \cref{sec:transexpr}:
we translate the let-binding as an \rocq{assert}.
Then, in the first-order case, we destruct the introduced dependent pair
and in the higher-order case, we introduce a pack for the new function.

Pattern matching on a refinement $\lhtm$ is translated either
with the \rocq{destruct} or the \rocq{induction} tactic.
The difference is that the second creates induction hypotheses.
For \rocq{destruct}, \rocq{eqn:?} asserts the equality between the term and its
pattern in the context of each branch.
Being a special case, ifs are also translated \rocq{destruct}.

We use \rocq{induction} when $\lhtm$ is a variable that is used as an
\emph{inductive variable} in the term.
We call inductive variable a variable that originates from the destruction of
a parameter $p$ and is used as argument of a recursive call in the position of $p$.
In the recursive call \hask{getDouble xs' (i-1)} from \cref{fig:trans-getDouble}, the recursive
variable is \verb!xs'!, obtained from the destruction of the parameter \verb!xs!
as \verb!Cons x xs'!.

During elaboration of a matching on a variable $y$, we look for all recursive calls
where $y$ is an inductive variable (and no other inductive variable has been found)
and annotate $y$ as the inductive variable of these calls, an information that
is used to translate the call (\cref{fig:transtm}).
If such calls have been found, we use \rocq{induction} and otherwise fall back
to \rocq{destruct}.

Before using \rocq{induction}, we must generalize the parameters that have not
been matched upon.
The set of those remaining parameters is here called $Z$,
in practice it is determined with an additional parameter to elaboration that
records which parameters have been destructed.
Without this generalization, the induction hypotheses generated by
\rocq{induction} would already be instantiated for all parameters of the functions,
making it impossible to create the correct recursive calls.

The tactic \rocq{induction} also has as parameters a sequence of patterns
$\texttt{des}_i$ that introduces the names of the datatype parameters.
For each $i \leq n$, $\texttt{des}_i$ is the sequence $(y_i^j)_{j \leq m_i}$
with names $\ihasm{y_i^j}$ inserted after $y_i^j$ if this variable is used as an
inductive variable later.

\begin{figure}
  \[\begin{array}{ll}
    \trLhtmCq\lhlit
    \doteq \trLit\lhlit
    & \trLhtmCq\lhConstr
    \doteq \lhConstr
    \\ \trLhtmCq{\varartop x {n>0}}
    \doteq \fInst x
    & \trLhtmCq{\vararlg{x}{\_}{\_}}
    \doteq x
    \\ \trLhtmCq{\lhnegb\lhtm}
    \doteq \lhnegb\trLhtmCq\lhtm
    & \trLhtmCq{\lhtm_1 \lhbop \lhtm_2}
    \doteq\trLhtmCq{\lhtm_1} \lhbop \trLhtmCq{\lhtm_2}
    \\ \trLhtmCq{\lhtm_1 \lheqb \lhtm_2}
    \doteq \trLhtmCq{\lhtm_1} \cqeqb \trLhtmCq{\lhtm_2}
    \\
    \multicolumn{2}{l}{
    \trLhtmCq{\lhsub \lhtm {\lhref_s} {\lhref_t}}
    \doteq \subCastExplicit {\trLhrefCq{\lhref_t}} {\trLhtmCq\lhtm} {\trLhrefCq{\lhref_s}}
    }\\
    \multicolumn{2}{l}{
\trLhtmCq{\lhinj \lhtm \lhref}
\doteq\, \ex {\!\!(\forall (\bind{x}{\cqtyp}), \cqprop)} {\!\trLhtmCq\lhtm} {\!\ltacOracle}
 \textit{,where } \refT x \cqtyp \cqprop = \trLhrefCq\lhref    }\\
\multicolumn{2}{l}{
  \trLhtmCq{\lhtm \ap \lhtm_1 \dots \lhtm_n}
  \doteq \begin{cases}
    \ihasm{y}\ap\ltacOracle \ap {\left(\cqtm_i \right)_{i \leq n}}
      &\textit{if } \lhtm = f_{/(n, \recSymbol(y, \texttt{state}))}
    \\ \defnmInst f \ap \trLhtmCq{\lhtm_1} \dots \trLhtmCq{\lhtm_n}
      &\textit{if } \lhtm = \vararloc f n
    \\ f \ap \trLhtmCq{\lhtm_1} \dots \trLhtmCq{\lhtm_n}
      &\textit{if } \lhtm = \varartop f n
    \\ \trLhtmCq{\lhtm} \ap \trLhtmCq{\lhtm_1} \dots \trLhtmCq{\lhtm_n}
      &\textit{otherwise}
  \end{cases}
}
  \end{array}\]
  \caption{Refined translation of refinements; in the last equation, $n > 0$ is maximal}
  \Description{Refined translation of refinements}
  \label{fig:transtm}%
\end{figure}

The definition of $\trLhCqName$ is in \cref{fig:transtm}.

The translation of a literal is a builtin literal of \coq.
They have a base type ($\trLhCqName$ uses the same function $\trLit\lhlit$ as $\utransName$),
but are always surrounded by an injection to their singleton type,
as explained in \cref{sec:typing}.
Data constructors, on the other hand, are the refined ones from \cref{fig:transdata}.
Operators are translated to the custom operators with refined types that we
define in our \coq library (and for which we use here the same symbols).
Equality however is translated to the equality operator from our equality typeclass,
which is defined for subset types as the equality of the projected terms.

Variables names, when higher-order, must always refer to a pack instance.
The translation of a top-level function name $x$ is thus its instance $\fInst x$,
similar to what happens in $\utransName$.
If $x$ is local, it is already a pack by construction.
An injection is translated as the construction of a dependent pair from the
translation of the term and a proof generated by $\solver$.

The translation of an application is subtle.
A non-recursive call (second and third case) is translated recursively on the subterms.
If the head term is a local variable $f$ (second case), then $f$ binds a pack,
so $\trLhCqName$ projects on the field containing the actual refined function.

A recursive call (first case) is translated as a call to the induction hypothesis $\ihasm y$
obtained from the translation of pattern matching,
where $y$ is the inductive variable assigned during elaboration of pattern matching,
appearing in the annotation of $f$.
There must always be an inductive variable for our translation to succeed,
which is why we are restricted to structural recursion.

The induction hypothesis $\ihasm y$ is already specialized to $y$.
It is also specialized on the positions of parameters that have been
destructed when $\ihasm y$ was introduced.
To know how to translate the arguments $\lhtm_1$ to $\lhtm_n$,
we use the second annotation in $f$, a sequence called \texttt{state} of length $n$.
\texttt{state} is a snapshot of the destruction of the parameters when the
induction hypothesis was created.
For each parameter $x_i$, \texttt{state} tells us if $x_i$ had been destructed and its arity.
For each $i \leq n$, we define \[
  \cqtm_i \doteq \begin{cases}
    \emptyset &\text{if parameter $i$ is destructed; otherwise:}
    \\ \trLhtmCq{\lhtm_i} &\text{if parameter $i$ is higher-order}
    \\ \proj{\trLhtmCq{\lhtm_i}} \ltacOracle &\text{if parameter $i$ is first-order}
  \end{cases} \]
First-order values and their refinement proofs must be given separately to the
induction hypothesis, which explains the last case.
The first case is of a position that is already instantiated.

  \section{Proofs of the Properties}\label{app:properties}
  \subsection{Preservation of Evaluation Semantics of Expressions}%
\label{sec:app-operational}

In this subsection, we look at the operational semantics of reflected functions.
We define a reduction semantics for expressions with no higher-order let-bindings
(that is, bound $\lambda$-abstractions).
Like in \cref{app:translation}, we encode conditionals as case expressions without loss of generality.
The goal is to show that the operational semantics of the expression is preserved
by the translation to the graph relation.

\begin{figure}
  \[\begin{array}{ll}
    \ssred {\lhtm} {\delta(\lhtm)}
    \quad\text{ if } \lhtm \in \{\lhnegb\lhlit, \lhlit_1 \lhbop \lhlit_2, \lhval_1 \lheqb \lhval_2\}
  &\ssred{\annlhqmark{\lhtm}{\lhtm_h}{\lhtm_p}} \lhtm
  \\\ssred{\lhinj \lhtm \lhref}{\lhtm'}
  &\ssred{\lhsub\lhtm{\lhref_s}{\lhref_t}}{\lhtm'}
  \\\ssred{\lhcase {\lhConstr_k \ap \seq{\lhval_k}} {(\case{\lhConstr_i \ap \seq{x_i}}{\lhexp_i})_{i \leq n}}}
  {\lhexp_k \seq{\subst{x_k}{\lhval_k}}}
  \\\ssred{\lhlet x \lhref \lhval \lhexp} {\lhexp \subst x \lhval}
  \\\multicolumn{2}{l}{\ssred{f \ap \lhval_1 \dots \lhval_n}{\lhexp \subst{x_i}{\lhval_i}_{i \leq n}}
  \quad\text{if } \lhrefl f \arrs \lhexp \in \seq\lhdecl}
  \end{array}\]
  Where $\delta$ is a function returning the correct value for primitive operations.
  \caption{Reduction relation}
  \Description{Reduction relation}
  \label{fig:ssred}
\end{figure}

\begin{definition}[Small-steps reduction]
  \label{def:ssred}%
  We equip $\er$ with a small-steps call-by-value reduction.
  The base reduction rules are given in \cref{fig:ssred}.
  They are applied under an environment of declarations $\seq\lhdecl$ from which
  we retrieve the definition of free variables.
  Reduction is closed under all contexts.
  We write $\rewn\er$ for the recursive and transitive cloture of $\rew\er$.

  Values are defined as $\lhval \Coloneqq x \mid \lhlit \mid \lhConstr\ap\seq\lhval$,
  where $x$ is not a constant of $\seq\lhdecl$.
  Normal forms of the reduction of a term closed in the declaration environment are closed values.
\end{definition}

We assume that the reduction terminates, as \lh checks termination of programs.
Therefore, we can prove confluence simply by proving local confluence.

\begin{lemma}[Confluence]
  The reduction on $\er$ is confluent.
\end{lemma}
\begin{proof}
  We show local confluence and conclude by termination of the reduction.
  There are three kinds of critical pairs.
  The first is when
  $\ssred{\annlhqmark\lhtm{\lhtm_h}{\lhtm_p}}\lhtm$ and either $\lhtm_h$ or $\lhtm_p$ is reduced
  (as $\lhtm_h'$ or $\lhtm_p'$).
  We reduce $\ssred{\annlhqmark\lhtm{\lhtm_h'}{\lhtm_p'}} \lhtm$.
  The second kind is trivial and comes from the reduction of two subterms that
  do not interact, as $\lhtm_1 \lhbop \lhtm_2$ that can reduce both $\lhtm_1$ and $\lhtm_2$.

  The third kind concerns pattern matching and let, both cases are similar.
  We can reduce $\lhlet x \lhref \lhval \lhexp$ both as
  $\lhexp\subst x v$ and $\lhlet x \lhref \lhval {\lhexp'}$.
  To close the diagram with $\lhexp'\subst x \lhval$, we show the following property:
  given an expression $\lhexp$ such that $\ssred \lhexp {\lhexp'}$
  and a substitution of variables to values $\sigma$,
  we have $\ssred{\lhexp\sigma}{\lhexp'\sigma}$.
  This property is easily proved by induction on $\ssred{\lhexp}{\lhexp'}$.
\end{proof}

The graph relation of a function is created by reorganizing the structure of the
term into the function paths.
The function $\eqguardName$ from \cref{def:eqguard} only commutes pattern matchings,
merges them and eliminate superflous matchings.
Since the reduction is confluent, these operations are harmless for the final result:
we could define a function to convert the set of function paths back into an expression,
and both the original function and that expression would reach the same normal form.
We formalize this below with an additional detail:
to be closer to the representation of the graph relation in \coq,
we build a $\er$ expression from the function paths, but where we also extract
function and operator calls into let-bindings.

\begin{definition}
  Consider a set $P \doteq \{(\sigma_\params^i, \sigma_\brSub^i, \lhtm_i)\}_{i \leq n}$
  of paths obtained by $\eqguardName$.
  We transform this set back into an expression with
  $\caseify(P) \doteq \caseifyaux(\{\sigma_i, \lhtm_i\}_{i \leq n}, \emptyset)$,
  where $\sigma_i$ is the concatenation of $\sigma_\params^i$ and $\sigma_\brSub^i$
  where we destructed back the pattern: for instance,
  a tuple $(\lhtm, \lhConstr\ap(\lhConstr_1\ap x_1)\ap x_2\ap \lhConstr_2)$ is replaced by the sequence
  $(\lhtm, \lhConstr\ap y \ap x_2 \ap z), (y, \lhConstr_1\ap x_1), (z, \lhConstr_2)$,
  for fresh variables $y$ and $z$.

  To mimic the construction of the graph relation (\cref{def:altInGraphRel}), we consider
  the map $\frApp$ from calls to free variables and
  we assume a function $\repl(\lhtm)$ that replaces the applications in $\lhtm$ according to $\frApp$.
  \begin{align*}
    \caseifyaux(\{(\emptyset, \lhtm)\}) \doteq
    &\;\lets{z}{\repl(\lhtm_z)}{(z,\lhtm_z) \in H}{\repl(\lhtm)}
    \\\caseifyaux(\{((\lhtm,\lhtm_i) : \sigma_i^j, \lhtm_i^j)\}_{j \leq m_i, i \leq n}) \doteq
    &\;\lets{z}{\repl(\lhtm_z)}{(z,\lhtm_z) \in H'}{
  \\&\; \lhcase {\lhtm'} {(\case{\lhtm_i}{
\caseifyaux(\{(\sigma_i^j, \lhtm_i^j)\}_{j \leq m_i}, H})_{i \leq n}}}
  \end{align*}
  Where in the second case $\lhtm' \doteq \repl(\lhtm)$ if $\lhtm$ is an application
  and $\lhtm' \doteq x$ if $\lhtm = x$, and in both cases
  $H \doteq \{(z,\lhtm_z) \in \frApp \mid \lhtm_z \text{ is a subterm of } \lhtm\}$.
\end{definition}

\begin{lemma}
  \label{l:caseify}%
  Let $\lhdecls \doteq \lhdecls'; \lhrefl f \arrs \lhexp \in \lhdecls$
  be a well-formed and annotated sequence of reflected declarations.
  Let $\lhval_1, \dots, \lhval_n$ be closed under $\lhdecls$.
  Then, $\ssredn{\lhexp\subst{x_i}{\lhval_i}_{i \leq n}}\lhval$ if and only if
  $\ssredn{\caseify(\mkGraph(\lhexp; (x_i)_{i \leq n}))\subst{x_i}{\lhval_i}_{i \leq n}}\lhval$.
\end{lemma}
\begin{proof}
  By confluence of the reduction, $\eqguardName$ preserves the result of evaluation.
  Confluence is what allows us to extract applications inside let-bindings.
\end{proof}

\operationalmain*
\begin{proof}
  The additional property concerning the refined translation of the function
  is proven directly in \coq in the lemmas we described in \cref{sec:transfunc}
  relating the refined definition and the graph relation on the refined domain.
  We make an assumption that the arguments are well-typed in \coq
  because the type system from this paper does not check refinements.
  However, we can ask that these values are of the correct underlying unrefined type,
  and this underlying type is preserved by the translation.
  The fact that refined types correspond in \lh and \coq is proved in \cref{l:same-model}.

  We now prove the equivalence between our semantics and the semantics of the graph relation.
  Notice that $\utrLhtmCq\lhval$ and each $\utrLhtmCq{\lhval_i}$ are
  of the correct type $\utrLhrefCq\lhref$ and $\utrLhrefCq{\lhref_i}$ in \coq.
  Also, the equivalence holds for operators and their graph relation because the latter is defined using the corresponding \coq's builtin operators.

  Let $\ssredn{f \ap \lhval_1 \dots \lhval_n}{\lhexp_f\subst{x_i}{\lhval_i}_{i \leq n}}$.
  According to the previous lemmas, we can reduce the reduction
  (in the hypothesis for $\Rightarrow$ or that we construct for $\Leftarrow$)
  from the instantiation of $\lhexp_f$ above to a reduction on
  $\caseify(\mkGraph(\lhexp_f; (x_i)_{i \leq n}))$.
  For $\Rightarrow$, since $f \ap \lhtm_1 \dots \lhtm_n$ reduces to a value,
  none of the intermediate computations is stuck on a pattern matching.
  Similarly for $\Rightarrow$, since there is a result to the application of
  the relation and by functionality of the relation, there is a unique
  constructor that can be traced from the arguments and the result.

  Therefore, there is a path
  $\branch = (\sigma_\params, \sigma_\lhtm, \lhtm_\branch) \in \mkGraph(\lhexp_f; (x_i)_{i \leq n})$
  corresponding to the reduction, and a constructor $f^\branch$ built from that path.

  Both sides of the proof are by induction on the reduction.
  We start with $\Rightarrow$.
  The structure of the function $\caseify$ models that of the constructor $f^\branch$.
  From the matching of the arguments, we obtain an instantiation of the free
  variables of the patterns in $\sigma_\params$, with which we instantiate the first
  variables of $f^\branch$.

  Then, let us consider the matchings on additional applications.
  The hypotheses introduced in $f^\branch$ are in one-to-one correspondence
  with the let-bindings introduced in $\caseify$.
  These hypotheses are of the shape
  $\exists z, g_0\ap \cqexp_1\ap \dots \cqexp_m\ap \cqexp$,
  where $g_0$ is $\relnm g$ or $\getUPackRel g$ or $\getPackRel g$
  and where $\cqexp_1$ to $\cqexp_n$ and $\cqexp$ are \coq expressions obtained by translation.
  The encoding of a relation inside a pack preserves its operational semantics,
  so that in all three cases we can apply the induction hypothesis,
  since this specific application of $g$ appears in the translation.
  This gives us an instantiation from the fresh variable of the hypothese
  to the result of the application.
  The evaluation of the pattern matching itself gives us an instantiation of the
  free variables of the pattern.

  Now, let $\sigma$ be the substitution obtained by the series of pattern matching
  on additional applications and the let-bindings corresponding to function calls.
  Given $z$ in $\dom(\sigma)$, there are two possilities:
  either $\sigma(z) = \lhval_z$ for
  $\ssredn{g \ap \seq{\lhval_g}}{\lhval_z}$,
  where $g$ was bound in a let-binding,
  and $z$ comes from $\frApp$,
  or $\sigma(z)$ is an arbitrary value,
  when $z$ is a parameter or introduced with a pattern.
  For the return type of the constructor $f^\branch$,
  we instantiate the first $n$ arguments of $\relnm f$
  correctly thanks to the previous instantiations
  (from $\sigma_\params$ and from $\sigma$ above).
  For the $n+1$th argument (the result),
  we use the fact that for every refinement $\lhtm$,
  if there is $\lhval$ such that $\ssredn\lhtm\lhval$, then
  $\utrLhtmCq{\lhtm}\sigma$ is equal modulo conversion to $\lhval$ in \coq.
  This proof is direct by induction on $\ssredn\lhtm\lhval$.

  The proof for $\Leftarrow$ is similar: there is a correspondance between the
  graph relation in \coq and the reduction on the $\caseify$ term,
  from which, using \cref{l:caseify}, we get a reduction of the original $\er$ function.
  By inversion on $\relnm f \ap \utrLhtmCq{\lhval_1} \dots \utrLhtmCq{\lhval_n} \ap \utrLhtmCq\lhval$
  and by functionality of the relation (proved directly in \coq),
  we get a unique branch of the relation corresponding to the correct path.
  This branch gives us instantiations for the parameters of the function
  and further hypothesis of the shape
  $\relnm g \ap \utrLhtmCq{\lhval_1'} \dots \utrLhtmCq{\lhval_m'} \ap \utrLhtmCq{\lhval'}$.
  By induction, these hypothesis correspond to a computation
  $\ssredn{g \ap \lhval_1' \dots \lhval_m' \lhval'}$.
\end{proof}

\subsection{Preservation of Models of Types}
\label{app:models}%

\subsubsection{Simply-typed \texorpdfstring{$\lambda$}{lambda}-calculus}
We take the simply-typed $\lambda$-calculus ($\stlc$) as a common language to interpret
terms of $\er$ and of \coq.

We use terms of the simply-typed $\lambda$-calculus as a common interpretation
for $\er$ and \coq.

\begin{definition}[Calculus $\stlc$]
  The grammar of $\stlc$ is given in \cref{fig:stlc-grammar}.
  Operators and literals are taken to be the same as in $\er$.
  \begin{figure}
    \[ \begin{array}{rrclll}
      \gramTitle{Typing environments}
        & \Gamma & \Coloneqq &
        \emptyset \mid \Gamma, x:\lctyp \mid \Gamma, \lcdt\lcTC{\seq{(\lcConstr:\lctyp)}}
        \\ \gramTitle{Types}
        & \lctyp & \Coloneqq &
        \Unit \mid \Bool \mid \Int \mid \lcTC \mid \lctyp \to \lctyp
        \\ \gramTitle{Expressions}
        & \lcexp & \Coloneqq &
        x \mid \lcval
        \mid \lhConstr
        \mid \lcexp \ap \lcexp
        \mid \lcexp \lcbop \lcexp \mid \lcnegb \lcexp
        \mid \lcexp = \lcexp
     \\ &&&\mid \lccase \lcexp {\seq{(\case{\lcConstr\ap\seq y}{\lcexp})}}
        \mid \lclet x \lcexp \lcexp
        \\ \gramTitle{Values}
      & \lcval & \Coloneqq &
      \lclit \mid \lcConstr\ap\seq v \mid \lclam x \lcexp
      \\ \gramTitle{Primitives}
      & \lccst & \Coloneqq &
      \unit \mid \true \mid \false \mid 0, 1, -1 \dots
\end{array} \]
\caption{Grammar of the simply-typed $\lambda$-calculus}
\Description{Grammar of the simply-typed lambda-calculus}
\label{fig:stlc-grammar}
\end{figure}
\end{definition}

We define two translations to the $\lambda$-calculus: from $\er$ and from the
image of the translation of type refinements.
The image of the translation of refinement types of $\er$ is formally defined in \cref{fig:ec}.
In this section, we use the simplified definition of pack without encoding (see \cref{app:pack}),
as both versions have the same semantics.

\begin{figure}
  \[\begin{array}{rrcll}
    \gramTitle{Base types}
        & \cqtyp & \Coloneqq & \cqUnit \mid \cqBool \mid \cqInt \mid \unrefnm\cqTC
        \\ \gramTitle{Subset types}
        & \cqref & \Coloneqq &
        \refT x \cqtyp \cqprop
        \mid \forall (x : \cqref), \cqref
        \mid \fClass (\forall \seq{(x:\cqref)}, \cqref)
        \\ \gramTitle{Refinements}
        & \cqtm & \Coloneqq &
        x \mid \proj x \mid \cqlit \mid \unrefnm\cqConstr \ap \seq\cqtm \mid \cqtm \cqeqb \cqtm
        \\ \gramTitle{Literals}
        & \cqlit & \Coloneqq &
        \cqunit \mid \cqtrue \mid \cqfalse \mid 0, 1, -1, 2 \ldots
        \\ \gramTitle{Prop operators}
        & \cqlop & \Coloneqq &
        \land \mid \lor \mid \to \mid \leftrightarrow
        \\ \gramTitle{Propositions}
        & \cqprop & \Coloneqq &
        \IsTrue(\cqtm) \mid \neg\cqprop \mid \cqprop \cqlop \cqprop
        \mid \wffnm\cqTC \ap x
        \mid \cqtm \cqeqp \cqtm
     \\ &&\mid& \exists z, (\relnm f \mid \relnmInst f \mid \relnmUInst f) \ap \seq\cqtm \ap z \land \cqprop
  \end{array}\]
  \caption{Image of the translation of types in \coq.}
  \Description{Image of the translation of types in Rocq}
  \label{fig:ec}
\end{figure}

\begin{definition}[Translation from $\er$ types to the $\lambda$-calculus]
  We define a translation $\trLhLc\cdot$ from $\er$ to $\stlc$ over:
  \begin{itemize}
    \item Base types as the identity and refinement types as
      $\trLhLc{\refT x \lhtyp \lhtm} \doteq \lhtyp$ and
      $\trLhLc\arr \doteq \trLhLc{\lhref_x} \to \trLhLc\lhref$.
    \item Refinements as
      $\trLhLc{\vararlg x \_ \_} \doteq x$,
      $\trLhLc{\lhinj\lhtm\_} \doteq \lhtm$,
      $\trLhLc{\lhsub\lhtm\_\_} \doteq \lhtm$,
      $\trLhLc{\lhproj\lhtm} \doteq \lhtm$
      and the identity for the rest.
    \item Expressions as
      $\trLhLc{\lhlet x \lhref {\lhexp_x} \lhexp} \doteq \lclet x {\trLhLc{\lhexp_x}} {\trLhLc\lhexp}$,
      \\$\trLhLc{\ite{\lhtm}{\lhexp_t}{\lhexp_f}} \doteq
      \lccase {\trLhLc\lhtm} {\case{\True}{\trLhLc{\lhexp_t}}; \case{\False}{\trLhLc{\lhexp_f}}}$,
      \\$\trLhLc{\lhcase \lhtm {(\case{\lhConstr_i\;\seq{x_i}}{\lhexp_i})_{i \leq n}}}
      \doteq {\lccase {\trLhLc\lhtm} {(\case{\lhConstr_i\;\seq{x_i}}{\trLhLc{\lhexp_i}})_{i \in I}}}$
      where $I = \{i \leq n \mid \lhexp_i \neq \dead\}$ and
      $\trLhLc{\annlhqmark{\lhexp}{\lhtm_h}{\lhtm_p}} \doteq \trLhLc\lhexp$.
    \item Typing environments with
      $\trLhLc{x:\lharr} \doteq x:\trLhLc\lharr$ and
      $\trLhLc{\lhdt{\lhTC}{(\lhConstr_i:\lharr_i)_{i \in I}}}
      \doteq \lcdt\lhTC{(\lhConstr_i:\trLhLc{\lharr_i})_{i \in I}}$.
  \end{itemize}
\end{definition}

\begin{definition}[Translation from the target $\coq$ grammar to the $\lambda$-calculus]
  We define a translation $\trCqLc\cdot$ from the \coq target to $\stlc$ over:
  \begin{itemize}
    \item Base types:
      $\trCqLc\cqUnit \doteq \Unit$,
      $\trCqLc\cqBool \doteq \Bool$,
      $\trCqLc\cqInt \doteq \Int$ and
      $\trCqLc{\unrefnm\cqTC} \doteq \cqTC$.
    \item Refinement types:
      $\trCqLc{\refT x \cqtyp \cqexp} \doteq \trCqLc\cqtyp$,
      $\trCqLc{\forall (x : \cqref_x), \cqref} \doteq \trCqLc{\cqref_x} \to \trCqLc\cqref$
      and $\trCqLc{\fClass (\forall (x:\cqref_i)_{i \leq n}, \cqref)}
      \doteq \trCqLc{\cqref_1} \to \dots \to \trCqLc{\cqref_n} \to \trCqLc{\cqref}$.
    \item Literals:
      $\trCqLc\cqunit \doteq \unit$,
      $\trCqLc\cqtrue \doteq \true$,
      $\trCqLc\cqfalse \doteq \false$ and
      $\trCqLc\cqlit \doteq \cqlit$ if $\cqlit$ is an integer.
    \item Refinements:
      $\trCqLc x \doteq x$, $\trCqLc{\proj x} \doteq x$,
      $\trCqLc{\unrefnm\lhConstr\ap\seq\cqtm} \doteq \lhConstr\ap\seq{\trCqLc\cqtm}$ and
      $\trCqLc{\cqtm_1 \cqeqb \cqtm_2} \doteq \trCqLc{\cqtm_1} = \trCqLc{\cqtm_2}$.
    \item Propositions:
      $\trCqLc{\IsTrue(\cqtm)} \doteq \trCqLc\cqtm$,
      $\trCqLc{\cqprop_1 \cqlop \cqprop_2} \doteq \trCqLc{\cqprop_1} \cqlop
      \trCqLc{\cqprop_2}$ (where $\cqlop$ is the builtin binary
      operator in $\stlc$),
      $\trCqLc{\cqprop_1 \cqeqp \cqprop_2} \doteq \trCqLc{\cqprop_1} = \trCqLc{\cqprop_2}$,
      \\$\trCqLc{\exists z, (\relnm f/\relnmInst f/\relnmUInst f) \ap\seq\cqtm \ap z \land \cqprop}
      \doteq \lclet z {f \ap \seq{\trCqLc\cqtm}} {\trCqLc\cqprop}$,
      \\ $\trCqLc{\wffnm \lhTC \ap x}
      \doteq \lccase x {(\case{\lhConstr_i\ap x_1 \dots
      x_{m_i}}{\trCqLc{\cqprop_i}})_{i \in I}}$
      if $\wffnm\cqTC x$ is defined as\\
      $\cqmatch x {\left(\case {\lhConstr_i \ap x_1 \dots x_{m_i}}
      \cqprop_i \right)_{i \in I}}$.
  \end{itemize}
\end{definition}

\begin{definition}
  The reduction semantics of the calculus is called $\rew\stlcredName$ and defined in \cref{fig:stlc-eval}.
  We reduce under all contexts except under abstractions.
  Evaluation is very similar to the one of $\er$.
  It is parametrized by an environment $\lhdecls$ of elaborated $\er$ declarations.
  We use a similar function $\delta$ giving the result of the evaluation of operators.
  \begin{figure}
    \[\begin{array}{ll}
      \stlcred {\lcexp} {\delta(\lcexp)}
      \quad\text{ if } \lcexp \in \{\lcnegb\lclit, \lclit_1 \lcbop \lclit_2, \lcval_1 \lceqb \lcval_2\}
    & \stlcred {(\lclam x \lcexp_1) \lcexp_2} {\lcexp_2 \subst x {\lcexp_1}}
      \\\stlcred x \lclam{x_1 \dots x_n}{\trLhLc\lcexp} &\text{if } \lhrefl f \arrs \lhexp \in \lhdecls
    \\\stlcred{\lclet x \lcval \lcexp} {\lcexp \subst x \lcval}
    \\\stlcred{\lccase {\lcConstr_k \ap \seq{\lcval_k}} {(\case{\lcConstr_i \ap \seq{x_i}}{\lcexp_i})_{i \leq n}}}
    {\lcexp_k \seq{\subst{x_k}{\lcval_k}}}
    \end{array}\]
    \caption{Reduction relation for the simply-typed $\lambda$-calculus}
    \Description{Reduction relation for the simply-typed lambda-calculus}
    \label{fig:stlc-eval}
  \end{figure}
\end{definition}

In the translation from \coq refinements,
an existential hypothesis $\exists z, \relnm f \ap\seq\cqtm \ap z \land \cqprop$
(and variants) is translated back to the function $f$,
that is a function defined in the context of the original \lh function.
From the point of view of reduction, this is correct thanks to \cref{l:operationalmain}:
both the graph relation and the \haskell function have the same input-output relation.
We know that $z$ exists \textit{iff} the application of $f$ results in a value $v$,
a value for which we can instantiate $z$.
On the other hand, if a result does not exist, this means that the application of $f$
is stuck in the translation.
If this is the case, the computation will be stuck, as we only reduce lets when
they bind a value.
Thus, the refinement will not evaluate to $\True$.

We define two interpretations of types of $\er$ and of the \coq target
respectively to sets of values of $\stlc$.
The interpretations are parametrized by an assignement $\rho$ from variables to
values of $\stlc$ and by a $\er$ environment of declarations $\lhdecls$ of
elaborated datatypes and functions.
The values in the interpretation are closed with respect to $\lhdecls$.
We use the notation $\fnreplace \rho x \lcval$ to indicate the update of the
value of $x$ to $\lcval$ in $\rho$.

In the following, we use the name $\typEnv_\lhdecls$
for the environment obtained from the elaboration into $\lhdecls$:
$\typEnv_\lhdecls$ contains the elaborated specifications of the functions and datatypes in $\lhdecls$.

\begin{definition}[Interpretation of $\er$ types]
  \label{def:lhtpint}%
  \begin{align*}
    \lhtpint\lhtyp
        &\doteq \{\lclit \mid \lctyj\emptyset\lclit\lhtyp\}
        \qquad \text{$\lhtyp$ is a builtin type}
        \\ \lhtpint\lhTC
        &\doteq \{ \lcval = \lhConstr\ap\lcval_1 \dots \lcval_n \mid
        \\&\quad\exists(\lhref_i)_{i \leq n}, x, \lhtm.\;
          \lhdata\lhTC{\dots, \lhConstr :: \arrs[x_i][\lhref_i][i \leq n][\refT x \lhTC \lhtm], \dots} \in \lhdecls
     \\ &\quad\land (\forall i \in [1..n].\,\lcval_i \in
     \lhtpint{\lhref_i}[\fnreplace \rho {x_j}{\lcval_j}_{j < i}])
   \land \stlcredn {\lhtm\subst{x_i}{\lcval_i}_{i \leq n}\subst{x}{\lcval}\rho} {\true} \}
   \\ \lhtpint{\refT x \lhtyp \lhtm}
        &\doteq \{\lcval \mid \lcval \in \lhtpint\lhtyp
        \land \stlcredn{\lhtm \subst x \lcval \rho} {\true} \}
        \\ \lhtpint\arr
        &\doteq \{\lcval \mid \lctyj{\trLhLc{\typEnv_\lhdecls}}\lcval{\trLhLc{\lhref_x} \to \trLhLc\lhref}
     \\ &\quad\land \forall \lcval_x \in \lhtpint{\lhref_x}.\;
     \stlcredn {\lcval \ap \lcval_x} {\lcval'}
   \text{ s.t. } \lcval' \in \lhtpint{\lhref}[\fnreplace \rho {x}{\lcval_x}] \}
  \end{align*}
\end{definition}

\begin{definition}[Interpretation of \coq types]
  \label{def:cqtpint}%
  \begin{align*}
    \cqtpint{\refT x \cqtyp \cqprop}
        &\doteq \{\lcval \mid \lctyj{\trLhLc{\typEnv_\lhdecls}}\lcval{\trCqLc\cqtyp}
        \land \stlcredn{\trCqLc\cqprop \subst x \lcval \rho} {\true} \}
        \\ \cqtpint{\forall (x:\cqref_x), \cqref}
        &\doteq \{\lcval \mid \lctyj{\trLhLc{\typEnv_\lhdecls}}\lcval{\trCqLc{\cqref_x} \to \trCqLc\cqref}
     \\ &\quad\land \forall \lcval_x \in \cqtpint{\cqref_x}.
     \ \stlcredn{\lcval \ap \lcval_x} {\lcval'}
   \text{ s.t. } \lcval' \in \cqtpint{\cqref}[\fnreplace \rho x {\lcval_x}] \}
    \\ \cqtpint{\fClass (\forall (x_i:\cqref_i)_{i \leq n}, \cqref)}
        &\doteq \cqtpint{\forall (x_i:\cqref_i)_{i \leq n}, \cqref}
  \end{align*}
\end{definition}

\begin{lemma}
  \label{l:tolc_correspond}%
  For all refinement $\lhtm$, $\lhtm = \trLhLc\lhtm = \trCqLc{\utrLhtmCqprop\lhtm}$.
\end{lemma}
\begin{proof}
  The first equality is by definition.
  For the second, notice that $\trCqLc{\utrLhtmCqprop\lhtm}$ does not contain
  a call to $\wffnm\lhTC$.
  It is easy to prove that $\trCqLc{\utrLhtmCqprop\lhtm}$ is equal to $\lhtm$
  where the outmost applications of functions, operators and top-level constants
  have been replaced by the corresponding fresh variable in $\frApp$.
  Let $(\cqprop_1, \dots, \cqprop_n) \doteq \hypFromAm{\lhtm}$ and for each $i$,
  \[\cqprop_i \doteq \exists z_i, (\relnm{f_i}/\relnmInst{f_i}/\relnmUInst{f_i})
  \ap \seq{\utrLhtmCq{\lhtm_i}} \ap z_i\]
  Then, we have
  \begin{align*}
        \trCqLc{\utrLhtmCqprop\lhtm}
      &= \trCqLc{\cqprop_1 \land \dots \land \cqprop_n \land \IsTrue(\utrLhtmCq\lhtm)}
      \\&= \trCqLc{\utrLhtmCq\lhtm}
      \subst{z_i}{f_i\ap\seq{\trCqLc{\utrLhtmCq{\lhtm_i}}}}_{i \leq n}
      \\&= \lhtm
      \qedhere
  \end{align*}
\end{proof}

\begin{lemma}
  \label{l:same-model-aux1}%
  Let $\lcval \in \lhtpint\lhtyp$.
  Then, $\lctyj{\trLhLc{\typEnv_\lhdecls}}\lcval\lhtyp$.
\end{lemma}
\begin{proof}
  By structural induction on $\lhtyp$.
  The case $\lhtyp$ is builtin is direct by definition of the models.
  If $\lhtyp = \lhTC$, then $\lcval = \lhConstr\ap\lcval_1\dots\lcval_n$,
  we have $\lhConstr :: (x_i:\refT{x_i}{\lhtyp_i}{\lhtm_i})_{i \leq n} \to \refT x \lhTC \lhtm$
  in $\lhTC$.
  Let $\rho_i \doteq \fnreplace \rho {x_j}{\lcval_j}_{j < i}$.
  By hypothesis, for each $i \leq n$, we have
  $\lcval_i \in \lhtpint{\lhref_i}[\rho_i]$ which implies by induction
  that $\lctyj{\trLhLc{\typEnv_\lhdecls}}{\lcval_i}{\lhtyp_i}$.
  Thus, by applications and since $\trLhLc{\typEnv_\lhdecls(\lhTC)(\lhConstr)}
  = \lhtyp_1 \to \dots \to \lhtyp_n \to \lhTC$,
  we derive $\lctyj{\trLhLc{\typEnv_\lhdecls}}{\lcval}{\lhTC}$.
\end{proof}

\samemodel*
\begin{proof}
  By induction on the structure of the type.
  First, let $\lhref \doteq \refT x \lhtyp {\lhtm_x}$.
  If $\lhtyp$ is a builtin, we are done because $\trCqLc\cqprop =
  \trCqLc{\utrLhtmCqprop{\lhtm_x}} = \lhtm_x$
  (by \cref{l:tolc_correspond})
  and $\lcval \in \lhtpint\lhbuiltin$ is clearly equivalent to
  $\lctyj{\trLhLc{\typEnv_\lhdecls}}\lcval\lhbuiltin$.

  Let $\lhtyp \doteq \lhTC$ and thus $\utrLhrefCq{\lhref}
  = \refT x {\unrefnm\lhTC} {\wffnm\lhTC \ap x \wedge \utrLhtmCqprop{\lhtm_x}}$.
  By \cref{l:tolc_correspond}, we have
  $\trCqLc{\utrLhtmCqprop{\lhtm_x}} = \lhtm_x$,
  so on both sides we have $\stlcredn {\lhtm_x \subst x \lcval} \true$.
  Then, we prove:
  \begin{itemize}
    \item $\lcval \in \lhtpint\lhref \implies \lcval \in
      \cqtpint{\utrLhrefCq\lhref}$: 
      By definition, $\lhval$ is of the shape $\lhConstr \ap \lcval_1 \dots \ap \lcval_n$.
      From \cref{l:same-model-aux1}, we have
      $\lctyj{\trLhLc{\typEnv_\lhdecls}}\lcval\lhTC$.
      We still need to prove \[
        (\forall i \in [1..n].\,\lhval_i \in \lhtpint{\lhref_i}[\rho_i])
        \land \stlcredn {\lhtm_\lhConstr \sigma \rho} \true
        \implies \stlcredn {\trCqLc{\wffnm\lhTC x} \subst x v \rho} \true
      \]
      Where for each $i$, $\rho_i = \fnreplace \rho {x_j}{\lcval_j}_{j < i}$,
      where $\sigma = \subst{x_i}{\lcval_i}_{i \leq n} \subst x \lcval$,
      and with $(\lhConstr) :: (x_i:\lhref_i)_{i \leq n} \to \refT
      {x_\lhConstr}{\lhTC}{\lhtm_\lhConstr}$ in $\lhTC$ in $\lhdecls$.

      Let $\lhref_i \doteq \refT{x_i'}{\lhtyp_i}{\lhtm_i}$ for each $i \leq n$
      (there are no higher-order arguments of data constructors in our
      system). By definition, \[
        \stlcredn{\trCqLc{\wffnm\lhTC \ap x}\subst x \lcval}
        {\trCqLc{\utrLhtmCqprop{\lhtm_\lhConstr}}\sigma}
        \wedge \bigwedge_{i \leq n} \trCqLc{\cqprop_i\subst{x_i'}{x_i}}\sigma
      \] where $\cqprop_i$ is such that
      $\utrLhrefCq{\lhref_i} = \refT{x_i'}{\trBuiltin{\lhtyp_i}}{\cqprop_i}$.
      By \cref{l:tolc_correspond},
      $\trCqLc{\utrLhtmCqprop{\lhtm_\lhConstr}} = \lhtm_\lhConstr$,
      and therefore $\stlcredn{\trCqLc{\utrLhtmCqprop{\lhtm_\lhConstr}}\sigma\rho} \true$.
      For each $i \leq n$, we obtain by induction hypothesis $\lcval_i \in
      \cqtpint{\lhref_i}$, and thus in particular
      $\stlcredn{\trCqLc{\cqprop_i\subst{x_i'}{x_i}}\rho_i} \true$
      and we conclude since $\fv(\lcval_i) \cap \dom(\sigma) = \dom(\rho_i)$.

    \item $\lcval \in \cqtpint{\utrLhrefCq\lhref} \implies \lcval \in \lhtpint\lhref$:
      By hypothesis, we have $\lctyj{\trCqLc{\typEnv_\lhdecls}}{\lcval}{\lcTC}$
      and $\stlcredn {\trCqLc{\wffnm\lcTC \ap x}\subst x \lcval \rho} \true$.
      Since we have a derivation for $\lcval$ of type $\lcTC$,
      $\lcval$ can only be of the shape $\lhConstr \ap \lcval_1 \dots \ap
      \lcval_n$, and there must be $(\lhref_i)_{i \leq n}$, $x$ and $\lhtm_\lhConstr$ such
      that $\typEnv(\lcTC)(\lhConstr) = \arrs[x_i][\lhref_i][i \leq n][\refT x
      \lcTC {\lhtm_\lhConstr}]$.
      We need to prove: (1) $\stlcredn {\lhtm_\lhConstr \sigma \rho} \true$
      and (2) $\forall i \in [1..n].\,\lhval_i \in
      \lhtpint{\lhref_i}[\rho_i]$.

      For (1), $\stlcredn{\trCqLc{\wffnm\lcTC \ap x}\subst x \lcval \rho} \true$
      implies as before $\stlcredn{\trCqLc{\utrLhtmCqprop{\lhtm_\lhConstr}}\sigma \rho =
      \lhtm_\lhConstr\sigma\rho} \true$ by \cref{l:tolc_correspond}.
      For (2), let $i \leq n$.
      Then $\lhref_i$ is a refinement of a base type by restriction of the
      system, let $\lhref_i \doteq \refT{x_i'}{\lhtyp_i}{\lhtm_i}$
      and $\utrLhrefCq{\lhref_i} \doteq
      \refT{x_i'}{\trBuiltin{\lhtyp_i}}{\cqprop_i}$.
      Again, $\stlcredn{\trCqLc{\wffnm\lcTC \ap x}\subst x \lcval \rho} \true$
      implies $\stlcredn{\trCqLc{\cqprop_i} \sigma \rho = \lhtm_\lhConstr\sigma\rho} \true$
      by \cref{l:same-model-aux1}.
      Moreover, since there is a derivation for $\lcval$, there is also a
      derivation for $\lcval_i$.
      Hence, $\lcval_i \in \cqtpint{\lhref_i}$, and thus by induction
      $\lcval_i \in \lhtpint{\lhref_i}$.
  \end{itemize}

  Now, consider $\lhref = \arr[x][\lhref_x][\lhref']$.
  If $\lhref_x$ is an arrow type, it is translated to an instance of $\fClass$.
  The translations of $\fClass$ and dependent arrows of \coq and of $\er$ in
  $\stlc$ are all equal.
  So we conclude by induction hypothesis, since the interpretations are built in the same way.
\end{proof}

  \section{Implementation of Packs}\label{app:pack}
  Our description of \fClass{}s in the main text was simplified to records containing
the refined definition, graph relation, the lemma relating them and functionhood
for the graph relation.
In reality, to be able to implement a generic \fClass{} typeclass for all function arities
and for higher-order arguments (becoming nested \fClass{}s), we encode function arguments as a list.
In this appendix, we give more details about the implementation of \fClass{}s,
and precise the translations.

\subsection{Definition of Function \fClass{}s}
\label{sec:pack}

The argument types of a function are stored in a dependent list,
the type $\ArgListT$ is the type of dependent lists of argument types.
The type $\ArgList$ represents a dependent list of terms of those types given by a term $\mathtt{argTps}:\ArgListT$.
Analogously, the types $\UArgListT$ and $\UArgList$ represent respectively lists
of unrefined types, and lists of arguments of a term
$\mathtt{uargTps}:\UArgListT$.
The relation $\projectsArgListT$ indicates whether an
$\UArgListT$ contains exactly the unrefined types of a given $\ArgListT$ and the
function $\prArgList$ extracts an $\UArgList$ out of an $\ArgList$.

We use these encodings of arrow types to define the typeclasses $\fClass$
(refined) and $\fuClass$ (unrefined) in \coq.
\begin{definition}[Function \fClass{}s]
  \label{def:pack}
  The definition of the typeclasses is given in \cref{fig:pack-classes}.
  In the parameters, \rocq{argTps} is the dependent list of the refined argument types,
  \rocq{uargsTps} the list of unrefined argument types, related to \rocq{argTps} by the projection \rocq{prj}.
  The unrefined return type is \texttt{T}, refined by the predicate \texttt{p}.
  \begin{listing}
  \begin{rocqcode}
Class !$\fClass$! (argTps:!$\ArgListT$!) {uargTps:!$\UArgListT$!} {prj:!$\projectsArgListT$! argTps uargTps}
(T: Type) (p: !$\forall$! (args: !$\ArgList$! argTps), T !$\to$! Prop) := {
  !$\defnmClass$! (args:!$\ArgList$! argTps): {v:T | p args v}; (* Refined definition *)
  !$\relnmClass$! (uargs: !$\UArgList$! uargTps) (v:T): Prop; (* Graph relation *)
  !$\unrefproofClass$! (args:!$\ArgList$! argTps) (v:T): !$\lfloor\defnmClass$! args!$\rfloor$! = v !$\leftrightarrow$! !$\relnmClass$! (!$\prArgList$! args uargTps prj) v;
  !$\functrelClass$! (uargs: !$\UArgList$! uargTps) v v': !$\relnmClass$! uargs v !$\to$! !$\relnmClass$! uargs v' !$\to$! v = v' }.

Class !$\fuClass$! (uargTps: !$\UArgListT$!) (T: Type) := {
  !$\relnmUClass$! (uargs: !$\UArgList$! uargTps): T !$\to$! Prop; (* Graph relation *)
  !$\functrelUClass$! (uargs:!$\UArgList$! uargTps) (v v':T): !$\relnmUClass$! uargs v !$\to$! !$\relnmUClass$! uargs v' !$\to$! v = v' }.
  \end{rocqcode}
  \caption{Typeclass for refined and unrefined function packs}
  \label{fig:pack-classes}
\end{listing}
\end{definition}
Because of the encoding, we use a tactic \texttt{buildPackG} to populate the
typeclass instances and tactics to project individual fields.
\tac{buildPackG} takes the refined definition, the unrefined graph relation,
and the two \fClass{} lemmas as parameters.
There is a similar tactic \texttt{buildUPackG} for an unrefined \fClass{}.
The field of the refined function is retrieved with the tactic \texttt{getPackF}
implicitely doing the conversion from the encoding.
The graph relation is retrieved with \texttt{getPackRel} from a refined \fClass{} and
\texttt{getUPackRel} from an unrefined \fClass{}.
There are also tactics to retrieve the lemmas.

The projection of a refined \fClass{} $x$ to an unrefined \fClass{} can still be defined as
$\uClassToClass x \doteq \cqrecuClass {\relnmInst x} {\functrelInst x}$,
despite the encoding.

\paragraph{Translations}

With the definition of \fClass{}s as using lists of arguments,
the translations of function types as \fClass{} are changed.
In the unrefined translation, we have the following,
using the list notation for an element of type $\UArgListT$.
\[
  \utrLhrefCq{\arrs} \doteq
  \fuClass \ap [\utrLhrefCq{\lhref_1}, \dots, \utrLhrefCq{\lhref_n}] \ap \utrLhrefCq\lhref
\]
In the refined translation, we have:
\[ \begin{split}
  &\trLhrefCq\arrs
\\&\quad\doteq\fClass\ [\bind{x_1}{\trLhrefCq{\lhref_1}},
  \dots, \trLhrefCq{\lhref_n}]
  \cqtyp \ap (\cqlamuntyped{\mathtt{args \ v}}{\ltacTac{\flattenP\ap p\ap args\ap v}})
\end{split}\]
Where $\refT x \cqtyp {p\ x_1\ \ldots\ x_n\ x} = \trLhrefCq\lhref$
and the first argument of $\fClass$ is a dependent list of type $\ArgList$.
We use a tactic \flattenP to transform the $n+1$-ary predicate $p$ into a predicate
that takes an argument list of length $n$ as first parameter
(if the predicate does not depend on its first $n$ parameters,
the translation avoids the tactic and directly inserts the $n+1$-ary predicate).

\begin{example}
  \label{ex:transtype}%
  \label{ex:pack_apply}%
  Consider the type $\lhref$ of the function \texttt{apply} from the overview (\cref{subsec:overview-ho}).
  \[
    \lhref \doteq
    \arr[f][(\arr[ys][\mathtt{IList}][\refT v \Int {0 \leq v}])][
    \arr[xs][\mathtt{IList}][\refT v \Int {v = f \ap xs}]]
  \]
  The type of the higher-order parameter $f$ is translated as:
  \[
    \cqref_f \doteq \fClass [\mathtt{IList}] \ap \cqInt \ap (\cqlamuntyped{\mathtt{args}\ v} {\cqmatch{\mathtt{args}}{[ys] => \relnm\lhleqZnm \ap 0 \ap v \ap \cqtrue}})
  \]
  And the full translation of $\lhref$ as the type of the \coq refined function is thus:
  \[
    \forall (f: \cqref_f) \ap (xs: \mathtt{IList}),
    \refT{v}{\cqInt}{\exists z, (\getPackRel f) \ap \proj{xs} \ap z \land v \cqeqp z}
  \]
\end{example}

\subsection{Subsumption for \fClass{}s}

In this section, we describe subsumption for \fClass{}s.
In order to uniformly translate subsumption casts of subset types and of
function \fClass{}s inside the list of arguments of a function, we define a typeclass
for subsumption casts.
This class has two fields: the cast result and a property that cast does not
affect projection (since subtyping only affects refinements).
\begin{rocqcode}
Class !$\gSub$! {U: Type} (S T : Type) {prS: !$\gProj$! S U} {prT: !$\gProj$! T U} (x: S) :=
  { !$\gSubVal$!: T; cast_pr: !$\proj\gSubVal$! = !$\proj{\texttt{x}}$! }.
\end{rocqcode}

The instance $\subRef$ for subset types is defined using the subsumption function given in \cref{sec:casts}.

To define an instance for \fClass{}s, we first define an inductive type to record the
requirements for a successful subsumption cast between \ArgList{}s:
\begin{rocqcode}
Inductive SubArgList: ArgListT !$\to$! ArgListT !$\to$! Type :=
| noArgsSub: SubArgList noArgsT noArgsT
| consArgsSub (U S T:Type) (prS: !$\gProj$! S U) (prT: !$\gProj$! T U)
    (argsS: !$\forall$!(x:S), ArgListT) (argsT: !$\forall$!(x:T), ArgListT) (cast: !$\forall$!(x:S), !$\gSub$! S T x) (x:S)
  : SubArgList (argsS x) (argsT (cast x).(!$\gSubVal$!)) !$\to$! SubArgList (S ::RT argsS) (T ::RT argsT).
\end{rocqcode}
Where \texttt{noArgsT} and \texttt{::RT} are the constructor for \texttt{ArgListT}.

The instance $\subPack$ for \fClass{}s has the signature below.
Its definition is technical but follows from the signature.
\begin{rocqcode}
Instance !$\subPack$! {argTps: ArgListT} {T: Type} {p: !$\forall$! x:ArgList argTps, T !$\to$! Prop}
  {argTps': ArgListT} {q: !$\forall$!(x:ArgList argTps'), T !$\to$! Prop} {uargTps: UArgListT}
  {z: projectsArgListT argTps uargTps} {z': projectsArgListT argTps' uargTps}
  (fpack: !$\fClass{}$! argTps uargTps z T p) (castArg: SubArgList argTps' argTps)
  (castT: !$\forall$! (args:ArgList argTps) (args':ArgList argTps'),
    !$\gSub$! {v:T | p args v} {v:T | q args' v} (fpack.(f) (subArgList castArg args)))
: !$\gSub$! (Pack argTps uargTps z T p) (Pack argTps' uargTps z' T q) fpack.
\end{rocqcode}

  \section{Detailed Explanations of Proof Automation}
  \subsection{Full Explanation of Lemmas Relating Functions with their Graph Relations}
\label{app:lemmas}

In \cref{sec:transfunc}, we described part of the translation of a
reflected function $f$: deriving the refined translation using tactics, the
unrefined graph relation, and a lemma \rocq!f__f_rel! to relate both.
On top of that, \toolname generates numerous other lemmas that are required both to
prove \rocq!f__f_rel! and for reasoning steps in automation tactics.

The lemmas generated are, in order:
\begin{enumerate}
  \item \label{functionhood}
    the \emph{functionhood lemma} \rocq!f_rel_funct! stating functionhood of
    the relation $\relnm{f}$, that is, the relation relates a single output to a set of inputs;
  \item \label{existence}
    the \emph{existence lemma} \rocq!f_rel_ex! stating that the projection
    of the result of applying the function $f$ lies in the image of its
    relation $\relnm f$;
  \item \label{rewriting}
    the theorem $\unrefproof{f}$ relating the function $f$ with its relation;
  \item the \emph{instantiating lemma} $\unrefmk{f}$ that takes terms $\seq{t_i}$ and proof terms stating that the $\seq{t_i}$ satisfy the refinements of the argument types of $f$ and constructs a terms $v$ satisfying $\relnm{f}\ \seq{t_i}\ v$; 
\end{enumerate}
Besides the these four lemmas, the translation also generates a series
of \emph{inversion lemmas} for the graph relation.
These lemmas state the equivalence of an application of $\relnm{f}$ with a different condition that is easier to prove. 
For example, the lemma for the successor case of the relation for
multiplication of Peano natural numbers states
\begin{rocqcode}
!$\forall$! (m n: Nat), (!$\exists$! v, mult_rel (Suc m) n v) !$\leftrightarrow$! (!$\exists$! v', mult_rel m n v' !$\land$! add_rel n v' v)
\end{rocqcode}
These lemmas are applied for backwards reasoning in the goal in the custom tactics, typically followed by instantiating the existential variables created by the lemmas.

The existence lemma is by far the hardest of these lemmas, for one simple reason:
it is the only lemma that actually has to reason about the function definition itself (all other lemmas are either trivial corollaries or avoid reasoning about the functions (possibly utilizing the rewriting lemma for this purpose)).

\subsection{Full Explanation of the Solver Tactic}

As mentioned already in \cref{subsubsec:oracle} the solver tactic is build on four levels of simplification and solving tactics of increasing strength.

We will now describe these simplification and goal solving tactics in more detail.

\paragraph{The \fastDone tactic}The \fastDone tactic combines several native \coq tactic to quickly solve a goal (or fail quickly).

\begin{listing}
\begin{rocqcode}
Ltac fast_done :=
solve [
  quick_wff_wit (* automatically apply the wff lemmas *)
  | unshelve eassumption
  | symmetry; unshelve eassumption
  | reflexivity (* works both for Leibniz equality and the translated equality *)
  | destruct_unit (* destruct unit terms, useful to prove they are all equal *)
  | lia (* solve integer arithmetic goals *)
  | congruence (* solve goals in theory of equality and uninterpreted functions *)
  | unshelve intuition (* simplify and apply tauto,
                        a decision procedure for constructive propositional logic *)
].
Ltac fast_solver := fast_done.
\end{rocqcode}
  \caption{Simplified definition of the \fastDone tactic. }
\end{listing}

\paragraph{The \quickSimpl and \quickSolve tactics}
\begin{listing}
    \begin{rocqcode}
Ltac quick_simpl := try quick_cleanup; repeat shape_based.
    \end{rocqcode}
    \begin{rocqcode}
Ltac !$\simplLoop$! := first [
  fast_solver
  | f_equal_ind
  | progress shape_based
  | congruence
  | progress subst
  | progress (simpl in *)
  | progress autorewrite with lia_rewrites fix_notation_hints in *
  | left; simpl_loop
  | right; simpl_loop
].
    \end{rocqcode}
    \begin{rocqcode}
Ltac !$\quickSolve$! := solve [ quick_wff_wit | quick_simpl; repeat !$\simplLoop$! ].
    \end{rocqcode}
  \caption{Simplified definition of quicksolve tactic. }
  \Description{Simplified definition of quicksolve tactic. }
  \label{fig:quicksolve}
\end{listing}
The \quickSimplCleanup tactic contains several (if-then-else related) goal rewrite steps, autounfold and autorewrites with some hint databases. In particular, it will simplify projections of casts and fully applied refined binary operators.
The \shapeBased tactic is defined as a big match (>100 branches) on the proof state doing a single simple forward reasoning step in a hypothesis or goal. It only makes reasoning steps that cannot interfere with our other automation tactics (i.e. they make proving the goal strictly easier) and they are all very simple and fast. Many of the branches simply rewrite with a number of lemmas regarding booleans, equality and boolean connectives. 
This tactic also serves as the main "plumbing" (along with a number of rewrite/unfold/resolve hints added to a number of hint databases which are applied automatically in many of the automation tactics described in this section) for the overall automation machinery bringing the proof state into a form the other automation tactics can (more easily) handle.
The \fEqualInd tactic applies the $f\_equal$ tactic to goals that are equalities between applications of the same constructor of an inductive data type. It is useful in various other automation tactics involving backwards reasoning, in particular for reasoning about the automatically derived equality operators for inductive data types. 
The \simplLoop tactic is the backbone of the \quickSolve tactic. It corresponds to a single goal solving or simplification step.
The \quickSolve tactic is then defined as (basically) running first \quickSimpl and then running \simplLoop until the goal is solved. We explicitly try solving the goal with \easyWit first, to produce faster and simpler proofs for simple well-formedness goals.

\paragraph{The \cleanup and \finish tactics}
\begin{listing}
    \begin{rocqcode}
Ltac quick_simple_cleanup_steps :=
  try unify_vars; try (specialize_hyps; try unify_vars).

Ltac initial_simple_cleanup_steps :=
  quick_simpl; 
  repeat progress autounfold with lia_unfold in *;
  simpl_proj;
  try !$\unpack$!;
  try (axiomatize_ih_specializations; quick_simpl);
  try (axiomatize_terms; quick_simpl);
  quick_simple_cleanup_steps.

Ltac inversion_cleanup :=
  (* forward reasoning *)
  try quick_simple_cleanup_steps;
  (* backwards reasoning *)
  (* backwards reasoning in hypotheses *)
  try (invert_axiomatizations; quick_simple_cleanup_steps;
       try instantiate_hyps; try specialize_hyps);
  (* backwards reasoning in goal *)
  try (autorewrite with f_rel_back; repeat instExistGoal);
  quick_simple_cleanup_steps.

Ltac cleanup_recurse :=
  quick_simpl;
  try initial_simple_cleanup_steps;
  try inversion_cleanup

Ltac cleanup := 
  quick_simpl;
  saturate_context;
  try cleanup_recurse;
  saturate_context;
  repeat cleanup_recurse.
    \end{rocqcode}
    \begin{rocqcode}
Ltac finish := solve [
  quick_wff_wit
  | first [quicksolve | repeat split; quicksolve] | timeout 5 verit ].
    \end{rocqcode}
    \begin{rocqcode}
Ltac oracle := solve [ quick_wff_wit | quicksolve | cleanup; finish].
    \end{rocqcode}
  \caption{Simplified definition of the \oracle tactic}
  \Description{Simplified definition of the “oracle” tactic}
\end{listing}
The \unpack tactic destructs context variables of type \fClass \_ or \fuClass \_.

The \axIhAppls tactic replaces the projections of the 
applications of induction hypothesis with free variables of same type and adds a witnesses of the refinement of that application of the induction hypothesis to the context. It also unifies applications of an induction hypothesis that differ only in witnesses (using proof irrelevance). This helps remove complicated proof terms from the proof state and unify applications of induction hypothesis that differ only in refinement witnesses. 

The \axTerms tactic replaces the projections of full applications of reflected functions with variables axiomatized as result of the corresponding relations (and similarly for projections of applications of functions extracted from \fClass{}s). A refinement witness for the application is also added to the context. 

The \unifyVars tactic "unifies" variables axiomatized as the results of the same applications of graph relation using the functionhood lemma of the relation. It also removes duplicated assumptions from the context and substitutes variables by equalities in the context if possible.

The \specHyps tactic specialize uniquely axiomatized antecedents in hypothesis (result variables of graph relations that are given elsewhere in the context). It also tries to specialize hypothesis with proof terms generated by the \easyWit or \quickSolve tactics.

The \invAxs tactic inverts the hypothesis axiomatizing variables using the relations, whenever this inversion results in no additional subgoals.

The \instHyps tactic creates terms for uniquely axiomatized variables in hypotheses satisfying the axiomatization.
For this, we transitively traverse the relations and create the refined term (uniquely) axiomatized as that variable, then we replace the projection of that refined term with a variable axiomatized using a graph relation. Finally we specialize the hypothesis with it that variable and with it's defining axiom.

\tac{autorewrite with f\_rel\_back} rewrites goals using the backwards reasoning lemmas. 

The \instExistGoal tactic creates (if not already present in the context) terms to specialize the free (existentially quantified) variables in the goal with, then specializes the goal with them (using the \tac{exists} tactic).
Due to the importance of instantiating existentials, this tactic might do substantial preprocessing (including \saturateCtx).
There is even a version of the tactic included in \liaPreprocessor that may use
the \oracle tactic for solving proof obligations when synthesizing terms with
which to instantiate the existentials.

The \saturateCtx tactic generates refinement witnesses for the variables axiomatized using graph relations in the context.
This requires reconstructing refined terms for result variables of graph relations (similar to what is done in the \instHyps tactic).
As a consequence, it is frequently very slow, so in the actual tactic definitions, it is only run after a full (failed) run without it. As a consequence, in the few cases where context saturation is actually necessary to prove a goal, the verification time for the goal will be quite substantial. 

\paragraph{The \solver tactic}
\begin{listing}
    \begin{rocqcode}
Ltac preprocessor :=
  repeat progress (try cleanup; try lia_preprocessor).
Ltac solver := first [
  quick_wff_wit
  | lia_preprocessor; lia
  | preprocessor; first [lia | oracle | shelve]
  ].
    \end{rocqcode}
    \begin{rocqcode}
Ltac solver := first [
  quick_wff_wit
  | quicksolve
  | quick_simpl; repeat progress first [simpl in *; cleanup | lia_preprocessor];
  first [quicksolve | solve [unshelve verit] | shelve]
].
    \end{rocqcode}
  \Description{Simplified explanation of “solver” tactic. }
  \caption{Simplified explanation of solver tactic. }
\end{listing}

The \liaPreprocessor tactic autounfolds and autorewrites with a number of lemmas to simplify integer reasoning goals to something the \tac{lia} tactic can hopefully handle. It also tries to decide which branch of a disjunction in the goal to take using the \oracle tactic to prove subgoals and it tries to specialize hypotheses with proof terms using the oracle tactic. 
The \solver tactic then improves upon the \oracle tactic merely by aditionally preprocessing using the \liaPreprocessor tactic and by shelving subgoals it cannot solve. 

\fi

\end{document}